%% file: conformal_blocks-c_equals_minus2-arXiv-v3.tex
\documentclass[a4paper,reqno,11pt]{amsart} %Loads amsthm, amsmath, amsfonts + a4paper and right-numbered equations

\usepackage[T1]{fontenc} 
\usepackage[utf8]{inputenc}
\usepackage[english]{babel} 
\usepackage{newpxtext} %Palatino

\usepackage{scalerel}

\usepackage{amssymb,textcomp,mathrsfs,stmaryrd,braket,commath,relsize}
\usepackage{bm} %To use many alphabets
	\SetSymbolFont{stmry}{bold}{U}{stmry}{m}{n} %Avoids meaningless warning for bold math symbols
\usepackage[new]{old-arrows} %For long hook-arrows
    
\usepackage{mathtools}
\usepackage{thm-restate}

\theoremstyle{plain} %Math enviroments I
    \newtheorem{theorem}{Theorem}[section]
    \newtheorem*{theorem*}{Theorem}
	
	\newtheorem{citedtheorem}{Theorem}[section]

    \newtheorem{proposition}[theorem]{Proposition}
    \newtheorem*{proposition*}{Proposition}
    
	\newtheorem{corollary}[theorem]{Corollary}
    \newtheorem*{corollary*}{Corollary}
    
	\newtheorem{lemma}[theorem]{Lemma}
    \newtheorem*{lemma*}{Lemma}

    \newtheorem*{conjecture*}{Conjecture}
    
\theoremstyle{definition} %Math enviroments II
    \newtheorem{definition}[theorem]{Definition}
    \newtheorem*{definition*}{Definition}

    \newtheorem*{notation*}{Notation}

\theoremstyle{remark} %Math enviroments III
    \newtheorem{remark}[theorem]{Remark}
    \newtheorem*{remark*}{Remark}
	
	\newtheorem{example}[theorem]{Example}
	\newtheorem*{example*}{Example}	
	
\numberwithin{equation}{section}
\numberwithin{figure}{section}

\usepackage{etoolbox} 
    \newcommand{\addQEDstyle}[2]{\AtBeginEnvironment{#1}{\pushQED{\qed}\renewcommand{\qedsymbol}{#2}}
    \AtEndEnvironment{#1}{\popQED}} %Symbol at the end of environment: call as \addQEDstyle{environmentname}{symbolname}
    \addQEDstyle{remark}{$\circ$}
    \addQEDstyle{remark*}{$\circ$} %Triangles at the end of remarks
	\addQEDstyle{example}{$\circ$}
    \addQEDstyle{example*}{$\circ$} %Triangles at the end of examples
	\addQEDstyle{definition}{$\circ$}
    \addQEDstyle{definition*}{$\circ$} %Triangles at the end of definitions

\usepackage{tikz}
	\usetikzlibrary{cd} %Commutative diagrams
	
\usepackage{subcaption}

\AtBeginDocument{\def\MR#1{}} %No MR number in bibliography
\apptocmd{\sloppy}{\hbadness 10000\relax}{}{} %For badboxes in .bbl (uses etoolbox)

\makeatletter %Print 2020AMSsubjclass (instead of 2010)
	\@namedef{subjclassname@2020}{
		\textup{2020} Mathematics Subject Classification}
\makeatother

\usepackage{ytableau}
\usepackage[percent]{overpic}
\usepackage{enumitem}
\usepackage{rotating}

\usepackage{multirow}
\usepackage{tabularx,stackengine,collcell}
\let\endminwd\relax
\newcolumntype{L}[1]{>{\collectcell\xminwd l{#1}}l<{\endminwd\endcollectcell}}
\newcolumntype{C}[1]{>{\collectcell\xminwd c{#1}}c<{\endminwd\endcollectcell}}
\newcolumntype{R}[1]{>{\collectcell\xminwd r{#1}}r<{\endminwd\endcollectcell}}
\def\minwd#1#2#3\endminwd{\stackengine{0pt}{#3}{\rule{#2}{0pt}}{O}{#1}{F}{F}{L}}
\newcommand\xminwd[1]{\minwd#1}
\allowdisplaybreaks

\usepackage[colorlinks,psdextra,pdfencoding=auto]{hyperref}
\hypersetup{colorlinks,linkcolor=red,filecolor=green,urlcolor=blue,citecolor=blue} 

\usepackage[sort=use]{glossaries}

\newglossary[slg]{symbolslist}{syi}{syg}{List of used notation}
\makeglossaries
 
\input{tex/macros}
\input{tex/list_of_symbols}

\renewcommand*{\glossarymark}[1]{}

\begin{document}

%\input{tex/macros}

%%%%%%%%%%%%%%%%%%%%%%%%%%%

\title{Planar UST Branches and $c=-2$ Degenerate Boundary Correlations}

\vspace{2.5cm}

\begin{center}
{%\LARGE 
\huge
\bf \scshape{
Planar UST Branches and $c=-2$ \\[.5em] Degenerate Boundary Correlations
}}
\end{center}

\vspace{0.75cm}

\begin{center}
{\Large \scshape Alex Karrila}{\footnotesize\footnotemark[1]} \\
{\footnotesize{\protect\url{alex.karrila@abo.fi}}}\\
\bigskip{}
{\Large \scshape Augustin Lafay}{\footnotesize\footnotemark[2]} \\
{\footnotesize{\protect\url{augustin.lafay@aalto.fi}}}\\
\bigskip{}
{\Large \scshape Eveliina Peltola}{\footnotesize\footnotemark[2]\textsuperscript{\&}\footnotemark[3]} \\
{\footnotesize{\protect\url{eveliina.peltola@aalto.fi}}}\\
\bigskip{}
{\Large \scshape Julien Roussillon}{\footnotesize\footnotemark[2]} \\
{\footnotesize{\protect\url{julien.roussillon@aalto.fi}}} 
\end{center}

\footnotetext[1]{{\r{A}}bo Akademi Matematik, Henriksgatan 2, FI-20500, {\r{A}}bo, Finland.}
\footnotetext[2]{Department of Mathematics and Systems Analysis, 
P.O. Box 11100, FI-00076, Aalto University, Finland.}
\footnotetext[3]{Institute for Applied Mathematics, University of Bonn, Endenicher Allee 60, D-53115 Bonn, Germany.}
	
\setcounter{footnote}{0}

\vspace{0.75cm}

\begin{center}
\begin{minipage}{0.85\textwidth} 
{\scshape Abstract.}
We provide a conformal field theory (CFT) description of
the probabilistic model of boundary effects in the wired uniform spanning tree (UST) and its algebraic content, 
concerning the entire first row of the Kac table with central charge $c=-2$.
Namely, we prove that all boundary-to-boundary connection probabilities for (potentially fused) branches in the wired UST converge in the scaling limit to explicit CFT quantities, expressed in terms of determinants, 
which can also be viewed as conformal blocks of degenerate primary fields in a boundary CFT with central charge $c=-2$. 

\smallskip

Moreover, we verify that the Belavin-Polyakov-Zamolodchikov (BPZ) PDEs (i.e.,~Virasoro degeneracies) of arbitrary orders hold, 
and we also reveal an underlying valenced Temperley-Lieb algebra action on the space of boundary correlation functions of primary fields in this model. 
To prove these results, we combine probabilistic techniques with representation theory.
\end{minipage}
\end{center}

%\vspace{0.75cm}
\newpage

{\hypersetup{linkcolor=black}
\setcounter{tocdepth}{2}
\tableofcontents}

\newpage
\bigskip{}
\section{Introduction}
\input{tex/intro.tex}

\bigskip{}
\section{Discrete connection probabilities and Fomin type formulas}
\label{sec:discrete}
\input{tex/sec2-Fomin.tex}

\bigskip{}
\section{Scaling limit results}
\label{sec:scaling_limit}
\input{tex/sec3-sclim.tex}

\bigskip{}
\section{CFT properties of pure partition functions}
\label{sec:blocks}
\input{tex/sec4-CFT.tex}

\bigskip{}
\section{Algebraic properties of pure partition functions}
\label{sec:algebra}
\input{tex/sec5-alg.tex}

\bigskip{}
\section{Asymptotic properties of pure partition functions}
\label{sec:ASY}
\input{tex/sec6-asy.tex}

%\newpage

\appendix

\bigskip{}
\section{Evaluation of two determinants involving factorials}
\label{app:det}
\input{tex/app-det.tex}

\bigskip{}
\section{A completely explicit solution basis}
\label{app:coblo_det}
\input{tex/app-coblo.tex}

\bigskip{}
\section{Examples}
\label{app:explicit determinants}
\input{tex/app-ex.tex}

\bigskip{}
\section{Two SLE applications}
\label{app:fused SLE}
\input{tex/app-fused.tex}

\newcommand{\changeurlcolor}[1]{\hypersetup{urlcolor=#1}} 
\changeurlcolor{black}

\bibliographystyle{alpha}

\newpage

%Print list of symbols
\printglossary[type=symbolslist]

\end{document}

%% file: tex/macros.tex
\global\long\def\bR{\mathbb{R}}

\global\long\def\bZ{\mathbb{Z}}

\global\long\def\bZpos{\mathbb{Z}_{>0}}

\global\long\def\bZnn{\mathbb{Z}_{\geq 0}}
\global\long\def\bQ{\mathbb{Q}}
\global\long\def\bC{\mathbb{C}}
\global\long\def\bH{\mathbb{H}}
\global\long\def\bD{\mathbb{D}}

 \global\long\def\sD{\mathcal{D}}
 \global\long\def\sG{\mathcal{G}}
 
 \global\long\def\sL{\mathcal{L}}
 
 \global\long\def\sE{\mathcal{E}}

 \global\long\def\sV{\mathcal{V}}

\global\long\def\ii{\mathfrak{i}}

\newcommand{\LP}{\mathsf{LP}}
\newcommand{\TL}{\mathsf{TL}}
\newcommand{\vTL}{\mathsf{TL}}

\newcommand{\Summed}{n}

\newcommand{\multii}{\varsigma}

\newcommand{\NB}{\textnormal{NB}^\lambda}
\newcommand{\SYT}{\textnormal{SYT}^\lambda}

\newcommand{\NBof}[1]{{\textnormal{NB}^{#1}}}
\newcommand{\RSYTof}[1]{{\textnormal{RSYT}^{#1}_\multii}}

\newcommand{\idpt}{\mathrm{p}_\multii}%{p_\multii}
\global\long\def\PartF{\mathcal{Z}}
\global\long\def\CobloF{\mathcal{U}}
\global\long\def\chamber{\mathfrak{X}}

\newcommand{\im}{\text{\upshape Im}}

\global\long\def\SLE{\textrm{SLE}}

\global\long\def\PR{\mathbb{P}}

\global\long\def\pder#1{\frac{\partial}{\partial#1}}
\global\long\def\pdder#1{\frac{\partial^{2}}{\partial#1^{2}}}

\newcommand{\sign}{\mathrm{sgn}}

\newcommand{\GreenK}{\mathsf{G}}
\newcommand{\ExcK}{\mathsf{K}}
\newcommand{\Walks}{\mathscr{W}}
\newcommand{\walk}{\omega}

\newcommand{\GreenKH}{\mathscr{G}}
\newcommand{\PoissonKH}{\mathscr{P}}
\newcommand{\ExcKH}{\mathscr{K}}

\newcommand{\KWleq}{\stackrel{\scriptsize{()}}{\leftarrow}}
\newcommand{\DPgeq}{\succeq} % partial order of DPs
\newcommand{\CItilingsof}{\mathcal{C}}
\newcommand{\diff}{\mathrm{d}} % stochastic differential

\global\long\def\Mob{\phi}
\global\long\def\domain{\Lambda}

\global\long\def\summ{q}%{p}%{\sigma}
\global\long\def\mult{\mathrm{m}}
\newcommand{\conn}{\vartheta_{\mathrm{UST}}}

\newcommand{\event}{\mathrm{Conn}}

\global\long\def\linkInEquation#1#2{\underset{\;#1\;\;#2}{
\vcenter{\hbox{\includegraphics[scale=0.5]{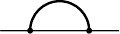}}}}}

\global\long\def\link#1#2{\raisebox{.5ex}{\hspace{-1mm}
\scalebox{.75}{$\linkInEquation{\boldsymbol{#1}}{\boldsymbol{#2}}$}}}

\global\long\def\quote#1{$``${{#1}}$"$}

\global\long\def\SolSp{\mathcal{S}}%{\mathcal{C}}

\global\long\def\fugacity{\nu}
\newcommand{\Vir}{\mathrm{Vir}}
\newcommand{\SymGrp}{\mathfrak{S}}
\newcommand{\Rows}{\mathfrak{R}^\lambda}%{\textnormal{R}^\lambda}
\newcommand{\Columns}{\mathfrak{C}^\lambda}%{\textnormal{C}^\lambda}

\newcommand{\super}[1]{^{\scaleobj{0.85}{(#1)}}}
\newcommand{\sub}[1]{_{\scaleobj{0.85}{(#1)}}}

\newcommand{\bs}[1]{{\boldsymbol{#1}}}

\global\long\def\nested{\boldsymbol{\underline{\Cap}}}

\newcommand{\rainbow}[1]{{\nested_{#1}}} 

\newcommand{\BSAOP}{\mathcal{P}}%{\Delta}

\global\long\def\one{\scalebox{0.9}{\textnormal{1}} \hspace*{-.75mm} \scalebox{0.6}{\raisebox{.3em}{\bf |}} }
\global\long\def\onesmall{\scalebox{0.7}{\textnormal{1}} \hspace*{-.75mm} \scalebox{0.5}{\raisebox{.3em}{\bf |}} }

%% file: tex/list_of_symbols.tex
\newglossaryentry{symb:sG}{
name=${\sG=(\sV,\sE)}$,
description={a finite connected planar graph with a fixed embedding.},
sort=symbolsG, type=symbolslist
}
\newglossaryentry{symb:sVcirc}{
name=$\sV_\circ$,
description={the set of interior vertices of a graph $\sG=(\sV,\sE)$ of the above type.},
sort=symbolsVcirc, type=symbolslist
}
\newglossaryentry{symb:sVpartial}{
name=$\sV_\partial$,
description={the set of boundary vertices of a graph $\sG=(\sV,\sE)$ of the above type.},
sort=symbolsVpartial, type=symbolslist
}
\newglossaryentry{symb:event}{
name=$\event$,
description={the event (assumed throughout) that the UST boundary branches from the interior vertices of the odd-index edges connect to the boundary via the even-index edges, each using a different even edge (Fig.~\ref{fig:UST combinatorial}).},
sort=symbolevent, type=symbolslist
}
\newglossaryentry{symb:conn}{
name=$\conn$,
description={the random partition (planar link pattern) of the indices into pairs on the event $\event$, encoding the pairing of the UST boundary-to-boundary branches (Fig.~\ref{fig:UST combinatorial}).},
sort=symbolconn, type=symbolslist
}
\newglossaryentry{symb:LP}{
name=$\LP_N$,
description={the set of planar pair partitions, a.k.a.~(planar) $N$-link patterns.},
sort=symbolLP, type=symbolslist
}
\newglossaryentry{symb:connevent}{
name=${\event(\alpha)}$,
description={the crossing event that the random UST connectivity $\conn$ equals $\alpha$.},
sort=symbolconnevent, type=symbolslist
}
\newglossaryentry{symb:domain}{
name=${\domain^\delta \subsetneq \bC}$,
description={a bounded simply connected domain, whose boundary consists of a simple loop on the rescaled square lattice $\delta\bZ^2$; identified with the planar graph with vertices $\sV^\delta := \sV(\delta\bZ^2) \cap \smash{\overline{\domain}}^\delta$ and edges $\sE^\delta := \sE(\delta\bZ^2) \cap \smash{\overline{\domain}}^\delta$.},
sort=symboldomain, type=symbolslist
}
\newglossaryentry{symb:multii}{
name=${\multii = (s_1,\ldots,s_d)}$,
description={the tuple of fused multiplicities, a.k.a.~valences, upon fusing the endpoints of $N = \frac12(s_1+\cdots+s_d)$ curves onto $d$ limit points.},
sort=symbolmultii, type=symbolslist
}
\newglossaryentry{symb:summ}{
name=${\summ_j := s_0 + \cdots + s_j }$,
description={the partial sum of the fused multiplicities, for $0 \leq j \leq d$.},
sort=symbolsumm, type=symbolslist
}
\newglossaryentry{symb:LPmultii}{
name=$\LP_\multii$,
description={the set of (planar) $\multii$-valenced link patterns.},
sort=symbolLPmultii, type=symbolslist
}
\newglossaryentry{symb:imath}{
name=$\imath$,
description={the unfusing map sending a valenced link pattern $\alpha$ in $\LP_\multii$ to a link pattern $\imath(\alpha)$ in $\LP_N$ (Fig.~\ref{fig: UST fused}).},
sort=symbolimath, type=symbolslist
}
\newglossaryentry{symb:polygon}{
name=${(\domain; p_1, \ldots, p_d)}$,
description={an admissible $d$-polygon.},
sort=symbolpolygon, type=symbolslist
}
\newglossaryentry{symb:discrpolygon}{
name=${(\domain^\delta; e_1^\delta, \ldots, e_{2N}^\delta)}$,
description={an (discrete) admissible $2N$-polygon.},
sort=symboldiscrpolygon, type=symbolslist
}
\newglossaryentry{symb:BEK}{
name=${\ExcKH_\domain(x, y)}$,
description={the Brownian excursion kernel in $\domain$ from $x$ to $y$.},
sort=symbolBEK, type=symbolslist
}
\newglossaryentry{symb:DPgeq}{
name=$\DPgeq$,
description={the partial order relation on $\LP_N$ defined in~\cite[Definition~2.1]{KKP:Boundary_correlations_in_planar_LERW_and_UST}.},
sort=symbolDPgeq, type=symbolslist
}
\newglossaryentry{symb:PartF}{
name=$\PartF_\alpha$,
description={the pure partition function associated to $\alpha \in \LP_\multii$.},
sort=symbolPartF, type=symbolslist
}
\newglossaryentry{symb:chamber}{
name=$\chamber_{d}$,
description={the configuration space $\chamber_{d}:=\{(x_1, \dots, x_d) \in \bR^d \; | \; x_1<\cdots<x_d\}$.},
sort=symbolchamber, type=symbolslist
}
\newglossaryentry{symb:BPZop}{
name=${\sD_{s_j+1}^{{(x_j)}}}$,
description={the BPZ differential operator defined in Eq.~\eqref{eq: BPZ operator at kappa equals 2}.},
sort=symbolBPZop, type=symbolslist
}
\newglossaryentry{symb:Virop}{
name=${\sL_{-m}^{{(x_j)}}}$,
description={the first order differential operators related to Virasoro generators.},
sort=symbolVirop, type=symbolslist
}
\newglossaryentry{symb:SolSp}{
name=${\SolSp_\multii}$,
description={the solution space of BPZ PDEs~\eqref{eq: BPZ PDE at kappa equals 2} with global M\"obius covariance (COV~\eqref{eq: COV general at kappa equals 2}).},
sort=symbolSolSp, type=symbolslist
}
\newglossaryentry{symb:TL}{
name=${\TL_{2N}(\fugacity)}$,
description={Temperley-Lieb algebra with fugacity parameter $\fugacity$.},
sort=symbolTL, type=symbolslist
}
\newglossaryentry{symb:vTL}{
name=${\vTL_\multii}$,
description={valenced Temperley-Lieb algebra with fugacity parameter specialized to $\fugacity = -2$.},
sort=symbolvTL, type=symbolslist
}
\newglossaryentry{symb:kernel}{
name=$\mathfrak{K}$,
description={a kernel used to define Fomin type determinants.},
sort=symbolkernel, type=symbolslist
}
\newglossaryentry{symb:FomDet}{
name=${\Delta^{\mathfrak{K}}_\beta}$,
description={a Fomin type determinant with kernel $\mathfrak{K}$ and rows/columns indexed by links of $\beta$.},
sort=symbolFomDet, type=symbolslist
}
\newglossaryentry{symb:CItilings}{
name=${\# \CItilingsof (\alpha / \beta)}$,
description={the number of cover-inclusive Dyck tilings of the skew Young diagram~$\alpha/\beta$ for $\alpha, \beta \in \LP_N$, defined in~\cite[Definition~2.8]{KKP:Boundary_correlations_in_planar_LERW_and_UST}.},
sort=symbolCItilings, type=symbolslist
}
\newglossaryentry{symb:Fomintypesum}{
name=${\mathfrak{Z}^{\mathfrak{K}}_\alpha}$,
description={an inverse Fomin type sum  associated to $\alpha$, with kernel $\mathfrak{K}$.},
sort=symbolFomintypesum, type=symbolslist
}
\newglossaryentry{symb:exchangker}{
name=${\mathfrak{K}^{(k, \ell)}}$,
description={the symmetric kernel obtained by exchanging the kernel entries at $k$ and $\ell$.},
sort=symbolexchangker, type=symbolslist
}
\newglossaryentry{symb:RWweight}{
name=${\mathsf{w} (\walk)}$,
description={the random-walk weight of a nearest-neighbor walk $\walk$.},
sort=symbolRWweight, type=symbolslist
}
\newglossaryentry{symb:RWGF}{
name=${\GreenK_{\sG}}$,
description={the random-walk Green's function with Dirichlet b.c.},
sort=symbolRWGF, type=symbolslist
}
\newglossaryentry{symb:RWEK}{
name=${\ExcK_{\sG}}$,
description={the random-walk excursion kernel with Dirichlet b.c.},
sort=symbolRWEK, type=symbolslist
}
\newglossaryentry{symb:tan}{
name=$D^{\mathrm{tan}}$,
description={the discrete counterclockwise tangential difference on a boundary edge.},
sort=symboltan, type=symbolslist
}
\newglossaryentry{symb:GreenKH}{
name=$\GreenKH$,
description={the Green's function.},
sort=symbolGreenKH, type=symbolslist
}
\newglossaryentry{symb:PoissonKH}{
name=$\PoissonKH$,
description={the Poisson kernel.},
sort=symbolPoissonKH, type=symbolslist
}
\newglossaryentry{symb:ExcKH}{
name=$\ExcKH$,
description={the Brownian excursion kernel.},
sort=symbolExcKH, type=symbolslist
}
\newglossaryentry{symb:Vir}{
name=$\Vir$,
description={the Virasoro algebra.},
sort=symbolVir, type=symbolslist
}
\newglossaryentry{symb:VirLn}{
name=$L_n$,
description={a Virasoro generator.},
sort=symbolVirLn, type=symbolslist
}
\newglossaryentry{symb:VirC}{
name=$C$,
description={the central element in the Virasoro algebra.},
sort=symbolVirC, type=symbolslist
}
\newglossaryentry{symb:VirEnv}{
name=$\mathcal U(\Vir^-)$,
description={the enveloping algebra of the Virasoro generatores with negative index..},
sort=symbolVirEnv, type=symbolslist
}
\newglossaryentry{symb:Virhwv}{
name=$v_{h}^{c}$,
description={a highest-weight vector for the Virasoro algebra of weight $h$ and central charge $c$.},
sort=symbolVirhwv, type=symbolslist
}
\newglossaryentry{symb:VirVerma}{
name=$M_{h}^{c}$,
description={a Verma module for the Virasoro algebra of weight $h$ and central charge $c$.},
sort=symbolVirVerma, type=symbolslist
}
\newglossaryentry{symb:VirVerma2}{
name=$M_{h}$,
description={the Verma module for the Virasoro algebra of weight $h$ and central charge $c=-2$.},
sort=symbolVirVerma2, type=symbolslist
}
\newglossaryentry{symb:Kac}{
name=$h_{\ell}$,
description={a Kac type conformal weight at level $\ell \geq 1$ with central charge $c=-2$.},
sort=symbolKac, type=symbolslist
}
\newglossaryentry{symb:Virhwv2}{
name=$v_{\ell}$,
description={the highest-weight vector of weight $h_{\ell}$ at level $\ell \geq 1$ with central charge $c=-2$.},
sort=symbolVirhwv2, type=symbolslist
}
\newglossaryentry{symb:BSAOP}{
name=$\BSAOP_\ell$,
description={the polynomial in the negative Virasoro generators which gives the singular vector at level $\ell$, as $w_\ell = \BSAOP_\ell v_{\ell}$, found in~\cite{BSA:Degenerate_CFTs_and_explicit_expressions_for_some_null_vectors}.},
sort=symbolBSAOP, type=symbolslist
}
\newglossaryentry{symb:VOAF}{
name=$V_{\alpha,h}$,
description={a space of formal series in the indeterminate $t$ with finitely many negative terms. Note that it is a $\Vir$-module with zero central charge $c=0$.},
sort=symbolVOAF, type=symbolslist
}
\newglossaryentry{symb:VOAFAct}{
name=$L_n^0$,
description={the action of the Virasoro generators on the $\Vir$-module $V_{\alpha,h}$.},
sort=symbolVOAFAct, type=symbolslist
}
\newglossaryentry{symb:VOAFW}{
name=${W\otimes V_{\alpha,h}}$,
description={a space of formal series in the indeterminate $t$ with finitely many negative terms and coefficients in a $\Vir$-module $W$.},
sort=symbolVOAFW, type=symbolslist
}
\newglossaryentry{symb:VOAFWAct}{
name=$\hat{L}_n$,
description={the action of the Virasoro generators on the $\Vir$-module $W\otimes V_{\alpha,h}$.},
sort=symbolVOAFWAct, type=symbolslist
}
\newglossaryentry{symb:BSAOPhat}{
name=$\hat{\BSAOP}_\ell$,
description={the polynomial $\BSAOP_\ell$ with the substitutions $L_n \mapsto \hat L_n$.},
sort=symbolBSAOPhat, type=symbolslist
}
\newglossaryentry{symb:SolSpN}{
name=$\SolSp_N$,
description={the solution space of BPZ PDEs~\eqref{eq: BPZ PDE at kappa equals 2}~\&~COV~\eqref{eq: COV general at kappa equals 2} with unit valences $\multii = (1^{2N})$.},
sort=symbolSolSpN, type=symbolslist
}
\newglossaryentry{symb:CSymGrp}{
name=$\bC[\SymGrp_{\Summed}]$,
description={the group algebra of the symmetric group $\SymGrp_{\Summed}$.},
sort=symbolCSymGrp, type=symbolslist
}
\newglossaryentry{symb:SymGrp}{
name=$\SymGrp_{\Summed}$,
description={the symmetric group over the set $\{1,2,\ldots,\Summed\}$.},
sort=symbolSymGrp, type=symbolslist
}
\newglossaryentry{symb:partition}{
name=$\lambda$,
description={a partition $\lambda=(\lambda_1,\lambda_2,\ldots ,\lambda_l)$, so $\lambda_1 \geq \lambda_2 \geq\cdots \geq\lambda_l \geq 0$ and $\lambda_1+\lambda_2+\cdots +\lambda_l = \Summed$.},
sort=symbolpartition, type=symbolslist
}
\newglossaryentry{symb:transposition}{
name=$\tau_i$,
description={the transposition $(i,i+1) \in \SymGrp_\Summed$.},
sort=symboltransposition, type=symbolslist
}
\newglossaryentry{symb:NB}{
name=$\NB$,
description={the set of numberings of shape $\lambda$.},
sort=symbolNB, type=symbolslist
}
\newglossaryentry{symb:SYT}{
name=$\SYT$,
description={the set of standard Young tableaux of shape $\lambda$.},
sort=symbolSYT, type=symbolslist
}
\newglossaryentry{symb:Rows}{
name=$\Rows(A)$,
description={the subgroup of $\SymGrp_\Summed$ preserving the set of entries of each row of $A \in \NB$.},
sort=symbolRows, type=symbolslist
}
\newglossaryentry{symb:Columns}{
name=$\Columns(A)$,
description={the subgroup of $\SymGrp_\Summed$ preserving the set of entries of each column of $A \in \NB$.},
sort=symbolColumns, type=symbolslist
}
\newglossaryentry{symb:tabloid}{
name=$\{A\}$,
description={a tabloid, i.e., an equivalence class of numberings defined by $\{A'\} = \{A\}$ if and only if $A'=\sigma.A$ for some $\sigma \in \Rows(A)$.},
sort=symboltabloid, type=symbolslist
}
\newglossaryentry{symb:Mlambda}{
name=$M^\lambda$,
description={the vector space spanned by tabloids of shape $\lambda$. Note that it is a $\SymGrp_\Summed$-module.},
sort=symbolMlambda, type=symbolslist
}
\newglossaryentry{symb:antisym}{
name=$\epsilon_A$,
description={the column antisymmetrizer for the numbering $A \in \NB$.},
sort=symbolantisym, type=symbolslist
}
\newglossaryentry{symb:polytabloid}{
name=$v_A$,
description={the polytabloid for the numbering $A \in \NB$.},
sort=symbolpolytabloid, type=symbolslist
}
\newglossaryentry{symb:Vlambda}{
name=$V^\lambda$,
description={the vector space spanned by polytabloids of shape $\lambda$, which is a simple $\SymGrp_\Summed$-module.},
sort=symbolVlambda, type=symbolslist
}
\newglossaryentry{symb:pin}{
name=${\pi=\pi_N}$,
description={the rectangular Young diagram with $2N$ boxes and $N$ rows.},
sort=symbolpin, type=symbolslist
}
\newglossaryentry{symb:JW}{
name=${\textnormal{JW}_{i,j}}$,
description={a Jones-Wenzl idempotent.},
sort=symbolJW, type=symbolslist
}
\newglossaryentry{symb:clSYM}{
name=${\SymGrp_{s_1} \times \cdots \times \SymGrp_{s_d}}$,
description={a colored symmetric group.},
sort=symbolclSYM, type=symbolslist
}
\newglossaryentry{symb:idpt}{
name=$\idpt$,
description={an idempotent given by a product of Jones-Wenzl idempotents of fugacity $\fugacity = -2$.},
sort=symbolidpt, type=symbolslist
}
\newglossaryentry{symb:RSYT}{
name=${\RSYTof{\pi}}$,
description={ the set of row-strict Young tableaux of shape $\pi$ and content $\multii$.},
sort=symbolRSYT, type=symbolslist
}
\newglossaryentry{symb:Dmultii}{
name=${\mathfrak{D}_\multii}$,
description={the set $\big\{ (x_1,\ldots, x_{2N}) \in \bC^{2N} \; | \;  
x_{\summ_{k-1}+1}=x_{\summ_{k-1}+2}=\cdots =x_{\summ_{k}}\textnormal{ for all } k \in \{ 1, \ldots, d\} \big\}
$.},
sort=symbolDmultii, type=symbolslist
}
\newglossaryentry{symb:feval}{
name=${[f]_\textnormal{eval}}$,
description={the function $\bC^d \to \bC$ obtained from 
$f(x_1,\ldots,x_{2N})$ by the evaluations of variables (projection) 
$x_{\summ_{k-1}+1}=x_{\summ_{k-1}+2}=\cdots =x_{\summ_{k}}$ for all $k \in \{ 1, \ldots,d\}$.},
sort=symbolfeval, type=symbolslist
}
\newglossaryentry{symb:mult}{
name=$\mult_{j, j+1}$,
description={the multiplicity of the link $\link{j}{j+1}=\{ j, j+1\}$ in $\alpha$.},
sort=symbolmult, type=symbolslist
}
\newglossaryentry{symb:const}{
name=${C(s_j, s_{j+1}, \mult_{j,j+1})}$,
description={the structure constants in the asymptotic properties~\eqref{eq: asymptotic properties} of $\PartF_\alpha$.},
sort=symbolconst, type=symbolslist
}
\newglossaryentry{symb:pow}{
name=${\Delta^{s_{j},s_{j+1}}_m}$,
description={the fusion powers $\Delta^{s_{j},s_{j+1}}_m = h_{1,s_j+s_{j+1}+1-2m} - h_{1,s_j+1} - h_{1,s_{j+1}+1}$.},
sort=symbolpow, type=symbolslist
}
\newglossaryentry{symb:rainbow}{
name=${\rainbow{N}}$,
description={the rainbow $N$-link pattern $\rainbow{N} := \{ \{ 1, 2N\}, \{ 2, 2N -1 \}, \ldots, \{ N, N+1\} \}$.},
sort=symbolrainbow, type=symbolslist
}
\newglossaryentry{symb:CobloF}{
name=$\CobloF_\alpha$,
description={the partition function defined in Equation~\eqref{eq:fused conformal block function}.},
sort=symbolCobloF, type=symbolslist
}
\newglossaryentry{symb:CobloFdisc}{
name=${\mathfrak{U}^{\mathfrak{K}}_\alpha}$,
description={the linear combination of determinants defined in Equation~\eqref{eq: define Ufrak}.},
sort=symbolCobloFdisc, type=symbolslist
}

%% file: tex/intro.tex
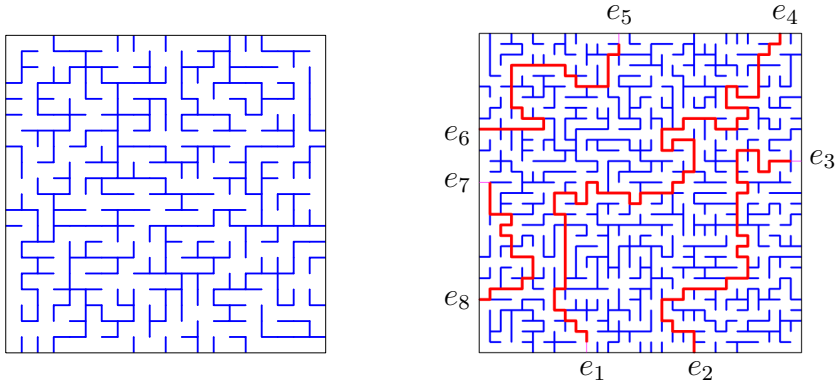
\begin{figure}
\includegraphics[width=0.3\textwidth]{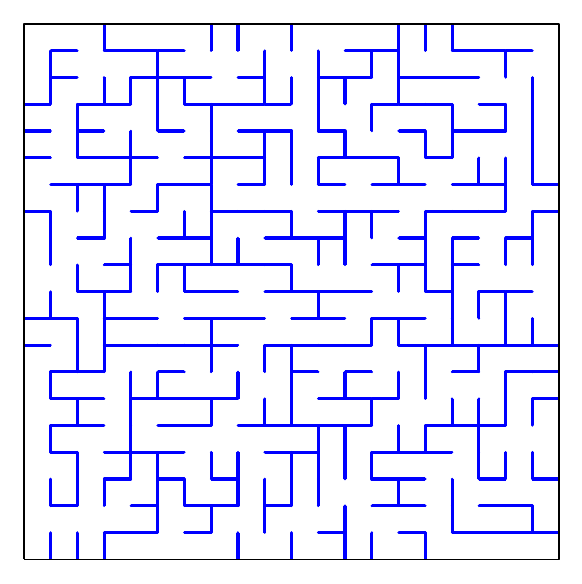}
\qquad \qquad
\raisebox{0.005\textwidth}{
\begin{overpic}[width=0.2918\textwidth, angle=-90, origin = c]{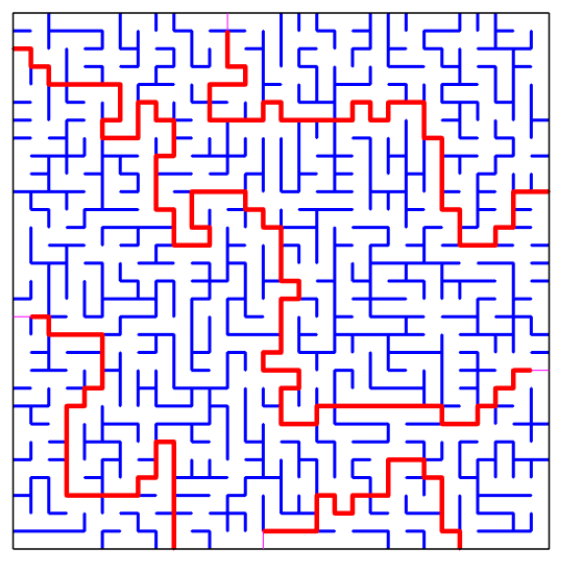}
 \put (-8,17) {\large $e_8$}
 \put (-8,52) {\large $e_7$}
 \put (-8,66) {\large $e_6$}
 \put (40,102) {\large $e_5$}
 \put (89,102) {\large $e_4$}
 \put (100,59) {\large $e_3$}
 \put (32,-4) {\large $e_1$}
 \put (64,-4) {\large $e_2$}
\end{overpic}
}
\caption{
\label{fig:UST combinatorial}
Two simulations of the wired uniform spanning tree (UST) on small graphs, illustrating the central combinatorial notions to this work. 
The~left panel depicts a wired uniform spanning tree on a grid graph of $20 \times 20$ squares. 
The~right panel depicts the same model on $30 \times 30$ squares, conditional on the event $\event$ with the marked boundary edges $e_1, \ldots, e_8$, 
where the related boundary-to-boundary UST branches are highlighted in (shades of) red. 
In~this case, the random connectivity of the interfaces of interest is $\conn = \{ \{ 1, 4 \}, \{ 2, 3\}, \{ 5, 6 \}, \{ 7, 8 \} \}$.
}
\end{figure}

Since the seminal work of Polyakov~\cite{Polyakov:Conformal_symmetry_of_critical_fluctuations} 
and in particular of Belavin, Polyakov, and Zamolodchikov (BPZ)~\cite{BPZ:Infinite_conformal_symmetry_in_2D_QFT}, conformal field theory (CFT) has made its appearance in a plethora of areas in both mathematics and physics. 
One striking application is that CFTs are generally expected to describe  scaling limits of critical lattice models\footnote{By a critical model, we refer to one with a second order (continuous) phase transition in Landau's classification. Polyakov's work~\cite{Polyakov:Conformal_symmetry_of_critical_fluctuations} is based on earlier works by Kadanoff, Patashinskii~\&~Pokrovskii, and Migdal, for instance, which crucially rely on the \quote{similarity hypothesis} that also underlies the celebrated renormalization group theory \`a la Wilson (see, e.g.,~\cite{Patashinskii-Pokrovskii:The_RG_method_in_the_theory_of_phase_transitions} for a physics review focusing on the mathematical methods).}. 
On the one hand, this statement is very strong in two dimensions, as (infinitesimal) conformal symmetries, 
encoded in the Virasoro algebra,
yield infinitely many conserved quantities which constrain the system enormously, providing information \quote{for free.}  
On the other hand, it largely remains mathematically mysterious, as even the construction (let alone the statistical physics interpretation) of the precise CFTs that are expected to describe many prominent lattice models remains unsolved. (See~\cite{CHI:Correlations_of_primary_fields_in_the_critical_planar_Ising_model} for the state-of-the art on the recent mathematical progress for the critical Ising model, 
and~\cite{GKR:Compactified_imaginary_Liouville} and references therein for a proposal that should also be relevant to the case of present interest.)
Consequently, it has been a vigorous quest in the past decades to verify that scaling limits of critical lattice models in two dimensions are indeed conformally invariant, and to reveal the CFT content therein. (E.g.~\cite{Schramm:ICM,
DCS:Conformal_invariance_of_lattice_models, Peltola:Towards_CFT_for_SLEs} 
provide reviews and further references on this topic.)

The purpose of this work is to investigate a general class of non-local crossing events in the uniform spanning tree (UST). 
More precisely, we construct all boundary-to-boundary connection probabilities for branches in wired UST, for any number of fused curves and for any valenced link pattern (as illustrated in Figure~\ref{fig: UST fused}, and formalized below). 
We also prove that these probabilities converge in the scaling limit to explicit determinantal quantities  (Section~\ref{sec:scaling_limit}), 
which we relate to appropriate boundary correlation functions in CFT with central charge $c=-2$. 
In particular, the limits satisfy a conformal covariance property, 
as well as BPZ PDEs 
of arbitrary orders appearing in the first row of the Kac table (Section~\ref{sec:blocks}).
In addition to the Virasoro algebra action, highlighted by the BPZ PDEs, 
we also reveal an underlying action of the valenced Temperley-Lieb algebra with fugacity~$-2$ on the correlation functions
(\quote{Schur-Weyl dual} to an $\mathfrak{sl}_2$ action~\cite[Appendix~C]{Flores-Peltola:Higher_spin_QSW}).

To our knowledge, the present work constitutes the first scaling limit result concerning the entire first row of the Kac table. 
Its first non-trivial element $h_{1,2}$, corresponding to second-order PDEs, has a well-established probabilistic interpretation --- with the same PDEs arising via It\^{o} calculus for SLE curves. For PDEs of higher order, only few rigorous relations to probabilistic lattice models are known, in special cases~\cite{Dubedat:Excursion_decomposition_for_SLE_and_Watts_crossing_formula, Sheffield-Wilson:Schramms_proof_of_Watts_formula, KKP:Boundary_correlations_in_planar_LERW_and_UST, Liu-Wu:Scaling_limits_of_crossing_probabilities_in_metric_graph_GFF,
FPW:Crossing_probabilities_of_critical_percolation_interfaces}. 
Hence, arguably, we provide one of the widest known dictionaries between a lattice model and correlation functions of  primary fields in the corresponding CFT.

Our mathematical arguments combine probability, combinatorics, and representation theory into the physical context, and in the presentation we have strived for accessibility for audiences from different backgrounds.
The rest of the introduction details the context of our work from CFT and probability perspectives and formulates our main results.

\subsection{CFT framework}

As tangented above, it is generally believed that scaling limits of relevant observables in critical lattice models are described by correlation functions of CFT local fields. In this context, the central charge $c \leq 1$ is a parameter that, roughly speaking, is expected to be related to different universality classes of the various critical models\footnote{For instance, for the critical Ising model the central charge equals $c=1/2$, and for the UST and LERW models considered in the present work the central charge equals $c=-2$.}.

An axiomatic approach to CFT termed \quote{conformal bootstrap} was developed by Belavin, Polyakov~\&~Zamolodchikov~\cite{BPZ:Infinite_conformal_symmetry_in_2D_QFT,Ribault:CFT_in_the_plane}. 
It provides an algebraic meaning to such local fields, thereby connecting the probabilistic and representation-theoretic frameworks that we adopt in the present work.  
A crucial ingredient of the conformal bootstrap is the \emph{spectrum} of the CFT, which is a module of the Virasoro algebra\footnote{In general, the spectrum is a module of one or two copies of the Virasoro algebra (depending on whether one considers a boundary or bulk CFT) --- the present work is concerned with boundary CFT.}. 
The so-called \emph{state-field correspondence} postulates the existence of an injective linear map from the spectrum of the CFT to the space of its local fields. 
However, local fields are usually mathematically ill-defined (and in any case should be treated, e.g., 
as operator-valued distributions rather than functions~\cite{DMS:CFT, Schottenloher:Mathematical_introduction_to_CFT}), whereas \emph{correlation functions} of local fields are expected to be well-defined functions of the field locations --- and indeed, it is the latter that feature in the statistical physics models, e.g., in formulas for scaling limits of connection probabilities.

The local fields whose correlation functions often appear in critical lattice model applications are so-called \emph{primary fields}. 
By the state-field correspondence, they are expected to correspond to highest-weight vectors in highest-weight modules over the Virasoro algebra. 
In particular, correlation functions of primary fields are \emph{conformally covariant}, with their behavior under conformal transformations entirely characterized by their \emph{conformal weights}. 
This is part of the conformal symmetry prominently present in CFT --- and the first feature that one tries to rigorously verify for a scaling limit of a critical lattice model 
(cf.~\cite{Schramm:ICM, 
Smirnov:Towards_conformal_invariance_of_2D_lattice_models}, Lemma~\ref{lem:covariance of fused partition functions} of the present article, as well as~\cite{CHI:Correlations_of_primary_fields_in_the_critical_planar_Ising_model} for recent progress). 

One of the most stunning discoveries of BPZ in~\cite{BPZ:Infinite_conformal_symmetry_in_2D_QFT} 
was that correlation functions of a special class of primary fields, called \emph{degenerate fields}, should furthermore satisfy certain linear homogeneous partial differential equations, now termed \emph{BPZ PDEs}. 
These equations proved to be extremely useful: 
for instance, BPZ %Belavin, Polyakov, and Zamolodchikov 
exploited them to derive critical exponents for the Ising model (and other minimal models)~\cite{BPZ:Infinite_conformal_symmetry_of_critical_fluctuations_in_2D}; 
starting from Cardy's work~\cite{Cardy:Critical_percolation_in_finite_geometries} 
the BPZ PDEs have been used to solve connection probabilities in many models (see~\cite{Peltola:Towards_CFT_for_SLEs, FLPW:Multiple_SLEs_Coulomb_gas_integrals_and_pure_partition_functions} and references therein); 
and they were recently used to prove the DOZZ formula for the structure constants of Liouville conformal field theory~\cite{KRV:Integrability_of_Liouville_theory:proof_of_the_DOZZ_formula}.

Let us briefly outline the formal derivation of the BPZ equations. The first important ingredient is representation-theoretic: a Verma module of the Virasoro algebra is said to be \emph{degenerate} if it contains a least one singular vector, i.e., a vector which generates a nontrivial submodule~\cite{Feigin-Fuchs:Verma_modules_over_Virasoro_book, Iohara-Koga:Representation_theory_of_Virasoro}. 
 Singular vectors have been completely classified (see, e.g.,~\cite{Kac:Highest_weight_representations_of_infinite_dimensional_Lie_algebras,
Schottenloher:Mathematical_introduction_to_CFT}), 
and they are characterized by (and only exist with) explicit 
conformal weights denoted $h_{r,s}$, where $r$ and $s$ are positive integers which lie in the so-called Kac table. 
Next, degenerate fields, denoted as $\Phi_{r,s}$, are defined as those that correspond to \emph{irreducible quotients} of degenerate Verma modules, and their conformal weights equal $h_{r,s}$. 
Finally, the conformal Ward identities, which express how CFT correlation functions change under infinitesimal variations of the complex structure of the underlying space (Riemann surface), 
(at least formally) enable one to translate singular vectors in degenerate Verma modules into BPZ PDEs satisfied by the correlation functions
(see~\cite{BPZ:Infinite_conformal_symmetry_in_2D_QFT, DMS:CFT, Dubedat:SLE_and_Virasoro_representations_localization, Dubedat:SLE_and_Virasoro_representations_fusion}).

Interestingly, special cases of BPZ PDEs have emerged independently from critical lattice models and random geometry.
In particular, BPZ PDEs of second order --- corresponding to the \quote{SLE field} $\Phi_{1, 2}$ --- 
naturally arise from martingale arguments that lie at the heart of the analysis 
of probabilistic observables in lattice models and their scaling limits 
(as observed early on, cf.~\cite{Cardy:Critical_percolation_in_finite_geometries, Bauer-Bernard:Conformal_field_theories_of_SLEs, Friedrich-Werner:Conformal_restriction_highest_weight_representations_and_SLE, BBK:Multiple_SLEs_and_statistical_mechanics_martingales, Smirnov:Towards_conformal_invariance_of_2D_lattice_models, Dubedat:Commutation_relations_for_SLE, Kozdron-Lawler:Configurational_measure_on_mutually_avoiding_SLEs}). 
These PDE solutions, commonly termed SLE \emph{partition functions} $\PartF (x_1, \ldots ,x_{2N})$ 
and representing correlation functions of the form 
\quote{$\langle \Phi_{1, 2} (x_1) \cdots \Phi_{1, 2} (x_{2N}) \rangle$,} 
provide scaling limits of (unfused) connection probabilities in various critical models~\cite{BBK:Multiple_SLEs_and_statistical_mechanics_martingales,
Izyurov:Smirnovs_observable_for_free_boundary_conditions_interfaces_and_crossing_probabilities,
FSKZ:A_formula_for_crossing_probabilities_of_critical_systems_inside_polygons,
Miller-Werner:Connection_probabilities_for_conformal_loop_ensembles,
Peltola-Wu:Global_and_local_multiple_SLEs_and_connection_probabilities_for_level_lines_of_GFF, 
Karrila:Computation_of_pairing_probabilities_in_multiple-curve_models,
Peltola-Wu:Crossing_probabilities_of_multiple_Ising_interfaces,
FPW:Connection_probabilities_of_multiple_FK_Ising_interfaces,  
LPW:UST_in_topological_polygons_partition_functions_for_SLE8_and_correlations_in_logCFT}.
Moreover, multiple non-intersecting chordal SLE curves, which describe critical interfaces in many models,
can be classified in terms of these SLE partition functions 
--- see, e.g.,~\cite{Peltola:Towards_CFT_for_SLEs, Karrila:Computation_of_pairing_probabilities_in_multiple-curve_models,  FLPW:Multiple_SLEs_Coulomb_gas_integrals_and_pure_partition_functions} and references therein. 

More generally, one could coalesce (or \quote{fuse}) starting points of the random curves (Figure~\ref{fig:fusion}), which should correspond to \emph{fusion} in CFT --- another important ingredient in conformal bootstrap.
Roughly speaking, the fields in a CFT are assumed to form an algebra with multiplication given by an asymptotic product (termed \textit{operator product expansion}, OPE) of their correlation functions as the field insertions collide.
For instance, comparing the OPE $\Phi_{1,2} \boxtimes \Phi_{1,2} = \Phi_{1,1} \boxplus \Phi_{1,3}$ 
of the SLE fields $\Phi_{1, 2}$
with the subleading asymptotic properties of the SLE partition functions with a given connection (see, e.g.,~\cite[Equation~(3.14), Remark~2.2, and the discussion in Section~4]{Peltola:Towards_CFT_for_SLEs} for details), 
it is natural to expect that two interface endpoints fused together at a point $x$ analogously corresponds to a field $\Phi_{1, 3} (x)$ in the correlation function. 
This was indeed rigorously verified for wired UST branches in~\cite{KKP:Boundary_correlations_in_planar_LERW_and_UST}.
In general,
as already proposed by physicists~\cite{Cardy:Conformal_invariance_and_surface_critical_behavior,Duplantier-Saleur:Exact_surface_and_wedge_exponents_for_polymers_in_two_dimensions, Bauer-Saleur:On_some_relations_between_local_height_probabilities_and_conformal_invariance,
Watts:A_crossing_probability_for_critical_percolation_in_two_dimensions, Duplantier:Conformal_random_geometry}, 
a number $s$ of fused interface endpoints at $x$ results in a field $\Phi_{1, s+1} (x)$.

In the present work, we focus on a particular case where the central charge equals $c=-2$ (and the $\SLE(\kappa)$ parameter $\kappa = 2$), 
corresponding to loop-erased random walk (LERW) 
paths viewed as branches in a wired uniform spanning tree (UST). 
The aforementioned results in~\cite{KKP:Boundary_correlations_in_planar_LERW_and_UST} verify the lattice\footnote{We briefly discuss continuum SLE curves in Appendix~\ref{app:fused SLE}. In particular, the fused partition functions in our main results indeed serve as SLE partition functions of fused SLE(2) curves.} 
curve interpretation of the fields $\Phi_{1,2}$ and $\Phi_{1,3}$ in terms of second and third order BPZ PDEs for their correlation functions, 
and the main purpose of the present work is to extend this result for correlation functions of all $\{\Phi_{1,s+1} \colon s \in \bZpos\}$. 

\subsection{CFT consequences}

Our proof of the higher-order BPZ PDEs in the UST model reveals an important feature about the spectrum of the underlying CFT. Indeed, it is sometimes 
wrongfully assumed that the spectrum of a given CFT can only contain irreducible quotients of degenerate Verma modules. 
A notable example which led to confusion in the physics community is the fact that the so-called \quote{imaginary} DOZZ formula does not vanish when one conformal dimension is specialized to $h_{1,1} = 0$ 
and the other two are arbitrary but different \cite{Ribault-Santachiara:Liouville_theory_with_a_central_charge_less_than_one}. 
This suggests the existence of a non-irreducible quotient of the degenerate Verma module associated with $h_{1,1}$ in the spectrum of imaginary Liouville CFT. 
In the context of Liouville CFT (with real coupling), all BPZ PDEs have been recently verified by probabilistic methods~\cite{KRV:Local_conformal_structure_of_LQG, Zhu:Higher_order_BPZ_equations_for_Liouville_CFT, Ang:Liouville_CFT_and_quantum_zipper, Baverez-Wu:Irreducibility_of_Virasoro_representations_in_Liouville_CFT}.
Interestingly, the so-called higher equations of motion in Liouville CFT were recently addressed mathematically~\cite{Baverez-Wu:Higher_equations_of_motion_at_level_2_in_LCFT, Cercle:Higher_equations_of_motion_for_boundary_LCFT} 
--- they provide another example leading to degenerate modules which are not the irreducible quotient of the corresponding Verma module. 
 
In contrast, a consequence of our work is that the spectrum of the CFT underlying the correlation functions 
of the UST boundary branches indeed contains \emph{irreducible quotients} of degenerate Verma modules, 
associated with conformal weights $h_{1,s+1}$ for any $s \in \mathbb Z_{>0}$. It would be very interesting to understand what such a CFT is precisely. 
A natural candidate would be an extension of the recently constructed \quote{compactified imaginary} Liouville theory~\cite{GKR:Compactified_imaginary_Liouville} 
--- here, in the case of surfaces with boundary, and central charge $c=-2$.

Finally, it is worth mentioning that the probabilistic works on Liouville theory generally focus on \emph{analytic} properties of the correlation functions, which tend to obscure the underlying \emph{algebraic} structures discussed above. 
Only very recently, the structure of degenerate Verma modules was addressed in~\cite{Baverez-Wu:Irreducibility_of_Virasoro_representations_in_Liouville_CFT} in the context of Liouville CFT, 
and in~\cite{Baverez-Jego:The_CFT_of_SLE_loop_measures_and_the_Kontsevich-Suhov_conjecture} in the context of SLE loops.
It would be interesting to extend the latter results to include field insertions at the boundary, which would correspond to chordal SLE curves.

\medskip 

We shall next describe the precise setup and the main (mathematical) results of our work. 

\subsection{UST framework}
\label{subsec:UST framework}

\begin{figure}
\begin{subfigure}[b]{\textwidth}
\begin{center}
\begin{overpic}[width=0.4\textwidth]{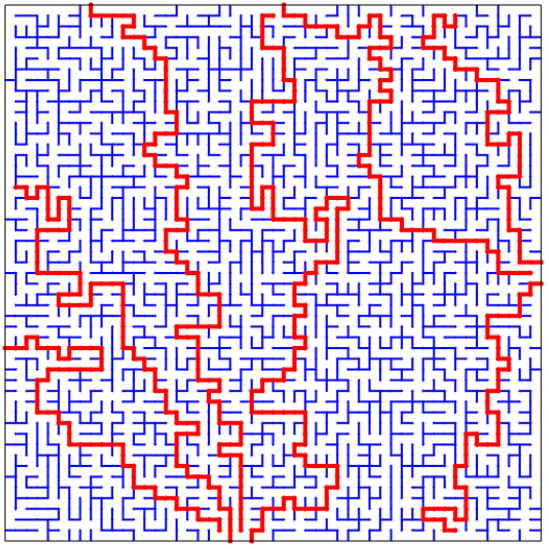}
 \put (-9,36) {\large $e_{14}$}
 \put (-9,63) {\large $e_{13}$}
 \put (16,102) {\large $e_{12}$}
 \put (42,102) {\large $e_{11} e_{10}$}
 \put (35,-3) {\large $e_1 \cdots e_4$}
 \put (80,-3) {\large $e_5$}
 \put (101,43) {\large $e_6$}
 \put (101,48) {\large $e_7$}
 \put (101,53) {\large $e_8$}
 \put (80,102) {\large $e_9$}
\end{overpic}
\qquad
\raisebox{0\textwidth}{
\includegraphics[width=0.35\textwidth]{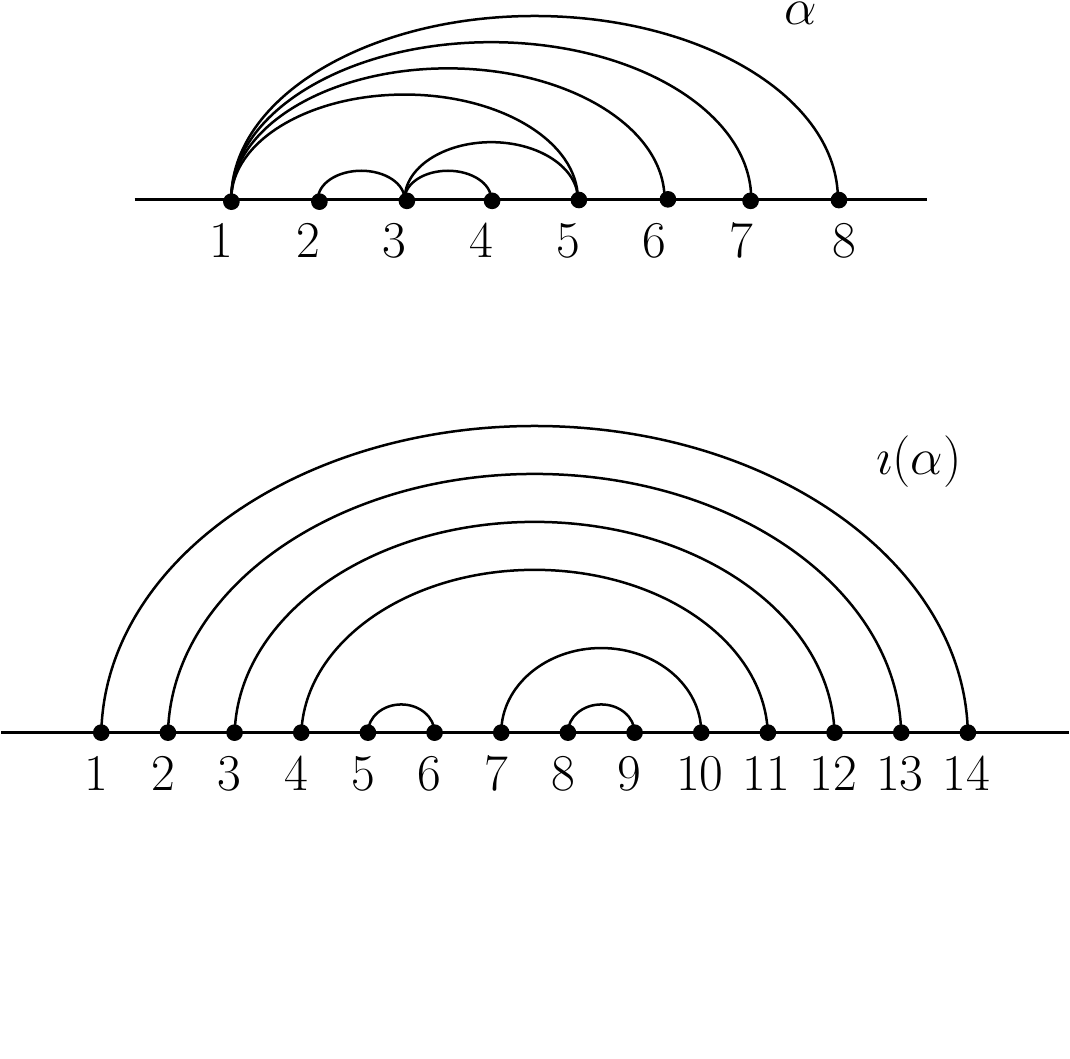}
}

\bigskip
\bigskip
\end{center}
\caption{A simulation of the wired uniform spanning tree on a grid of $50 \times 50$ squares, conditional on the presence of long boundary-to-boundary branches pairing up the marked boundary edges $e_1, \ldots, e_{14}$ with $7$ disjoint curves according to the link pattern $\imath(\alpha)$. Some of the marked edges are only microscopically apart ---  
hence, the scaling limit of such pairings would rather be depicted by the \textit{valenced} (i.e., \quote{fused}) link pattern $\alpha$. 
Here and throughout, $\imath$ denotes the \quote{unfusing} map, which, as illustrated on the right side of this figure, associates to each valenced link pattern $\alpha \in \LP_\multii$ with multiplicities $\multii = (s_1,\ldots,s_d) \in \bZpos^d$ 
an $N$-link pattern $\imath(\alpha) \in \LP_N$ with $2N = \sum_{i=1}^d s_i$. 
Precisely, $\imath(\alpha)$ is defined combinatorially as follows: we split each index $i \in \{1,2,\ldots,d \}$ of $\alpha$ into $s_i$ distinct indices, and attach the $s_i$ links ending at $i$ in $\alpha$ to these new $s_i$ indices,  
so that the link possessing the leftmost endpoint gets the leftmost one of the new indices, etc. 
Finally, we label all indices from left to right by $1,2,\ldots,2N$ to obtain an $N$-link pattern $\imath(\alpha) \in \LP_N$. 
}
\end{subfigure}

\bigskip
\bigskip
\bigskip

\begin{subfigure}[b]{\textwidth}
\begin{center}
\begin{overpic}[width=0.5\textwidth]{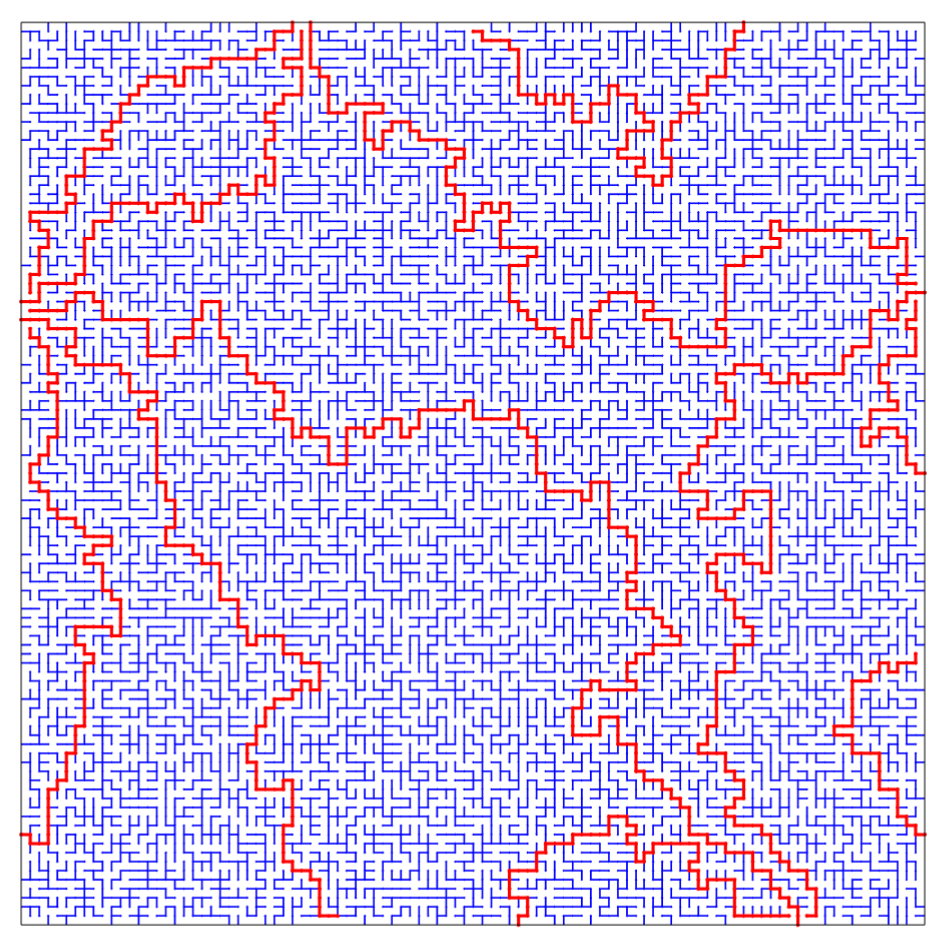}
 \put (34,-2) {\large $e_1$}
  \put (55,-2) {\large $e_2 \cdots$}
\end{overpic}
\qquad
\raisebox{0.1\textwidth}{
\includegraphics[width=0.38\textwidth]{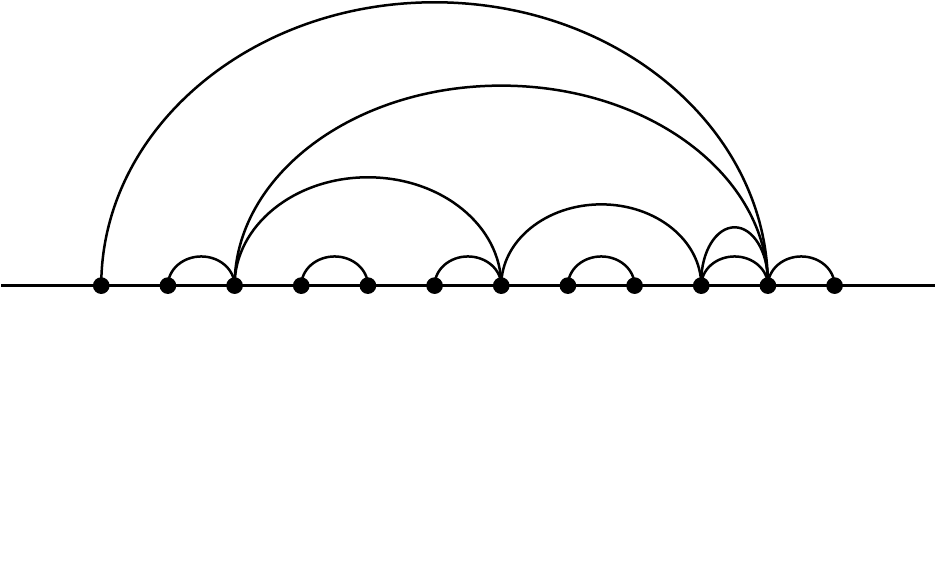}
}

\bigskip
\bigskip
\end{center}
\caption{A simulation and the valenced link pattern of a more complicated geometry with $11$ curves and a $100\times 100$ grid. The main result of the present article relates the scaling-limit probability of such fused pairings to certain explicit conformal field theory correlation functions.}
\end{subfigure}
\caption{\label{fig: UST fused}
Simulations of USTs with some boundary branches highlighted.}
\end{figure}
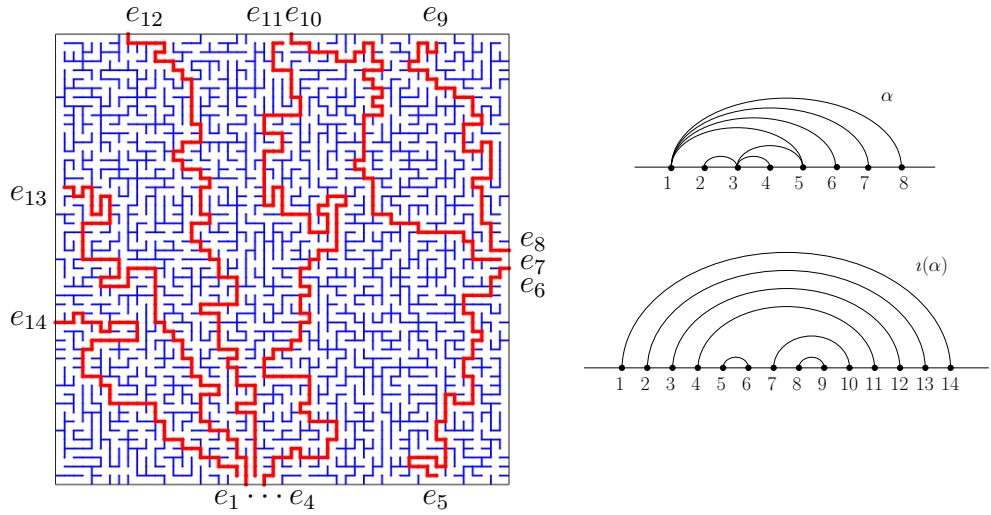
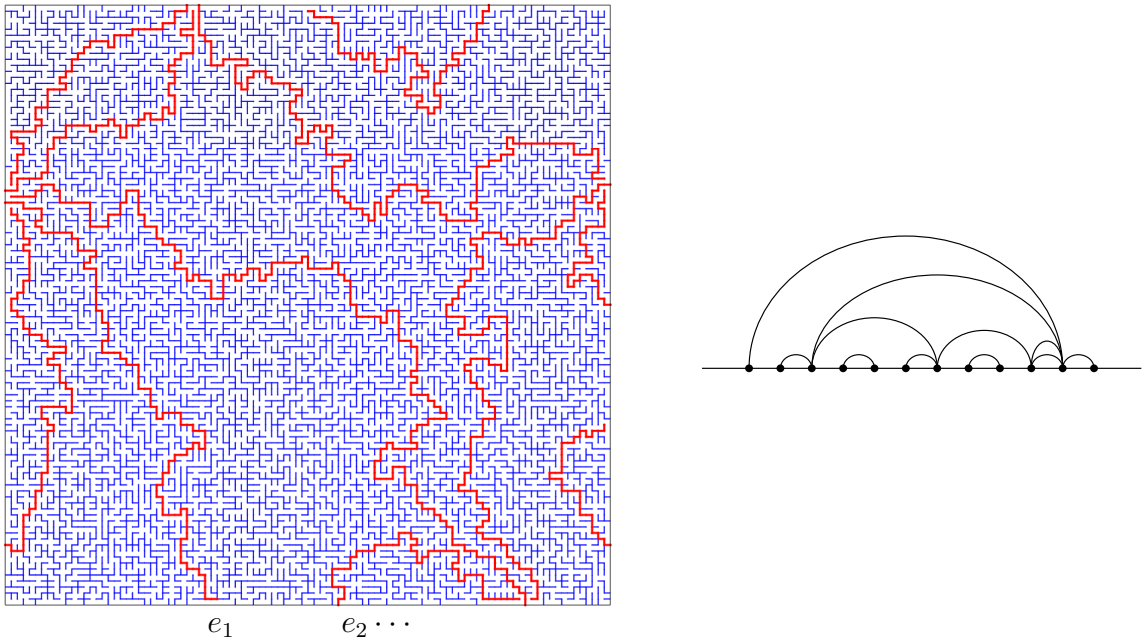

We first discuss the combinatorial setup without reference to scaling limits. 
A~\emph{spanning tree} of a finite connected (multi)graph is a connected acyclic subgraph (i.e., a tree) that contains all the vertices of the original graph (i.e., is spanning). 
On any finite connected graph, there is a finite positive number of spanning trees, so one can be sampled uniformly at random.

Conventionally in planar statistical mechanics 
(and throughout this article) \quote{wired} boundary conditions are imposed to this model, in order to guarantee a spatial Markov property. 
Let %$\sG=(\sV, \sE)$ 
\gls{symb:sG} 
be a finite connected planar graph with a fixed embedding, 
and partition $\sV$ into \emph{interior vertices} 
%$\sV_\circ$ 
\gls{symb:sVpartial} 
and \emph{boundary vertices} 
%$\sV_\partial$, 
\gls{symb:sVcirc}, the latter being those adjacent to the infinite face. 
The~\emph{wired uniform spanning tree} (referred to from now on as \emph{UST}) on $\sG$ is a uniformly sampled random spanning tree on the (multi)graph 
$G = (V, E)$ 
%\gls{symb:multigraph} 
obtained by identifying all vertices $\sV_\partial$ on $\sG$ together. 
Via the bijection $E \leftrightarrow \sE$, we then identify such spanning trees on $G$ with certain spanning forests on~$\sG$ (namely those whose component trees each contain one of the boundary vertices). 
See Figure~\ref{fig:UST combinatorial} (left) for an illustration.

In a wired spanning tree, each interior vertex $v \in \sV_\circ$ connects to the set $ \sV_\partial$ of boundary vertices by a unique simple path, called the \emph{boundary branch from $v$}. 
Let $e_1, \ldots, e_{2N}$ be distinct boundary edges 
(defined as edges connecting an interior vertex in $\sV_\circ$ to a boundary vertex in $\sV_\partial$), by convention always indexed counterclockwise along the infinite face. 
We will denote by 
%$\event$ 
\gls{symb:event} 
the event that the boundary branches from the interior vertices of the odd-index edges $e_1, e_3, \ldots, e_{2N-1}$ connect to the boundary $\partial \sV$ via the even-index edges $e_2, e_4, \ldots, e_{2N}$, 
each using a different even edge; 
see Figure~\ref{fig:UST combinatorial} (right).
On this event, the union of the odd edges $e_1, e_3, \ldots, e_{2N - 1}$ and the boundary branches from their interior vertices to the even edges form a collection of $N$ chordal, vertex-disjoint simple paths on $\sG$, 
called the \emph{boundary-to-boundary branches}, 
pairing the marked edges $e_1, \ldots, e_{2N}$.

On the event $\event$, we denote by 
%$\conn$ 
\gls{symb:conn} 
the random partition of the indices $\{ 1,2, \ldots, 2N \}$ into pairs, encoding the pairing of the edges $e_1, \ldots, e_{2N}$ by the boundary-to-boundary branches. 
Due to topological constraints,  the partition $\conn$ only takes values in the set 
%$\LP_N$ 
\gls{symb:LP} 
of \emph{planar pair partitions}, a.k.a.~(planar) $N$-\emph{link patterns},
i.e., $\conn$ is a random partition 
\begin{align*}
\alpha = \{ \linkInEquation{a_1\;}{\;b_1\,}, \ldots, \linkInEquation{a_N\;}{\;b_N\,} \} 
= \{ \{a_1,b_1\} , \ldots, \{a_N,b_N\} \} 
\end{align*}
of the integers $\{ 1, \ldots, 2N\}$ into pairs $\link{a\;}{\;b\,} = \{a,b\} \subset \{ 1, \ldots, 2N\}$ called \emph{links},
such that the pairs of $\alpha$ (as pairs of points on the real axis) 
can be connected by $N$ disjoint curves in the upper half-plane. 
For fixed $N$, there are exactly a Catalan number 
$\tfrac{1}{N+1} \binom{2N}{N}$ of different $N$-link patterns.

The central objects of interest in this article are the probabilities of the different crossing events, 
which we shall denote as 
\begin{align} \label{eq: event of interest}
%\event(\alpha)
\textnormal{\gls{symb:connevent}}
 := \event \cap \{\conn = \alpha\}.
\end{align}
They describe important geometric information of the model 
(indeed, one could even define a topology in terms of all crossing events, say, in the spirit of~\cite{Schramm-Smirnov:Scaling_limits_of_planar_percolation}) and provide natural martingale observables (e.g.,~\cite{BBK:Multiple_SLEs_and_statistical_mechanics_martingales, Kemppainen-Smirnov:Configurations_of_FK_Ising_interfaces, Karrila:UST_branches_martingales_and_multiple_SLE2}). 
Moreover, their scaling limits can be interpreted both as CFT correlation functions (Theorem~\ref{thm:CFT properties})
and as $\SLE(2)$ measure conversions (Appendix~\ref{app:fused SLE}); see~\cite{BBK:Multiple_SLEs_and_statistical_mechanics_martingales, Dubedat:Commutation_relations_for_SLE, Kozdron-Lawler:Configurational_measure_on_mutually_avoiding_SLEs, JJK:SLE_boundary_visits, KKP:Boundary_correlations_in_planar_LERW_and_UST, Karrila:UST_branches_martingales_and_multiple_SLE2}.

The key to deriving the scaling limit results is an underlying \emph{determinantal structure} of the connection probabilities. Indeed, 
in both the prior and present scaling limit results, the starting point is to express the probabilities of $\event(\alpha)$ 
in terms of determinants of matrices whose entries are suitable random walk excursion kernels,
building on Fomin's  formula~\cite{Fomin:LERW_and_total_positivity} 
for loop-erased random walks (LERW) and 
Wilson's algorithm~\cite{Wilson:Generating_random_spanning_trees_more_quickly_than_cover_time}  
relating LERW to UST branches~\cite{
Kenyon-Wilson:Boundary_partitions_in_trees_and_dimers, 
Kenyon-Wilson:Double_dimer_pairings_and_skew_Young_diagrams, 
KKP:Boundary_correlations_in_planar_LERW_and_UST}.
With appropriate boundary regularity assumptions and renormalizations, such excursion kernels are known to converge to the Brownian excursion kernel~\cite{CFL:Uber_PDE_der_mathphys, Chelkak-Smirnov:Discrete_complex_analysis_on_isoradial_graphs}. 
Thus, in the case where the marked boundary edges $e_1, \ldots, e_{2N}$ 
stay at a \emph{macroscopic} distance away from each other, it is not hard to see that 
the probabilities of the different UST boundary-to-boundary branch connectivities (when suitably renormalized) converge 
to a conformally covariant limit, expressed analogously in terms of determinants of Brownian excursion kernels~\cite{Kenyon-Wilson:Boundary_partitions_in_trees_and_dimers,
 Kenyon-Wilson:Double_dimer_pairings_and_skew_Young_diagrams,
 KKP:Boundary_correlations_in_planar_LERW_and_UST}. 
Their CFT/SLE interpretation then boils down to the second order BPZ PDEs for these limits, 
which can be verified by a direct computation, see~\cite[Theorem~4.1]{KKP:Boundary_correlations_in_planar_LERW_and_UST}.

In contrast, in the case where the distance of some of the boundary edges is \emph{microscopic} (for instance, one lattice step as in Figure~\ref{fig: UST fused}), some of the UST boundary-to-boundary branches would typically be shrunken to a point in the scaling limit. 
In order to still force the existence of $N$ macroscopic interfaces connecting the marked boundary edges, we restrict our attention to connectivities that do not have a link between 
boundary edges that converge together in the limit. 
Upon passing to the scaling limit (which we do in Theorems~\ref{thm:scaling limit of pinched pertition functions}~\&~\ref{thm:CFT properties} below),
both the convergence of the connection probabilities 
and especially the BPZ PDEs turn out to be significantly harder than in the above, unfused case.

\subsection{Main results: UST connection probabilities and boundary correlations}

\subsubsection{Precise scaling limit setup}
\label{subsubsec:scaling limit setup}

Our scaling limit results are valid in any setup where the random walk excursion 
kernels and their discrete boundary derivatives on the graphs converge to the Brownian 
excursion kernels and their derivatives. 
To be specific, we present the results in square grid approximations obtained from subgraphs of the rescaled square lattice $\delta\bZ^2 = (\sV(\delta\bZ^2), \sE(\delta\bZ^2))$ as the scaling parameter $\delta$ tends to zero. 
Let %$\domain^\delta \subsetneq \bC$ 
\gls{symb:domain} 
be a bounded simply connected domain, whose boundary $\partial \domain^\delta$ consists of a simple loop on $\delta\bZ^2$. 
Abusing notation, 
we also denote by $\domain^\delta = (\sV^\delta, \sE^\delta)$ the planar graph with vertices\footnote{Here, $\smash{\overline{\domain}}^\delta \subsetneq \bC$ denotes the closure as a planar set.}
 $\sV^\delta := \sV(\delta\bZ^2) \cap \smash{\overline{\domain}}^\delta$ and edges $\sE^\delta := \sE(\delta\bZ^2) \cap \smash{\overline{\domain}}^\delta$. 
The boundary vertices of $\domain^\delta$ are $\sV^\delta_\partial = \sV^\delta \cap \partial \domain^\delta$. 

Consider distinct counterclockwise boundary edges $e_1^\delta, \ldots, e_{2N}^\delta$ of $\domain^\delta$ 
and a bounded simply connected domain $\domain \subsetneq \bC$ together 
with distinct counterclockwise boundary points $p_1, \ldots, p_d \in \partial \domain$, with $d \leq 2N$. 
Fusing the endpoints of $N$ curves onto $d$ limit points, 
we shall denote throughout by 
%$\multii = (s_1,\ldots,s_d) \in \bZpos^d$ 
\gls{symb:multii} $ \in \bZpos^d$  
their fused \emph{multiplicities}, a.k.a.~\emph{valences}, 
\begin{align*}
s_0 := 0, 
\qquad \sum_{i=1}^d s_i = 2N ,
\qquad 
\underset{1 \leq j \leq d}{\max} s_j \leq N ,
\qquad \textnormal{and} \qquad 
\textnormal{\gls{symb:summ}.}
\end{align*}

The corresponding scaling limit objects are naturally labeled by the set 
%$\LP_\multii$ 
\gls{symb:LPmultii} 
of $\multii$-\emph{valenced link patterns} without defects.
These are defined precisely in~\cite[Section~2.5]{Peltola:Basis_for_solutions_of_BSA_PDEs_with_particular_asymptotic_properties}, where they were called \quote{planar link patterns} involving multiplicities 
(and in~\cite[Section~2]{Flores-Peltola:Standard_modules_radicals_and_the_valenced_TL_algebra}, where they were called \quote{valenced link patterns}). 
Due to the embedding of $\LP_\multii$ into the set  $\LP_N$ of $N$-link patterns given by the \quote{unfusing} map 
%$\imath$ 
defined in Figure~\ref{fig: UST fused} and its caption, 
we will throughout identify each valenced link pattern $\alpha \in \LP_\multii$ with the corresponding unfused link pattern $\imath(\alpha) \in \LP_N$. 
For convenience, we will primarily just use the same notation $\alpha$ for both. 
(It is important to note, however, that not all elements in $\LP_N$ have a corresponding valenced counterpart; the map $\imath$ is not a surjection.) 

\begin{definition}
Consider $(\domain^\delta; e_1^\delta, \ldots, e_{2N}^\delta)$, valences $\multii$, and $(\domain; p_1, \ldots, p_d)$ as above.
\begin{itemize}
\item 
We say that $(\domain^\delta; e_1^\delta, \ldots, e_{2N}^\delta)$ converges as $\delta \to 0$ to $(\domain; p_1, \ldots, p_d)$ with valences $\multii$ 
in the \emph{Carath\'{e}odory sense} if there exist conformal maps $\varphi^\delta$ from $\domain^\delta$ onto the unit disk $\bD$,
and a conformal map $\varphi$ from $\domain$ onto $\bD$,
such that $(\varphi^\delta)^{-1}$ converges to $\varphi^{-1}$ locally uniformly on $\bD$,
and for each $1 \leq j \leq d$, 
the boundary edges $e_{\summ_{j-1}+1}^\delta, \ldots, e_{\summ_j}^\delta$ converge to the boundary point $p_j$ in the sense that 
$\smash{\varphi^\delta(v_i^\delta) \to \varphi(p_j)}$ for all $i \in \{ \summ_{j-1}+1,\ldots,\summ_j \}$, where $v_i^\delta \in \sV^\delta_\partial$ is the boundary vertex of $e_i^\delta$.

\smallskip

\item 
We will furthermore need compactness and boundary regularity: we always assume that the sets $\domain^\delta$, $\delta > 0$, are uniformly bounded and that the boundary of both $\domain$ and each $\domain^\delta$ 
is locally a straight horizontal or vertical line segment at some fixed neighborhoods of the marked boundary points $p_1, \ldots, p_d$. 
\end{itemize}
Throughout, we call %$(\domain; p_1, \ldots, p_d)$ 
\gls{symb:polygon}
with the above assumptions an \emph{admissible $d$-polygon}, and 
%$(\domain^\delta; e_1^\delta, \ldots, e_{2N}^\delta)$ 
\gls{symb:discrpolygon}
with this property \emph{(discrete) admissible $2N$-polygons}.
\end{definition}

\subsubsection{Convergence of connection probabilities}

To state the convergence result, we need a few more definitions. 
First, the \emph{Brownian excursion kernel} is the unique conformally covariant function which in the upper half-plane 
$\bH := \{z \in \bC \;|\; \im(z) > 0 \}$ equals
\begin{align*}
\ExcKH_\bH(x, y) %= \ExcKH(x, y) 
:= \frac{1}{\pi} \frac{1}{(x-y)^2} , \qquad 
\textnormal{for $x , y \in \bR$, $x \neq y$},
\end{align*}
and which in other simply connected domains $\domain$ is obtained from $\ExcKH_\bH$ by the formula
\begin{align}
\label{eq:BM exc ker covariance}
\textnormal{\gls{symb:BEK}} 
%\ExcKH_\domain(x, y) 
:= |\varphi'(x)| \, |\varphi'(y)| \, \ExcKH_\bH(\varphi(x), \varphi(y)) 
=  \frac{1}{\pi} \frac{|\varphi'(x)| \, |\varphi'(y)|}{(\varphi(x) - \varphi(y))^2}  ,
\end{align}
where $\varphi \colon \domain \to \bH$ is any conformal bijection and the boundary of $\domain$ is assumed to be sufficiently regular near $x$ and $y$ (e.g. $C^{1+\epsilon}$ for some $\epsilon>0$). 

Additionally, Equation~\eqref{eq: LPdet scaling limit} in our result includes certain combinatorial objects, namely 
the \emph{partial order relation} %$\DPgeq$ 
\gls{symb:DPgeq} 
on $\LP_N$, and the count $\# \CItilingsof (\alpha / \beta)$ of \emph{cover-inclusive Dyck tilings}
of the skew Young diagram~$\alpha/\beta$ for $\alpha, \beta \in \LP_N$, as detailed in Definitions~2.1~\&~2.8 of~\cite{KKP:Boundary_correlations_in_planar_LERW_and_UST}, 
respectively (both extended to $\LP_\multii$ via the unfusing map $\imath$). 
Because we shall not need the precise definitions in the present work, let us just make two comments:  
the limit~\eqref{eq:scaling limit of pinched pertition functions} below is a linear combination of certain, explicitly defined determinants, 
but no general closed-form algebraic expression is known for $\# \CItilingsof (\alpha / \beta)$ (to the best of our knowledge).

\medskip

Now, with the setup explicitly described, we state our first main result. 

\begin{theorem} \label{thm:scaling limit of pinched pertition functions} 
Let admissible $2N$-polygons $(\domain^\delta; e_1^\delta, \ldots, e_{2N}^\delta)$ converge as $\delta \to 0$ to an admissible $d$-polygon $(\domain; p_1, \ldots, p_d)$ with valences $\multii = (s_1,\ldots,s_d)$ 
in the Carath\'{e}odory sense, 
furthermore so that the fused boundary edges are for each $\delta$ exactly one lattice step apart 
(as illustrated in Figure~\ref{fig: UST fused}).
Then, for any $\multii$-valenced link pattern $\alpha \in \LP_\multii$ we have 
\begin{align} \label{eq:scaling limit of pinched pertition functions}
\lim_{\delta \to 0} \, \delta^{-2N} \prod_{i=1}^d \delta^{-(s_i-1)s_i/2} \, \PR^\delta[\event(\alpha)] 
= \; & \sum_{ \substack{\beta \in \LP_N \\ \beta \DPgeq \alpha}} \# \CItilingsof (\alpha / \beta) \, \Delta^{\mathfrak{K}}_\beta  ,
\end{align} 
where 
\begin{align} \label{eq: LPdet scaling limit}
\Delta^{\mathfrak{K}}_\beta := \det \big( \mathfrak{K}(a_k, b_\ell) \big)_{k,\ell=1}^N 
, \qquad 
\beta = \{ \linkInEquation{a_1}{b_1}, \ldots, \linkInEquation{a_N}{b_N} \} \in \LP_N ,
\end{align}
is a Fomin type determinant (defined in Section~\ref{subsec:inv Fomin}); the kernel $\mathfrak{K}$ here is given by
\begin{align}
\label{eq: valenced kernel}
\mathfrak{K}(a,b) & =
(\partial^{\mathrm{tan}}_i)^{m_a-1} (\partial^{\mathrm{tan}}_j)^{m_b-1}
\ExcKH_\domain(p_i, p_j) ,
\quad 
a = \summ_{i-1}+m_a \; \textnormal{ and } \; b = \summ_{j-1}+m_b ,
\end{align}
if the boundary edges $e^\delta_a$ and $e^\delta_b$ are not fused together, i.e., $\summ_{i-1} <  a \leq \summ_i$  and  $\summ_{j-1} < b \leq \summ_j$ for some $i \neq j$, 
and $\mathfrak{K}(a,b) = 0$ otherwise, i.e., if $\summ_{j-1} < a, b \leq \summ_j$ for some $1 \leq j \leq d$.
Here, $\partial^{\mathrm{tan}}_k$ denotes the counterclockwise tangential derivative at the boundary point $p_k$, for $1 \leq k \leq d$.
\end{theorem}

We give explicit examples of the limiting functions in Appendix~\ref{app:explicit determinants}.

Determinantal expressions could, of course, be \quote{zeroes in disguise.} 
However, property \textnormal{(POS)} in Theorem~\ref{thm:CFT properties} below states that the right-hand side of Equation~\eqref{eq:scaling limit of pinched pertition functions} is not only non-negative but also strictly positive. 
This is also crucial in SLE applications (Appendix~\ref{app:fused SLE}).

Let us remark that it is only for the sake of simplicity of presentation that we assume all the fused boundary edges to be exactly one lattice step apart. 
\emph{Mutatis mutandis}, for separations of, e.g., five or $\lfloor 1/\sqrt{\delta} \rfloor$ lattice steps, the correct renormalization factors in the product in~\eqref{eq:scaling limit of pinched pertition functions} would become $(5\delta)^{-(s_i-1)s_i/2}$ or $\delta^{-(s_i-1)s_i/4}$, respectively.

The proof of Theorem~\ref{thm:scaling limit of pinched pertition functions} (which will be completed in Section~\ref{subsec:proof of pinched pertition functions}) 
relies on deriving exactly analogous determinantal formulas for the discrete probabilities $\PR^\delta[\event(\alpha)]$ 
given in terms of discrete excursion kernels and discrete derivatives, 
and then controlling their scaling limits 
(see Theorems~\ref{thm:pinched pertition functions}~\&~\ref{thm:conv of exc kernels and derivatives} in Sections~\ref{sec:discrete}~\&~\ref{sec:scaling_limit}). 
The combinatorics in the proof of Theorem~\ref{thm:scaling limit of pinched pertition functions} will be also needed later to prove Theorem~\ref{thm:CFT properties} \textnormal{(FUS)} in Section~\ref{subsec:fusion of partition functions}.

\subsubsection{Degenerate boundary correlation functions}
\label{subsec:degenerate conformal blocks}
 
The linear combinations of the determinants
$\Delta^{\mathfrak{K}}_\beta = \Delta^{\mathfrak{K}}_\beta(\domain; p_1, \ldots, p_d)$ 
appearing in Equation~\eqref{eq:scaling limit of pinched pertition functions}, 
\begin{align} \label{eq:fused partition function}
\textnormal{\gls{symb:PartF}} (\domain; p_1, \ldots, p_d) 
%\PartF_\alpha(\domain; p_1, \ldots, p_d) 
:= \sum_{ \substack{\beta \in \LP_N \\ \beta \DPgeq \alpha}} \# \CItilingsof (\alpha / \beta) \, \Delta^{\mathfrak{K}}_\beta (\domain; p_1, \ldots, p_d) , \qquad \alpha \in \LP_\multii ,
\end{align}
with valences $\multii = (s_1,\ldots,s_d)$, could be interpreted as boundary correlation functions 
\begin{align}\label{eq:fusion_corrf}
`` \langle \Phi_{1,s_1+1}(p_1) \cdots \Phi_{1,s_d+1}(p_d) \rangle_\alpha , " 
\qquad \alpha \in \LP_\multii ,
\end{align}
in a CFT with central charge $c=-2$ of primary fields 
$(\Phi_{1,s_1+1},\ldots,\Phi_{1,s_d+1})$ with conformal weights 
$h_{1,s_j+1} = s_j (s_j+1)/2$ for all $j$, 
which in particular belong to the first row of the Kac table of degenerate weights. 
These functions $\{\PartF_\alpha \; | \; \alpha \in \LP_\multii\}$ are conventionally termed \emph{pure partition functions}.
We prove in Theorem~\ref{thm:CFT properties} that they are positive and linearly independent. 
 
Interestingly, the correlation functions of these fields are expected to satisfy a system of BPZ PDEs of orders $(s_1+1,\ldots,s_d+1)$, Equation~\eqref{eq: BPZ PDE at kappa equals 2}. 
(The explicit expressions~\eqref{eq: BPZ operator at kappa equals 2} for the BPZ partial differential operators corresponding to the first row of the Kac table were found by Beno\^it~\&~Saint-Aubin~\cite{BSA:Degenerate_CFTs_and_explicit_expressions_for_some_null_vectors}.)
We indeed verify this rigorously in Theorem~\ref{thm:CFT properties}, 
for the complete first row of the Kac table, 
for the pure partition functions $\PartF_\alpha$ of~\eqref{eq:fused partition function}. 
In the special cases where the valences equal one or two, the PDEs 
were verified in~\cite{KKP:Boundary_correlations_in_planar_LERW_and_UST}.

By virtue of conformal covariance (stated in Lemma~\ref{lem:covariance of fused partition functions}) 
it suffices to investigate the functions $\PartF_\alpha$ in the domain $\bH$ and with boundary point configurations in the chamber
\begin{align} \label{eq:def_chamber_d}
%\chamber_{d} 
\textnormal{\gls{symb:chamber}} := \{(x_1, \dots, x_d) \in \bR^d \; | \; x_1<\cdots<x_d\}.
\end{align}

\begin{theorem} \label{thm:CFT properties} 
The functions 
$\PartF_\alpha(\cdot) := \PartF_\alpha(\bH; \cdot) \colon \chamber_{d} \to \bR$ 
of~\eqref{eq:fused partition function} satisfy the following properties. 
\begin{itemize}
\item[\textnormal{(PDE)}] 
\textnormal{\bf System of $c=-2$ BPZ PDEs:} 
With valences $\multii = (s_1,\ldots,s_d) \in \bZpos^d$, we have
\begin{align} \label{eq: BPZ PDE at kappa equals 2} 
%\sD_{s_j+1}\super{x_j} 
\textnormal{\gls{symb:BPZop}} 
\; \PartF_\alpha(x_1,\ldots ,x_d) = 0 , \qquad \textnormal{for all } j=1,\ldots,d ,
\end{align}
where $\sD_{s_j+1}\super{x_j}$ are the BPZ differential operators
\begin{align} \label{eq: BPZ operator at kappa equals 2} 
\sD_{s_j+1}\super{x_j} 
:= \sum_{k=1}^{s_j+1} \sum_{\substack{m_1,\ldots ,m_k \geq 1 \\ m_1+\cdots +m_k = {s_j+1}}} \frac{ (-1/2)^{k-s_j-1} (s_j!)^2}{\prod_{l=1}^{k-1} (\sum_{i=1}^l m_i)(\sum_{i=l+1}^k m_i)} \times \sL_{-m_1}\super{x_j} \; \cdots \; \sL_{-m_k}\super{x_j} ,
\end{align} 
and where $\sL_{-m}\super{x_j}$ are the first order differential operators
\begin{align*} 
%\sL_{-m}\super{x_j} 
\textnormal{\gls{symb:Virop}} 
:= -\sum_{\substack{1\leq i \leq d\\i\neq j}} \Big( (x_i-x_j)^{1-m} \pder{x_i} + \frac{(1-m) (s_i + 1)s_i}{2} (x_i-x_j)^{-m}\Big) .
\end{align*}

\medskip

\item[\textnormal{(COV)}] 
\textnormal{\bf Covariance:} 
For all M\"obius maps $\Mob \colon \bH \to \bH$ such that $\Mob(x_1) < \cdots < \Mob(x_{d})$, we have
\begin{align}
\label{eq: COV general at kappa equals 2} 
& \PartF_\alpha(x_1, \ldots, x_{d}) 
= \prod_{j=1}^{d} |\Mob'(x_j)|^{(s_j+1)s_j/2} \times \PartF_\alpha(\Mob(x_1), \ldots, \Mob(x_{d})) .
\end{align}

\medskip

\item[\textnormal{(POS)}] 
\textnormal{\bf Positivity:} 
For each $\alpha \in \LP_\multii$, we have $\PartF_\alpha(x_1,\ldots ,x_d)>0$ for all $(x_1,\ldots ,x_d)\in\chamber_{d}$.

\medskip

\item[\textnormal{(LIN)}] 
\textnormal{\bf Linear independence:} 
The functions $\{\PartF_\alpha \; | \; \alpha \in \LP_\multii\}$ are linearly independent. 

\medskip

\item[\textnormal{(FUS)}] 
\textnormal{\bf Fusion:} 
For each $\alpha \in \LP_\multii$, 
we have the iterated limit expression
\begin{align} \label{eq:iterated limit expression}
\PartF_\alpha (p_1, \ldots, p_d) = \; &  \bigg( \prod_{j=1}^d ( 0! \cdot 1! \cdot \ldots \cdot (s_j-1)!) \bigg) \\
\nonumber & \times \bigg(
 \lim_{x_{\summ_d} \to p_d} |x_{\summ_d} - p_d|^{-s_d+1}
 \ldots
 \lim_{x_{\summ_{d-1} + 2} \to p_d} |x_{\summ_{d-1} + 2} - p_d|^{-1}
 \lim_{x_{\summ_{d-1} + 1} \to p_d}
 \bigg) \\
\nonumber & \; \ldots \;
\bigg( \lim_{x_{\summ_1} \to p_1} 
|x_{\summ_1} - p_1|^{-s_1+1} \ldots \lim_{x_2 \to p_1} |x_2 - p_1|^{-1} \lim_{x_1 \to p_1} \bigg)
\PartF_{\iota(\alpha)} (x_1, \ldots, x_{2N}) 
\end{align} 
(where the limits are taken one at a time), 
denoting by $\iota(\alpha)$ the $N$-link pattern obtained by \quote{opening up} the valenced nodes in $\alpha$, as illustrated in Figure~\ref{fig: UST fused} (top right).

\end{itemize}
\end{theorem}

There is an internal hierarchy to Theorem~\ref{thm:CFT properties}, where the core statement is \textnormal{(PDE)}, which is proven via a representation-theoretic argument.
\textnormal{(FUS)} provides the link to the probabilistic part, and a necessary input to \textnormal{(PDE)}. 
Another crucial input are certain Frobenius series expansions and local complex analytic extensions of the functions $\PartF_\alpha$, 
which follow from the definition of $\PartF_\alpha$ as linear combinations of determinants of rational functions (see Section~\ref{subsec:Frobenius series}) 
The rest of the above theorem follows from \textnormal{(PDE)} and \textnormal{(FUS)} with softer arguments.

We shall prove Theorem~\ref{thm:CFT properties} in Sections~\ref{subsec:fusion of partition functions},~\ref{subsec:covariance of partition functions}~\&~\ref{sec:blocks}. 
Importantly, one can derive the BPZ PDEs by an inductive argument, 
starting from the known second order PDEs in the unfused case~\cite[Theorem~4.1]{KKP:Boundary_correlations_in_planar_LERW_and_UST} 
and recursively fusing boundary points and increasing the order of the PDEs (see Figure~\ref{fig:fusion} for an illustration and Section~\ref{sec:blocks} for an outline).
Dub\'edat~\cite{Dubedat:SLE_and_Virasoro_representations_fusion} 
provided a representation-theoretic argument for this when the central charge $c < 1$ is irrational\footnote{The representation theory of the Virasoro algebra is completely understood.
In the $c \notin \bQ$ case, all of its Verma modules are either of \quote{point} or \quote{link} type, while for $c \in \bQ$, 
the submodule structure can be much more complicated (point, link, \quote{chain,} or \quote{braid}) --- see~\cite[Figure~1]{Kytola-Ridout:On_staggered_indecomposable_Virasoro_modules}, 
and~\cite{Feigin-Fuchs:Verma_modules_over_Virasoro_book, Iohara-Koga:Representation_theory_of_Virasoro} for details.},
while Karrila, Kyt{\"o}l{\"a}~\&~Peltola~\cite{KKP:Boundary_correlations_in_planar_LERW_and_UST} 
performed an elementary computation 
in the case of third-order PDEs only, but with any central charge $c \leq 1$.
In the present work, we shall carry out Dub\'edat's approach in the case where $c=-2$, 
filling in details where the differences to the irrational case arise\footnote{The key difference concerns the Virasoro module structure: we encounter Verma modules of \quote{chain} type (see~\cite[Figure~1]{Kytola-Ridout:On_staggered_indecomposable_Virasoro_modules}, 
and~\cite{Feigin-Fuchs:Verma_modules_over_Virasoro_book, Iohara-Koga:Representation_theory_of_Virasoro} for details), which complicates some of the arguments slightly.} 
(see also~\cite{LPR:Fused_Specht_polynomials_and_c_equals_1_degenerate_conformal_blocks}
for $c=1$).

\begin{figure}
\raisebox{-0.07\textwidth}{
\includegraphics[width=0.19\textwidth]{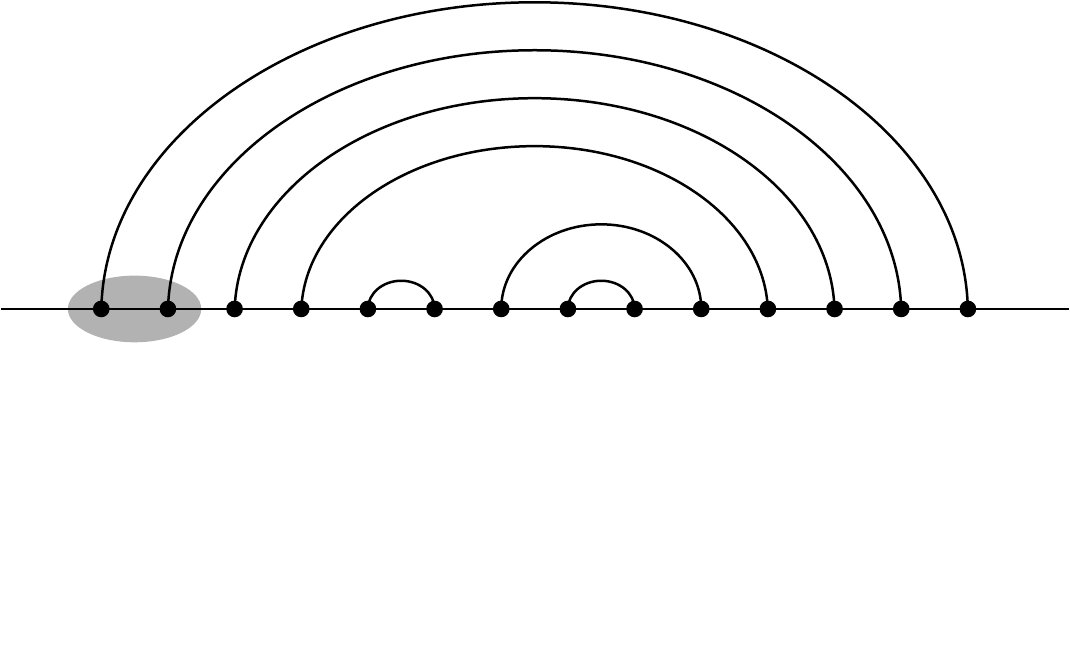}
}%
$\rightsquigarrow$
\raisebox{-0.07\textwidth}{
\includegraphics[width=0.19\textwidth]{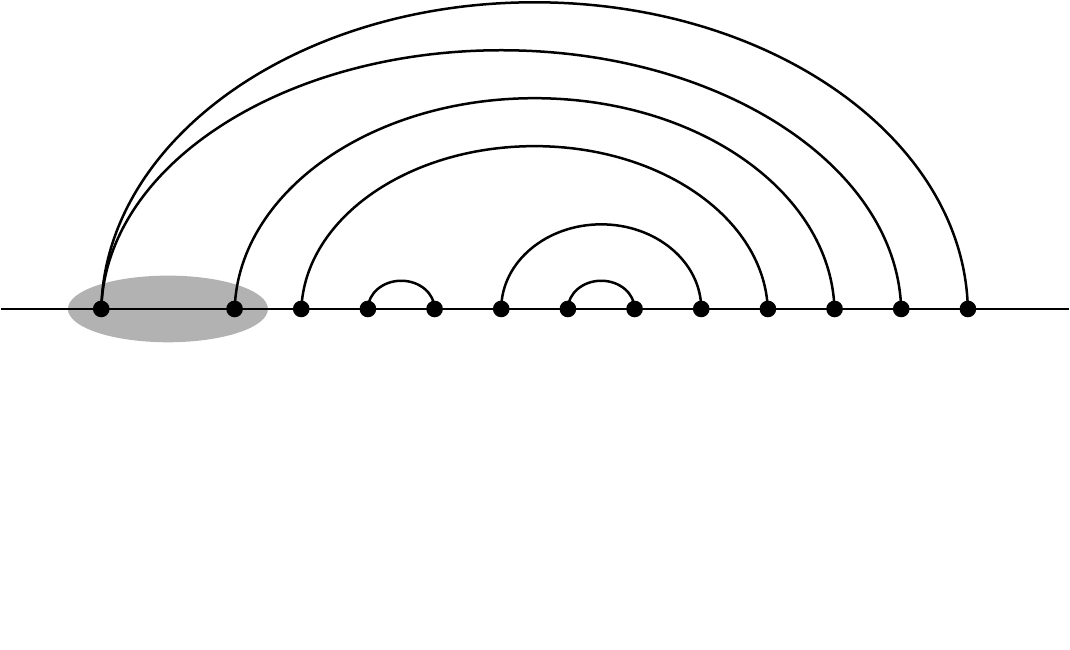}
}%
$\rightsquigarrow \ldots \rightsquigarrow$
\raisebox{-0.07\textwidth}{
\includegraphics[width=0.19\textwidth]{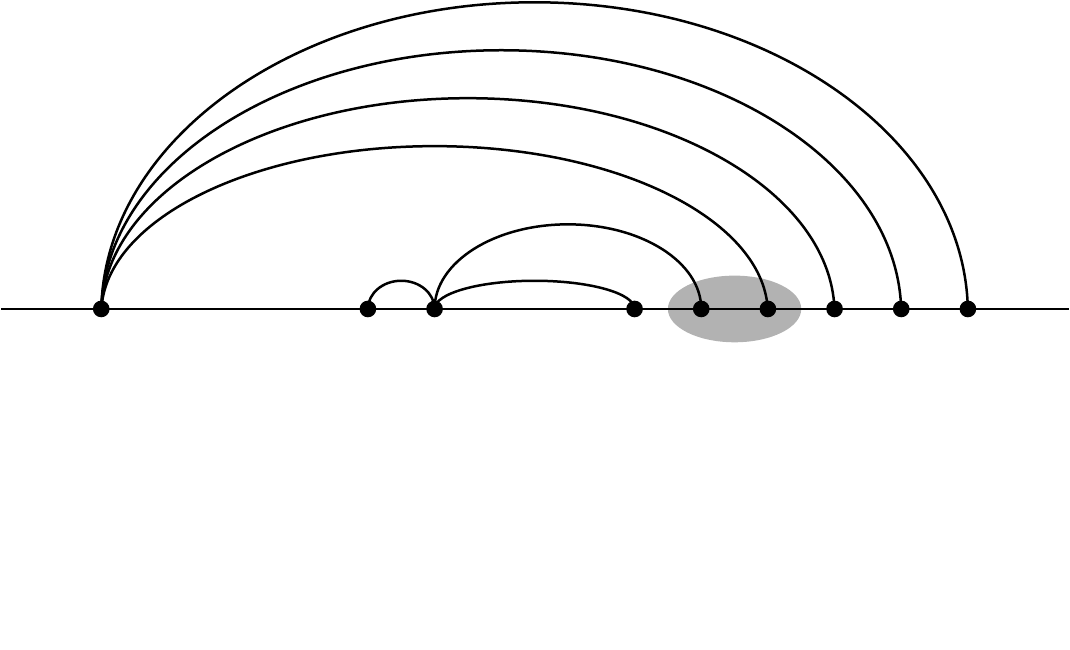}
}%
$\rightsquigarrow$
\raisebox{-0.07\textwidth}{
\includegraphics[width=0.19\textwidth]{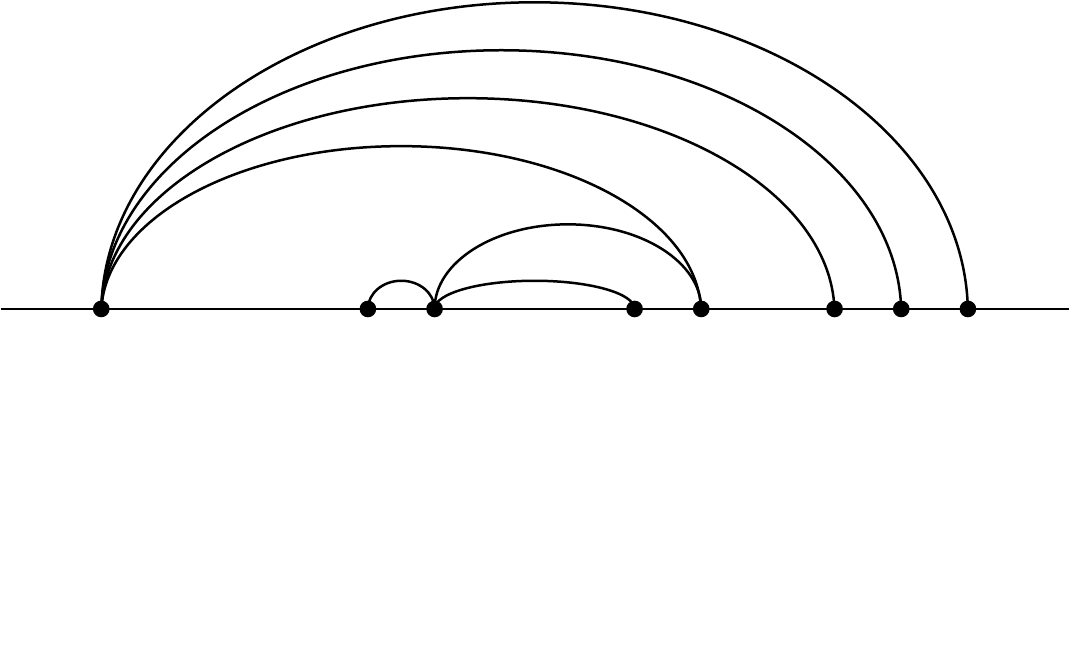}
}
\caption{
\label{fig:fusion}
An illustration of the the iterated limits that by Theorem~\ref{thm:CFT properties} \textnormal{(FUS)} produce $\PartF_\alpha$ with \quote{increasingly fused} valenced patterns $\alpha$, all corresponding to the same \quote{unfused pattern} $\imath(\alpha)$. Parts \textnormal{(PDE)}, \textnormal{(POS)} and \textnormal{(LIN)} of Theorem~\ref{thm:CFT properties} are proven inductively over such increasing fusion.
}
\end{figure}

We complement the above theorem with two further results. 
First, we derive the asymptotic behavior for the functions $\PartF_\alpha$ of~\eqref{eq:fused partition function}, 
which is very natural in light of CFT fusion rules --- see Theorem~\ref{thm:ASY} in Section~\ref{sec:ASY}. 
While thematically closely connected, the statement takes some more notation and the proof is somewhat separate. 
Theorem~\ref{thm:ASY} is an analogue of~\cite[Theorem~5.3~(ASY)]{Peltola:Basis_for_solutions_of_BSA_PDEs_with_particular_asymptotic_properties}. 
Second, in Theorem~\ref{thm:CFT properties coblo} in Appendix~\ref{app:coblo_det} 
we will give an alternative, completely explicit basis for the covariant solution space spanned by $\{ \PartF_\alpha \; | \; \alpha \in \LP_\multii \}$ of the BPZ PDEs (recall that for $\PartF_\alpha$, no closed-form expression appears to be known).

\subsection{Algebraic results: Temperley-Lieb representations}

Lastly, in Section~\ref{sec:algebra} we will investigate in detail algebraic properties of the space of correlation functions defined by
\begin{align*}
%\SolSp_\multii 
\textnormal{\gls{symb:SolSp}}
:= \textnormal{span}_\bC \{\PartF_\alpha \; | \; \alpha \in \LP_\multii \}.
\end{align*}
It turns out that the algebraic properties of the space $\SolSp_\multii$ are naturally inherited from those of the unfused space $\SolSp_{(1, \ldots, 1)} =: \SolSp\sub{1^{2N}}$ by a certain fusion procedure.

Let us first discuss the unfused case $\multii=(1, \ldots, 1)$. 
Our first result (Proposition~\ref{prop:Lascoux}) proves that the space $\SolSp_{(1, \ldots, 1)}$ carries an action of the symmetric group on $2N$ sites, 
whose action consists of simply permuting variables. 
The idea of the proof utilizes the fact that
\begin{align*}
\SolSp\sub{1^{2N}}
= \textnormal{span}_\bC \{\PartF_\alpha \; | \; \alpha \in \LP_N \} 
= \textnormal{span}_\bC \{\Delta_\alpha^\mathfrak{K} \; | \; \alpha \in \LP_N \},
\end{align*}
and that the Fomin type determinants $\{\Delta_\alpha^\mathfrak{K} \; | \; \alpha \in \LP_N\}$ 
form a basis for this space. 
Then, rewriting the determinants $\Delta_\alpha^\mathfrak{K}$ as minors of a symmetric $2N\times 2N$ matrix enables us to apply a result of Lascoux~\cite{Lascoux:Pfaffians_and_representations_of_the_symmetric_group}, leading to Proposition~\ref{prop:Lascoux}. Next, we prove that this symmetric group representation descends to a certain quotient of (the group algebra of) the symmetric group, 
called \emph{Temperley-Lieb algebra} 
%$\TL_{2N}(\fugacity)$ 
\gls{symb:TL} 
with fugacity parameter $\fugacity=-2$ (see Proposition~\ref{prop:descendtoTL}).  
It is well-known that $\TL_{2N}(\fugacity)$ 
is Schur-Weyl dual to the action of the quantum group $U_q(\mathfrak{sl}_2)$ on 
the tensor product of $2N$ fundamental type-one $U_q(\mathfrak{sl}_2)$-modules~\cite{Jimbo:q_analog_of_UqglN_Hecke_algebra_and_YBE,Dipper-James:Q_Schur_algebra},
where $\fugacity=-q-q^{-1}$ and $q = e^{4 \pi \ii / \kappa}$ with $\kappa > 0$, when $q$ is not a root of unity. 
Hence, our findings are somewhat analogous to those in~\cite{Kytola-Peltola:Pure_partition_functions_of_multiple_SLEs,Kytola-Peltola:Conformally_covariant_boundary_correlation_functions_with_quantum_group} (and references therein),
where a quantum group symmetry on an analogous solution space was investigated  
(with central charge $c = (3\kappa-8)(6-\kappa)/2\kappa$ and $\kappa \in (0,8)$ irrational).

In the fused setting, that is, for general $\multii = (s_1,\ldots,s_d) \in \bZpos^d$, 
we obtain similar results except that the role of the Temperley-Lieb algebra is replaced by the \emph{valenced} Temperley-Lieb algebra. 
It is an algebra of the form $\idpt \TL_{2N}(\fugacity) \idpt$ with unit $\idpt$, 
where $\idpt$ is a certain idempotent given by a product of \emph{Jones-Wenzl idempotents}, 
one for each group of points to be fused (see \eqref{eq:def_idpt} for a precise definition with $\fugacity = -2$). 
When $q$ is not a root of unity, this algebra is Schur-Weyl dual to $U_q(\mathfrak{sl}_2)$ on a tensor product of simple type-one modules, 
where the $i$:th tensorand has dimension $s_i+1$ for each $i$ (see~\cite{Flores-Peltola:Higher_spin_QSW}). 
We shall denote by 
%$\vTL_\multii 
\gls{symb:vTL}  $:= \idpt \TL_{2N}(-2) \idpt$ 
the valenced Temperley-Lieb algebra with fugacity parameter $\fugacity = -2$, 
(which is Schur-Weyl dual to $\mathfrak{sl}_2$, with $q=1$)
--- see Definition~\ref{def:vTL}. 

The valenced Temperley-Lieb algebra is a cellular algebra in the sense of~\cite{Graham-Lehrer:Cellular_algebras}
(this has been verified, e.g., in~\cite{Flores-Peltola:Generators_projectors_and_the_JW_algebra}, for an isomorphic algebra called  \quote{{J}ones-{W}enzl algebra}).
Its finite-dimensional representations have been completely classified (see~\cite{Flores-Peltola:Standard_modules_radicals_and_the_valenced_TL_algebra} and references therein). 
The special case $\vTL_\multii$ with fugacity parameter $\fugacity = -2$ (that is the case appearing in the present article) 
is closely related to the fused Hecke algebra~\cite{Crampe-Poulain-d-Andecy:Fused_braids_and_centralisers_of_tensor_representations_of_Uq_gln}. 

\begin{theorem} \label{thm:vTLrep} 
The space $\SolSp_\multii$ is a simple $\vTL_\multii$-module, with explicit action in Corollary~\ref{cor:corollaryvTL}.
\end{theorem}

The proof has two steps: 
we first prove that $\idpt(\SolSp\sub{1^{2N}})$ is a \emph{simple} $\vTL_\multii$-module, that is, it is non-zero and does not have any non-trivial irreducible submodules (Proposition~\ref{prop:5p12}). 
Then, we construct a linear isomorphism $\psi \colon \idpt (\SolSp\sub{1^{2N}}) \to \SolSp_\multii$ (Proposition~\ref{prop:5p15}). 
We finally use the isomorphism $\psi$ to transport the action of $\vTL_\multii$ on $\SolSp_\multii$, leading to Corollary~\ref{cor:corollaryvTL}. 
Note that this action is natural in the sense that it commutes with fusion, see Remark~\ref{rem:commute}.

As discussed above, elements of $\mathcal S_\varsigma$ may be thought of as conformal blocks of degenerate primary fields in a boundary CFT with central charge $c=-2$. 
In the spirit of the conformal boostrap approach \`a la BPZ (and assuming for simplicity that the theory is diagonalizable), we expect that correlation functions of such fields in the \emph{bulk} CFT take the form
\begin{align} \label{eq:bulkcorr}
\sum_{\alpha, \beta \in \text{LP}_\varsigma} c_{\alpha, \beta} \, \PartF_{\alpha}(z_1,\dots,z_d) \, \PartF_{\beta}(\bar{z}_1, \dots, \bar{z}_d),
\end{align}
for $z_1, \dots, z_d \in \mathbb C$, where the \emph{structure constants} $c_{\alpha, \beta}$ are to be determined. 
These constants $c_{\alpha, \beta}$ are constrained by the postulate that bulk correlation functions must be single-valued, i.e., monodromy-invariant. 
Our Theorem \ref{thm:vTLrep} is the first step towards constructing such a monodromy-invariant quantity in the following sense. 
To impose invariance under monodromy transformations, one has to compute the action the pure braid group on the $\PartF_{\alpha}$ basis. 
(Indeed, the pure braid group is the fundamental group of the $d$ times punctured sphere.)
This action, in turn, descends to an action of the valenced Temperley-Lieb algebra $\vTL_\multii$. 
In particular, it is now well-known that there exists a distinguished basis --- dubbed the $\emph{web}$ basis --- in which the pure braid group action 
can be computed diagrammatically using Kauffmann calculus~\cite{Kauffman:State_models_and_the_Jones_polynomial, Kauffman-Lins:TL_recoupling_theory_and_invariants_of_3_manifolds}.
Therefore, relating the basis $\{\PartF_\alpha \; | \; \alpha \in \LP_\multii \}$ to the web basis is an interesting open problem: 
this would make it possible to compute explicitly the monodromy of \eqref{eq:bulkcorr}, and thereby to find the (hopefully unique) set of structure constants $c_{\alpha, \beta}$ that renders it monodromy invariant
(see also~\cite{Flores-Peltola:Standard_modules_radicals_and_the_valenced_TL_algebra} and references therein).

\medskip

\input{tex/acknowledgments}

%% file: tex/acknowledgments.tex
\textbf{Acknowledgments.}

\begin{itemize}[leftmargin=1em] 
\item A.K. is supported by the Academy of Finland grant number 339515.

\item A.L. is supported by the Academy of Finland grant number 340461 \quote{Conformal invariance in planar random geometry.}

\item J.R. is supported by the Academy of Finland Centre of Excellence Programme grant number 346315 \quote{Finnish centre of excellence in Randomness and STructures} (FiRST).

\item This material is part of a project that has received funding from the  European Research Council (ERC) under the European Union's Horizon 2020 research and innovation programme (101042460): 
ERC Starting grant \quote{Interplay of structures in conformal and universal random geometry} (ISCoURaGe) 
and from the Academy of Finland grant number 340461 \quote{Conformal invariance in planar random geometry}.
E.P.~is also supported by 
the Academy of Finland Centre of Excellence Programme grant number 346315 \quote{Finnish centre of excellence in Randomness and STructures} (FiRST) 
and by the Deutsche Forschungsgemeinschaft (DFG, German Research Foundation) under Germany's Excellence Strategy EXC-2047/1-390685813, 
as well as the DFG collaborative research centre \quote{The mathematics of emerging effects} CRC-1060/211504053.

\item We would like to thank the anonymous referees for their comments, which greatly helped to improve this article. 

%E.P. is affiliated at Aalto University, Espoo, Finland and at the University of Bonn, Germany.

\end{itemize}

%% file: tex/sec2-Fomin.tex
The purpose of this section is to derive, for the wired UST model introduced in Section~\ref{subsec:UST framework}, an explicit determinantal formula for the probabilities $\PR[\event(\alpha)]$ of the events~\eqref{eq: event of interest}. 
Part of the analysis in this section builds on the earlier work~\cite{KKP:Boundary_correlations_in_planar_LERW_and_UST}, 
and we omit some details that have been extensively explained there. 
Let us also mention that, while our proof methods are relatively specific to the UST model, the combinatorial content of these arguments 
in fact reflects a more general phenomenon, as observed in~\cite{KKP:Conformal_blocks_q_combinatorics_and_quantum_group_symmetry}.

Throughout this section, we let $\sG = (\sV, \sE)$ be a finite, connected, properly embedded planar graph,
with distinct boundary edges $e_1, \ldots, e_{2N}$ in counterclockwise order along the boundary. 
(The present section is combinatorial, and no square grid structure is needed.)

\begin{theorem} \label{thm:pinched pertition functions}
Fix valences $\multii = (s_1,\ldots,s_d) \in \bZpos^d$ such that $s_1+\cdots+s_d = 2N$. Then, we have 
\begin{align} \label{eq:pinched pertition functions}
\PR[\event(\alpha)]
= \; & \sum_{ \substack{\beta \in \LP_N \\ \beta \DPgeq \alpha}} \# \CItilingsof (\alpha / \beta) \, \Delta^{\mathfrak{K}}_\beta , \qquad \alpha \in \LP_\multii ,
\end{align} 
where the notations $\DPgeq$ and
$\# \CItilingsof (\alpha / \beta)$ are the same as in Theorem~\ref{thm:scaling limit of pinched pertition functions}, and $\Delta^{\mathfrak{K}}_\beta$
is a Fomin type determinant (defined in Section~\ref{subsec:inv Fomin}) with kernel $\mathfrak{K}$ given by
\begin{align*}
\mathfrak{K}(a,b) = (D^{\mathrm{tan}}_1)^{m_a-1} (D^{\mathrm{tan}}_2)^{m_b-1} \ExcK_{\sG}(e_{\summ_{i-1} + 1}, e_{\summ_{j-1} + 1}) , 
\quad 
a = \summ_{i-1}+m_a \; \textnormal{ and } \; b = \summ_{j-1}+m_b ,
\end{align*}
if $\summ_{i-1} <  a \leq \summ_i$  and  $\summ_{j-1} < b \leq \summ_j$ for some $i \neq j$, and
$\mathfrak{K}(a,b) = 0$ if $\summ_{j-1} < a, b \leq \summ_j$ for some $1 \leq j \leq d$. 
Here, $\ExcK_{\sG}$ is the excursion kernel of the symmetric random walk on $\sG$ (defined in Section~\ref{subsec:connection probabilities preli}),
and $D^{\mathrm{tan}}_1$ and $D^{\mathrm{tan}}_2$ are the discrete counterclockwise tangential differences 
on the marked boundary edges (defined in Equation~\eqref{eq:tan derivative} in Section~\ref{subsec:connection probabilities Fusion}) with respect to the first and second argument, respectively.
\end{theorem}

\subsection{General inverse Fomin type sums}
\label{subsec:inv Fomin}

As hinted by the similarity --- and generality --- of both 
Theorems~\ref{thm:scaling limit of pinched pertition functions}~\&~\ref{thm:pinched pertition functions}, it will be useful to study expressions of a particular determinantal form 
(see also Appendix~\ref{app:explicit determinants} for examples). 
Let %$\mathfrak{K} \colon \{1,\ldots,2N\} \times \{1,\ldots,2N \} \to \bC$
\gls{symb:kernel}$\colon \{1,\ldots,2N\} \times \{1,\ldots,2N \} \to \bC$
be a symmetric kernel, i.e., 
$\mathfrak{K}(a,b) = \mathfrak{K}(b,a)$ for all\footnote{Throughout, the values of the diagonal entries of $\mathfrak{K}$ actually need not to be specified.} $a \neq b$. 
Let 
$\beta \in \LP_N$ be an $N$-link pattern.
We will assume throughout that the link patterns $\beta = \{ \link{a_1\;}{\;b_1\,}, \ldots, \link{a_N\;}{\;b_N\,} \}$ are endowed with the \emph{left-to-right orientation} $a_1 < \cdots < a_N$ and $a_k < b_k$ for all $k$. 
We then define the determinant of $\beta$ with kernel $\mathfrak{K}$ as
\begin{align}
\label{eq: LPdet}
%\Delta^{\mathfrak{K}}_\beta 
\textnormal{\gls{symb:FomDet}} 
:= \det \big( \mathfrak{K}(a_k, b_\ell) \big)_{k,\ell=1}^N 
, \qquad 
\beta = \{ \linkInEquation{a_1}{b_1}, \ldots, \linkInEquation{a_N}{b_N} \} \in \LP_N .
\end{align}
We set
\begin{align}\label{eq: inverse Fomin-type sum}
%\mathfrak{Z}^{\mathfrak{K}}_\alpha 
\textnormal{\gls{symb:Fomintypesum}} 
:= \sum_{ \substack{\beta \in \LP_N \\ \beta \DPgeq \alpha}} \# \CItilingsof (\alpha / \beta) \, \Delta^{\mathfrak{K}}_\beta , \qquad \alpha \in \LP_N ,
\end{align}
where %$\# \CItilingsof (\alpha / \beta)$ 
\gls{symb:CItilings} 
is the number of cover-inclusive Dyck tilings
of the skew Young diagram~$\alpha/\beta$
(as detailed in Definition~2.8 of~\cite{KKP:Boundary_correlations_in_planar_LERW_and_UST}).
We call $\mathfrak{Z}^{\mathfrak{K}}_\alpha$ the \emph{inverse Fomin type sum} associated to $\alpha$, with the kernel $\mathfrak{K}$. 
Note that $\Delta^{\mathfrak{K}}_\beta$ and $\mathfrak{Z}^{\mathfrak{K}}_\alpha$ are polynomials in the kernel entries 
$\{ \mathfrak{K}(a,b) \;|\; 1 \leq a < b \leq 2N\}$. 
The following properties are readily checked:
\begin{enumerate}[leftmargin=2em,label=\textnormal{(\alph*):}, ref=(\alph*)]
\item Each term of these polynomials is a product of $N$ kernel entries $\mathfrak{K}(\cdot, \cdot)$ 
--- furthermore so that every index $1\leq a \leq 2N$ appears exactly once as an argument of the kernel $\mathfrak{K}$. 

\item
Let $\mathfrak{K}(b, \cdot) = \mathfrak{K}(\cdot,b)$
stand for the collection of kernel entries 
that involve the index~$b$, i.e.,
$\big( \mathfrak{K}(1,b), \mathfrak{K}(2,b), \ldots, \mathfrak{K}(b-1,b) , \mathfrak{K}(b,b+1), \ldots, \mathfrak{K}(b,2N) \big)$. 
Then, for any $b$, both $\Delta^{\mathfrak{K}}_\beta$ and $\mathfrak{Z}^{\mathfrak{K}}_\alpha$
are linear in
$ \mathfrak{K}(b, \cdot)$, e.g., of the form
\begin{align*}
\Delta^{\mathfrak{K}}_\beta = \sum_{\substack{1 \leq a \leq 2N \\ a \neq b}} \big[ \Delta^{\mathfrak{K}}_\beta \big]_{a,b} \, \mathfrak{K}(a,b) ,
\end{align*}
where $[ \Delta^{\mathfrak{K}}_\beta ]_{a,b}$ is 
a polynomial in the variables other than $\mathfrak{K}(a,\cdot)$
and $\mathfrak{K}(\cdot,b)$.

\item
Observe also that
$[ \Delta^{\mathfrak{K}}_\beta ]_{a,b} = [ \Delta^{\mathfrak{K}}_\beta ]_{b,a}$.
\end{enumerate}

For $k \neq \ell$, let %$\mathfrak{K}^{(k, \ell)}$ 
\gls{symb:exchangker}
be the symmetric kernel obtained by exchanging
the kernel entries at $k$ and $\ell$ in the kernel $\mathfrak{K}$, i.e.,
\begin{align*}
\begin{cases}
\mathfrak{K}^{(k, \ell)}(k, a) = \mathfrak{K}^{(k, \ell)}(a,k)= \mathfrak{K} (\ell, a) & \textnormal{and} \quad \mathfrak{K}^{(k, \ell)}(\ell, a) = \mathfrak{K}^{(k, \ell)}(a, \ell) = \mathfrak{K} (k, a), \textnormal{ all } a \not \in \{k, \ell \} \\
\mathfrak{K}^{(k, \ell)}(a, b) = \mathfrak{K}(a, b), & \textnormal{for all other pairs $(a, b)$}.
\end{cases}
\end{align*}
We will use repeatedly a technical argument applying such a kernel entry exchange; 
in~particular, let us recall~\cite[Proposition~2.18]{KKP:Boundary_correlations_in_planar_LERW_and_UST}: 
\begin{enumerate}[leftmargin=2em,label=\textnormal{(d):}, ref=(d)]
\item Suppose that $\link{j}{j+1} \notin \alpha$. 
Then,
we have
\begin{align} \label{eq:simple anti-symmetry of inv Fomin}
\mathfrak{Z}^{\mathfrak{K}^{(j, j+1)}}_\alpha = - \mathfrak{Z}^{\mathfrak{K}}_\alpha.
\end{align}
\end{enumerate}
Taking this property --- and the above simple observations --- as a black box, 
we can now prove the key result (Lemma~\ref{lem:Fomin properties}) needed for the limit analysis in the present work.

\begin{lemma} \label{lem:Fomin properties}
Fix valences 
$\multii = (s_1,\ldots,s_{d})$, a valenced link pattern $\alpha \in \LP_\multii$, and an index $1 \leq j \leq d$. 
Then, the following hold true:
\begin{enumerate}[leftmargin=3em,label=\textnormal{(\alph*):}, ref=(\alph*)]
\item \label{item:general anti-symmetry of inv Fomin}
For any $k, \ell \in \{ \summ_{j-1}+1, \ldots, \summ_j \}$, we have
\begin{align*}
\mathfrak{Z}^{\mathfrak{K}^{(k, \ell)}}_\alpha = - \mathfrak{Z}^{\mathfrak{K}}_\alpha.
\end{align*}

\medskip

\item \label{item:general zero-replacing rule a}
If for some $k, \ell \in \{ \summ_{j-1}+1, \ldots, \summ_j \} $ it holds that $\mathfrak{K} (\ell, a) = \mathfrak{K} (k, a)$ for all $a \not \in \{k, \ell \}$, then we have $\mathfrak{Z}^{\mathfrak{K}}_\alpha = 0$.

\medskip

\item \label{item:general zero-replacing rule b}
For any $k, \ell \in \{ \summ_{j-1}+1, \ldots, \summ_j \}$, the coefficient of $\mathfrak{K}(k, \ell)$ in $\mathfrak{Z}^{\mathfrak{K}}_\alpha$ is zero: 
$[ \mathfrak{Z}^{\mathfrak{K}}_\alpha ]_{k, \ell} = 0$.
In~particular, if we define the symmetric kernel $\mathfrak{K}_0$ as 
\begin{align*}
\begin{cases}
{\mathfrak{K}_0}(k, \ell) = 0 , & \textnormal{if } k, \ell \in \{ \summ_{j-1}+1, \ldots, \summ_j \} \\
{\mathfrak{K}_0}(a,b) = \mathfrak{K}(a,b) , & \textnormal{for all other pairs $(a,b)$},
\end{cases}
\end{align*}
then we have 
\begin{align*}
\mathfrak{Z}^{{\mathfrak{K}_0}}_\alpha = \mathfrak{Z}^{\mathfrak{K}}_\alpha . 
\end{align*}
\end{enumerate}
\end{lemma}

\begin{proof}
It is a basic property of permutations of finite sequences 
(here, consisting of the elements $\{ \summ_{j-1}+1, \ldots, \summ_j \}$) 
that any transposition of two elements (here, $(k, \ell)$) 
can be obtained by composing an odd number of transpositions of two elements that are neighboring in the sequence right before (and right after) that transposition. 
Item~\ref{item:general anti-symmetry of inv Fomin} then follows by applying Equation~\eqref{eq:simple anti-symmetry of inv Fomin} an odd number of times.

For Item~\ref{item:general zero-replacing rule a}, note that
if $\mathfrak{K} (\ell, a) = \mathfrak{K} (k, a)$ for all $a \not \in \{k, \ell \}$, then $\mathfrak{K}^{(k,\ell ) } = \mathfrak{K}$. 
Hence, by Item~\ref{item:general anti-symmetry of inv Fomin}, we see that 
$\mathfrak{Z}^{{\mathfrak{K}}}_\alpha = \mathfrak{Z}^{{\mathfrak{K}^{(k,\ell ) }}}_\alpha = - \mathfrak{Z}^{{\mathfrak{K}}}_\alpha$, which implies that $\mathfrak{Z}^{{\mathfrak{K}}}_\alpha = 0$.

To prove Item~\ref{item:general zero-replacing rule b}, 
define the modified kernel $\tilde{\mathfrak{K}}$ by setting 
$\tilde{\mathfrak{K}}(a, k) = \tilde{\mathfrak{K}}(k, a) = \delta_{a, \ell}$ and $\tilde{\mathfrak{K}}(a, \ell ) = \tilde{\mathfrak{K}}(\ell, a) = \delta_{a, k}$ and $\tilde{\mathfrak{K}}(a,b) = \mathfrak{K}(a, b)$ whenever $a, b \not \in \{ k, \ell\}$. 
The, on the one hand, by the first two observations above, it follows that
\begin{align*}
\mathfrak{Z}^{\tilde{\mathfrak{K}}}_\alpha = [ \mathfrak{Z}^{\mathfrak{K}}_\alpha ]_{k, \ell}.
\end{align*}
On the other hand, by Item~\ref{item:general zero-replacing rule a} we have $\mathfrak{Z}^{\tilde{\mathfrak{K}}}_\alpha = 0$.
This proves the first part of Item~\ref{item:general zero-replacing rule b}. The second part follows by the first two elementary properties. 
\end{proof}

\subsection{Discrete connection probabilities without fusion}
\label{subsec:connection probabilities preli}

The probability $\PR[\event(\alpha)]$ has a well-known formula in terms of inverse Fomin type sums and random-walk Green's functions~\cite{Fomin:LERW_and_total_positivity, 
Kenyon-Wilson:Boundary_partitions_in_trees_and_dimers,
Kenyon-Wilson:Double_dimer_pairings_and_skew_Young_diagrams,
KKP:Boundary_correlations_in_planar_LERW_and_UST}, which we review in what follows.

For the formulas, we introduce some further notation and terminology. 
We define the \emph{random-walk weight} %$\mathsf{w} (\walk)$ 
\gls{symb:RWweight} 
of a nearest-neighbor walk $\walk = (v_0, v_1, \ldots, v_m)$ on $\sG=(\sV, \sE)$ as
\begin{align*}
\mathsf{w} (\walk) := \frac{1}{ \prod_{k=0}^m \deg (v_k) } ,
\end{align*}
where $\deg (v)$ is the degree of the vertex $v \in \sV$ 
(also the terminal vertex is included in the product).
We define the \emph{random-walk Green's function} (excursion kernel) with Dirichlet boundary conditions as
\begin{align*}
%\GreenK_{\sG} 
\textnormal{\gls{symb:RWGF}}(v, w) := \sum_{\walk \in \Walks_\circ (v, w) } \mathsf{w} (\walk),
\end{align*}
where $\Walks_\circ (v, w) $ denotes the set of nearest-neighbor walks from $v$ to $w$ on $\sG$ that only contain interior vertices (allowing the one-vertex walk if $v=w$). 
Clearly, $\GreenK_{\sG} (v, w) = \GreenK_{\sG} (w, v)$, and by a standard first-step argument $\GreenK_{\sG} (v, w)$ indeed is the Green's function of 
the discrete Laplacian\footnote{\emph{Discrete Laplacian} $\Delta^\#$ sends a function 
$f \colon \sV \to \bR$ on the vertices to another function $\Delta^\# f \colon \sV \to \bR$ given by
\begin{align*}
(\Delta^\# f) (v) = 
\sum_{\langle v , u \rangle \in \sE} (f(u)-f(v)) , \qquad v \in \sV.
\end{align*}
Now, if $f(v) = \GreenK_{\sG}(v, w)$ (with fixed $w \in \sV_\circ$) then $(\Delta^\# f)(v) = - \delta_{v, w}$ for $v \in \sV_\circ$, and $f(v) = 0$ for $v \in \sV_\partial$.} 
(with Dirichlet boundary conditions).

Lastly, we define the \emph{excursion kernel} %$\ExcK_{\sG}$ 
\gls{symb:RWEK} 
of the symmetric random walk on $\sG$ between two boundary edges $e_1$ and $e_2$ as the Green's function between the respective interior vertices $v_1$ and $v_2$ of the edges $e_1$ and $e_2$, i.e., 
$\ExcK_{\sG}(e_1,e_2) := \GreenK_{\sG}(v_1, v_2)$.

\begin{citedtheorem}[{See~\cite[Theorem~3.12 and Section~3.6]{KKP:Boundary_correlations_in_planar_LERW_and_UST} for this formulation}]
\label{thm: disc connectivity probas}
For the wired UST on any finite, connected, properly embedded planar graph $\sG$, with marked boundary edges $e_1, \ldots, e_{2N}$ in counterclockwise order, and for any $N$-link pattern $\alpha \in \LP_N$, we have
\begin{align} \label{eq: soln of disc ptt fcns}
\PR[\event(\alpha)] \; = \mathfrak{Z}^{\mathfrak{K}}_\alpha \; = \sum_{ \substack{\beta \in \LP_N \\ \beta \DPgeq \alpha}} \# \CItilingsof (\alpha / \beta) \, \Delta^{\mathfrak{K}}_\beta , 
\end{align} 
where the kernel is $\mathfrak{K}(a,b) = \ExcK_{\sG}(e_a, e_b)$.
\end{citedtheorem}

While this purely combinatorial formula holds for any planar graph, 
in the scaling limit, if some of the marked boundary edges tend together, the formula in Theorem~\ref{thm: disc connectivity probas} is 
not directly amenable to analysis.
Indeed, due to the discrete fusion of points, if one naively tries to apply the renormalization that would give convergence of discrete excursion kernels to the Brownian ones, 
some kernels $\ExcK_{\sG}$ blow up in the scaling limit; while in some determinants, 
e.g., two rows $\ExcK_{\sG}(\cdot, v_b)$ and $\ExcK_{\sG}(\cdot , v_{b'})$, with $b \neq b'$, 
may become identical in the scaling limit. 
We will use tools from Section~\ref{subsec:inv Fomin} to introduce an algorithmic sequence of modifications in symmetric kernels $\mathfrak{K}$ which will not change the value of $\mathfrak{Z}^{\mathfrak{K}}_\alpha$ 
but which will lead to a more tractable formula, that we then use to establish the convergence to the scaling limit.

\subsection{Fusion of discrete connection probabilities}
\label{subsec:connection probabilities Fusion}

Given some counterclockwise ordered boundary edges $e_1, \ldots, e_n$ and a function $f$ defined on them, we denote
\begin{align} \label{eq:tan derivative}
%D^{\mathrm{tan}} 
\textnormal{\gls{symb:tan}} f (e_i) = f(e_{i+1}) - f(e_i), \qquad 1 \leq i \leq n-1.
\end{align}
Note that iteratively, for example $(D^{\mathrm{tan}})^2 f (e_i)$ is defined for $1 \leq i \leq n-2$. 

\begin{remark}
We will ultimately be interested in the case where $e_{\summ_{j-1}+1}, \ldots, e_{\summ_{j}}$ are a chain of neighbors on the lattice $\tfrac{1}{\delta} \bZ^2$ and $\mathfrak{K}(a, b) = \ExcK_\sG (e_a, e_b) $ coincides with a function defined on all boundary edges. Then, $\frac{1}{\delta} D^{\mathrm{tan}}$ is the natural discretization of the tangential derivative. 
\end{remark}

\begin{proposition}
\label{prop:higher-valency solution}
Fix $\multii = (s_1,\ldots,s_d) \in \bZpos^d$ and $1 \leq j \leq d$.
We interpret below $\mathfrak{K}(\cdot,\cdot)$ as a function on 
$\{ e_1, \ldots, e_{2N} \} \times \{ e_1, \ldots, e_{2N} \}$ 
(instead of $\{ 1, \ldots, 2N \} \times \{ 1, \ldots, 2N \}$).
Then, for any $\alpha \in \LP_\multii$, 
the following algorithm of modifications in $\mathfrak{K}(\cdot,\cdot)$ leaves $\mathfrak{Z}^{\mathfrak{K}}_\alpha$ defined in~\eqref{eq: inverse Fomin-type sum} invariant:
\begin{enumerate}[leftmargin=3em,label=\textnormal{(\roman*):}, ref=(\roman*)]
\item \label{step: modification 1}
set $\mathfrak{K}(a,b) = 0$ for all $a,b \in \{ \summ_{j-1}+1 , \ldots, \summ_j \}$, 
and then,

\medskip

\item \label{step: modification 2} 
for all $b \not \in \{ \summ_{j-1}+1 , \ldots, \summ_j \} $ and for all $m \in \{2,\ldots,s_j\}$, 
set both $\mathfrak{K}(\summ_{j-1}+m,b)$ and $\mathfrak{K}(b,\summ_{j-1}+m)$ to equal 
\begin{align*}
(D^{\mathrm{tan}}_1 )^{m-1} \mathfrak{K}(\summ_{j-1}+1,b)
\qquad \textnormal{instead of} \qquad
\mathfrak{K}(\summ_{j-1}+m,b) ,
\end{align*} 
where $D^{\mathrm{tan}}_1 $ means $D^{\mathrm{tan}} $ with respect to the first argument.
\end{enumerate}
\end{proposition}

\begin{proof}
Rule~\ref{step: modification 1} leaves $\mathfrak{Z}^{\mathfrak{K}}_\alpha$ invariant
by a direct application of Lemma~\ref{lem:Fomin properties}~\ref{item:general zero-replacing rule b}, 
so we only need to prove invariance under Rule~\ref{step: modification 2}.
First, from~\eqref{eq:tan derivative}, by induction on $m$ and the base case $(D^{\mathrm{tan}})^0 f =f$, 
one readily deduces that for each $m$, there exist $a_{m,1}, \ldots, a_{m, m-1} \in \bR$ such that
\begin{align*}
(D^{\mathrm{tan}})^{m-1} f(e_{\summ_{j-1}+1}) 
= a_{m,1} \, f(e_{\summ_{j-1}+1}) + \cdots + a_{m, m-1} \, f(e_{\summ_{j-1}+m-1}) + f(e_{\summ_{j-1}+m}) .
\end{align*} 
The matrix corresponding to this linear relation is lower-triangular with ones on the diagonal, and so is its inverse. Thus, for each $m$ there exist $b_{m,1}, \ldots, b_{m, m-1} \in \bR$ so that
\begin{align*}
f(e_{\summ_{j-1}+m}) = b_{m,1} \,  f(e_{\summ_{j-1}+1}) +
\ldots + b_{m, m-1} \,  (D^{\mathrm{tan}})^{m-2} f(e_{\summ_{j-1}+1}) + (D^{\mathrm{tan}})^{m-1} f(e_{\summ_{j-1}+1}).
\end{align*} 
We will use this formula for $f(e_{\summ_{j-1}+m}) = \mathfrak{K}(\summ_{j-1}+m,b)$ for each $b$.

First, when $m=2$, since $\mathfrak{Z}^{\mathfrak{K}}_\alpha$ is linear in $\mathfrak{K}(\summ_{j-1}+2, \cdot)$, it can be split into a linear combination of 
two inverse Fomin type sums where $\mathfrak{K}(\summ_{j-1}+2, \cdot)$ is in the first case given by $\mathfrak{K}(\summ_{j-1}+2, \cdot)= \mathfrak{K}(\summ_{j-1}+1, \cdot) $, 
while the second case has $\mathfrak{K}(\summ_{j-1}+2, \cdot)= D^{\mathrm{tan}}_1 \mathfrak{K}(\summ_{j-1}+1, \cdot) $.
The first case gives a zero inverse Fomin type sum by Lemma~\ref{lem:Fomin properties}~\ref{item:general zero-replacing rule a}, 
while the second case is the claimed modification in Rule~\ref{step: modification 2} 
(and had the coefficient \quote{$1$} in the linear combination).
One then analogously splits $\mathfrak{Z}^{\mathfrak{K}}_\alpha$ (with $\mathfrak{K}(\summ_{j-1}+2, \cdot)$ now modified to $D^{\mathrm{tan}}_1 \mathfrak{K}(\summ_{j-1}+1, \cdot)$) into a linear conbination of three inverse Fomin sums where $\mathfrak{K}(\summ_{j-1}+3, \cdot) $
 is either $\mathfrak{K}(\summ_{j-1}+1, \cdot) $, $D^{\mathrm{tan}_1} \mathfrak{K}(\summ_{j-1}+1, \cdot)$, or $ (D^{\mathrm{tan}}_1)^2 \mathfrak{K}(\summ_{j-1}+1, \cdot)$. 
The first two cases again give zero contribution by
Lemma~\ref{lem:Fomin properties}~\ref{item:general zero-replacing rule a}, 
while the third case is exactly of the claimed form, with coefficient \quote{$1$.} 
Going on in this way throughout the index set $\{ \summ_{j-1}+1 , \ldots, \summ_j \}$, 
we establish Rule~\ref{step: modification 2}.
\end{proof}

\begin{proof}[Proof of Theorem~\ref{thm:pinched pertition functions}]
We begin with the expression~\eqref{eq: soln of disc ptt fcns} for $\PR[\event(\alpha)]$ given by Theorem~\ref{thm: disc connectivity probas}, which is 
an inverse Fomin type sum as in~(\ref{eq: LPdet},~\ref{eq: inverse Fomin-type sum}) 
having kernel $\mathfrak{K}(a,b) = \ExcK_{\sG}(e_a, e_b)$. 
We obtain the sought~\eqref{eq:pinched pertition functions} by applying Proposition~\ref{prop:higher-valency solution} iteratively to all $j \in \{1,2,\ldots,d\}$.
\end{proof}

%% file: tex/sec3-sclim.tex
The purpose of this short section is to use 
Theorem~\ref{thm:pinched pertition functions} and discrete complex analysis to prove our scaling limit result, 
Theorem~\ref{thm:scaling limit of pinched pertition functions},
concerning the UST connection probabilities. 
We also prove the fusion and covariance properties \textnormal{(FUS)}~\&~\textnormal{(COV)} in Theorem~\ref{thm:CFT properties},
which will follow from the explicit limit expression with the tools of the previous Section~\ref{sec:discrete}.

We define the \emph{Green's function} 
%$\GreenKH \colon \bH \times \bH \to \bR$, 
\gls{symb:GreenKH}$\colon \bH \times \bH \to \bR$, 
\emph{Poisson kernel} 
%$\PoissonKH \colon \bH \times \bR \to \bR$, 
\gls{symb:PoissonKH}$\colon \bH \times \bR \to \bR$, 
and \emph{Brownian excursion kernel} 
%$\ExcKH \colon \bR \times \bR \to \bR$ 
\gls{symb:ExcKH}$\colon \bR \times \bR \to \bR$ 
in the upper half-plane $\bH := \{z \in \bC \;|\; \im(z) > 0 \}$ as
\begin{align*}
\GreenKH_\bH (z, w) 
= \; &  \GreenKH (z, w) 
:= - \frac{1}{2 \pi} \, \log \Big| \frac{ z - w }{ z - \overline{w} } \Big| ,  \qquad (z, w)  \in \bH \times \bH , \; z \neq w ,
\\
\PoissonKH_\bH (z, x) 
= \; &  \PoissonKH (z, x) 
:= - \frac{1}{\pi} \, \im \Big( \frac{1}{z-x} \Big) 
= \frac{1}{\pi} \frac{\im (z)}{| z-x |^2} ,  \qquad (z, x)  \in \bH \times \bR ,
\\
\ExcKH_\bH(x, y) 
= \; &  \ExcKH(x, y) 
:= \frac{1}{\pi} \frac{1}{(x-y)^2} ,  \qquad (x, y) \in \bR \times \bR , \; x \neq y ,
\end{align*}
respectively. 
The first two objects are the Green's function and Poisson kernel of the (positive)  
Laplacian on $\bH$ with Dirichlet boundary conditions (with positive eigenvalues). 
The normal derivatives of $\GreenKH$ and $\PoissonKH$ can be defined on $\bR$ by Schwarz reflection, so
\begin{align*}
\PoissonKH (z, x) = \left( \partial_y \GreenKH (z, x + \ii y) \right)_{y=0} 
\qquad \textnormal{and} \qquad
\ExcKH(x_1, x_2) = \left( \partial_y \PoissonKH (x_1 + \ii y, x_2) \right)_{y=0}.
\end{align*}
For any other simply connected domain $\domain$, the Green's function $\GreenKH_\domain (z, w)$ is given by
\begin{align*}
\GreenKH_\domain (z, w) := \GreenKH (\varphi(z),\varphi( w)) ,
\qquad z, w \in \domain ,
\end{align*}
where $\varphi \colon \domain \to \bH$ is any conformal bijection. 
Assuming that the marked boundary points $p_1, p_2 \in \partial \domain$ lie 
on sufficiently regular boundary segments (e.g. $C^{1+\epsilon}$ for some $\epsilon>0$),  
we define the Poisson kernel and the Brownian excursion kernel as normal derivatives, 
\begin{align*}
\PoissonKH_\domain (z, p_1) &:= |\varphi'(p_1) | \, \PoissonKH (\varphi(z), \varphi(p_1)) ,
\qquad z \in \domain ,
\\
\ExcKH_\domain (p_1, p_2) &:= |\varphi'(p_1)| \, |\varphi'(p_2)| \, \ExcKH (\varphi(p_1), \varphi(p_1)) .
\end{align*}

\subsection{Scaling limits of connection probabilities: Proof of Theorem~\ref{thm:scaling limit of pinched pertition functions}}
\label{subsec:proof of pinched pertition functions}

We now get to the first main result of this article. 
Recall the notions of admissible polygons and Carath\'{e}odory convergence from Section~\ref{subsubsec:scaling limit setup}.
Basic discrete harmonic analysis gives the following result.

\begin{theorem} \label{thm:conv of exc kernels and derivatives}
Suppose that admissible $2$-polygons $(\domain^\delta; e_1^\delta, e_2^\delta)$ converge as $\delta \to 0$ to an admissible $2$-polygon $(\domain; p_1, p_2)$ in the Carath\'{e}odory sense.
Then, 
we have
$\ExcK_{\delta}(e_1, e_2) = \delta^2 \ExcKH_\domain (p_1, p_2) + o(\delta^2)$, and more generally, for all
$i_1, \ldots,i_m \in \{ 1,2\}$ and $m \in \bZnn$, we have 
\begin{align*}
D^{\mathrm{tan}}_{i_1} D^{\mathrm{tan}}_{i_2} \cdots D^{\mathrm{tan}}_{i_m} \ExcK_{\delta}(e_1, e_2)
= \; & \delta^{2+m} \partial^{\mathrm{tan}}_{i_1} \partial^{\mathrm{tan}}_{i_2} \cdots \partial^{\mathrm{tan}}_{i_m} \ExcKH_\domain(p_1, p_2) + o(\delta^{2+m}) ,
\qquad \delta \to 0 ,
\end{align*}
where $\ExcK_{\delta} = \ExcK_{\domain^\delta}$ is the excursion kernel of the symmetric random walk on the graph $\domain^\delta$ and $D^{\mathrm{tan}}$ are differences between nearest neighbor boundary edges 
\textnormal{(}defined in Equation~\eqref{eq:tan derivative} in Section~\ref{subsec:connection probabilities Fusion}\textnormal{)}. 
\end{theorem}

\begin{proof}
This follows from well-known convergence results of discrete harmonic functions. 
The discrete Green's $\GreenK_{\domain^\delta}$ functions on $\domain^\delta$ (which are discrete harmonic in both of their variables) converge to their continuous counterparts on $\domain$~\cite{Chelkak-Smirnov:Discrete_complex_analysis_on_isoradial_graphs}.
Moreover, the convergence of discrete harmonic functions and all their discrete derivatives
is uniform on compact subsets of the domain~$\domain$
\cite{CFL:Uber_PDE_der_mathphys, Chelkak-Smirnov:Discrete_complex_analysis_on_isoradial_graphs}. 
By discrete/continuous Schwarz reflection, the convergence of such functions 
and their discrete derivatives holds up to the boundary at straight boundary segments 
(using the polygons' admissibility).
It remains to note that $\ExcK_{\delta}(e_1, e_2)$ coincides the discrete normal derivative (with respect to both arguements) of the discrete Green's function $\GreenK_{\domain^\delta}$ on $\domain^\delta$. 
\end{proof}

\begin{proof}[Proof of Theorem~\ref{thm:scaling limit of pinched pertition functions}]
This is a direct consequence of Theorems~\ref{thm:pinched pertition functions}~\&~\ref{thm:conv of exc kernels and derivatives}. 
Valence $s_i$ produces discrete derivatives of orders $0,1, \ldots, s_i - 1$, 
and they sum up to $(s_i-1)s_i/2$. 
The renormalization factor $\prod_{i=1}^d \delta^{(s_i-1)s_i/2}$ thus divides out the discrete tangential derivatives by the appropriate power of $\delta$.
The factor $\delta^{-2N}$ comes from the scaling of Poisson kernels. 
\end{proof}

\subsection{Proof of property \textnormal{(FUS)} in Theorem~\ref{thm:CFT properties}}
\label{subsec:fusion of partition functions}

For each $\alpha \in \LP_\multii$, 
we denote by $\iota(\alpha)$ the $N$-link pattern obtained by \quote{opening up} the valenced nodes in $\alpha$ (Figure~\ref{fig: UST fused}(top-right)).

\begin{lemma} \label{lem:fusion of fused partition functions}
Let $\domain$ be a simply-connected domain and let $p_1, \ldots, p_{d} \in \partial \domain$ distinct counterclockwise boundary points on sufficiently regular\footnote{E.g. $C^{1+\epsilon}$ for some $\epsilon>0$ --- we need derivatives of conformal maps defined on the boundary.} boundary segments, 
so that $\PartF_\alpha (\domain; p_1, \ldots, p_d)$ can be defined as in Equation~\eqref{eq:fused partition function}.
Then, for each $\alpha \in \LP_\multii$, 
the function $\PartF_\alpha (\domain; p_1, \ldots, p_d)$ equals
\begin{align*} 
%\PartF_\alpha (\domain; p_1, \ldots, p_d) = 
\; & \bigg( \prod_{j=1}^d ( 0! \cdot 1! \cdot \ldots \cdot (s_j-1)!) \bigg) \\
& \times \bigg(
 \lim_{x_{\summ_d} \to p_d} |x_{\summ_d} - p_d|^{-s_d+1}
 \ldots
 \lim_{x_{\summ_{d-1} + 2} \to p_d} |x_{\summ_{d-1} + 2} - p_d|^{-1}
 \lim_{x_{\summ_{d-1} + 1} \to p_d}
 \bigg) \\
& \; \ldots \;
\bigg( \lim_{x_{\summ_1} \to p_1} 
|x_{\summ_1} - p_1|^{-s_1+1} \ldots \lim_{x_2 \to p_1} |x_2 - p_1|^{-1} \lim_{x_1 \to p_1} \bigg)
\PartF_{\iota(\alpha)} (\domain; x_1, \ldots, x_{2N}).
\end{align*} 
\end{lemma}

Lemma~\ref{lem:fusion of fused partition functions}, and thereby 
property \textnormal{(FUS)} in Theorem~\ref{thm:CFT properties}, 
is a straightforward consequence of 
recognizing~\eqref{eq:fused partition function} in the explicit expressions in 
the following Proposition~\ref{prop:higher-valence solution limit} 
(whose proof, less surprisingly, resembles that of the expression~\eqref{eq:fused partition function} in Theorem~\ref{thm:scaling limit of pinched pertition functions}).
We prove it for the case where $\domain = \bH$; the general case is similar.
Below, we shall consider symmetric kernels $\mathfrak{K}$ of the form $\mathfrak{K}(a, b) = \mathfrak{K}(b, a)=K_{a, b}(x_a, x_b)$ where $K_{a, b}$, for $a<b$, are smooth functions defined for $x_a < x_b$, and we interpret thus $\mathfrak{Z}^{\mathfrak{K}}_\alpha$ 
also as a function on $\chamber_{2N}$.

\begin{proposition} \label{prop:higher-valence solution limit}
Fix valences $\multii = (s_1,\ldots,s_d) \in \bZpos^d$ and and an index $j \in \{1,2,\ldots,d\}$.
Denote $J = J_j = \{ \summ_{j-1} +1, \ldots, \summ_j\}$. 
Suppose that given any $b \not \in J$ and $x_b$, 
the function $K_{a,b}(x_a, x_b)$ takes the same form $f_{b}(x_a)$ for all $a \in J$.
Then, for any $\alpha \in \LP_\multii$, the limit 
\begin{align*}
\left( \lim_{x_{\summ_j} \to \xi} 
|x_{\summ_j} - \xi|^{-s_j+1} \ldots \lim_{x_{\summ_{j-1}+2} \to \xi} |x_{\summ_{j-1}+2} - \xi|^{-1} \lim_{x_{\summ_{j-1}+1} \to \xi} \right)
\mathfrak{Z}^{\mathfrak{K}}_\alpha (x_1, \ldots, x_{2N}) 
\end{align*}
is finite and given by the following algorithm of modifications in the kernel $\mathfrak{K}(\cdot,\cdot)$.
\begin{enumerate}[leftmargin=3em,label=\textnormal{(\roman*):}, ref=(\roman*)]
\item \label{step: modification 1 again}
Set $\mathfrak{K}(a,b) = 0$ for all $\summ_{j-1} < a, b \leq \summ_j$;
and then,

\medskip

\item \label{step: modification 2 again}
for all $b \not \in J$ and for all $m \in \{1,\ldots,s_j \}$, 
set $\mathfrak{K}(\summ_{j-1}+m,b) = \mathfrak{K}(b,\summ_{j-1}+m)$ to equal 
\begin{align*}
\frac{1}{(m-1)!}
(\partial^{\mathrm{tan}} )^{m-1} f_b (\xi)
\qquad \textnormal{instead of} \qquad
f_b (x_{\summ_{j-1}+m}) .
\end{align*} 
\end{enumerate}
\end{proposition}

\begin{proof}
The first replacement~\ref{step: modification 1 again} 
is Lemma~\ref{lem:Fomin properties}~\ref{item:general zero-replacing rule b}, 
so it leaves $\mathfrak{Z}^{\mathfrak{K}}_\alpha$ invariant\footnote{Note that, in the case of the Brownian excursion kernel, this replacement removes the kernel entries that would have otherwise diverged in the limit.}.

Taking the first limit $x_{\summ_{j-1}+1} \to \xi$ amounts, by continuity, 
to substituting the variable $x_{\summ_{j-1}+1} = \xi$ in all the kernel entries. 
For the second limit
$x_{\summ_{j-1}+2} \to \xi$, we Taylor expand
\begin{align*}
\mathfrak{K}(\summ_{j-1}+2, b) 
\, = \,  f_b (x_{\summ_{j-1}+2}) 
\, = \,  f_b (\xi) \, + \,  |x_{\summ_{j-1}+2} - \xi| \, \partial^{\mathrm{tan}} f_b (\xi) \,  + \,  |x_{\summ_{j-1}+2} - \xi|^2 \; O(1).
\end{align*}
Now, analogously to the proof of Rule~\ref{step: modification 2} in Proposition~\ref{prop:higher-valency solution}, 
the first term gives no contribution, while the third one vanishes in the renormalized limit.
For the later renormalized limits, one performs higher-order Taylor expansions 
to obtain a derivative of an order that has not appeared before.
This implies the assertions concerning the limit. 
\end{proof}

We prove a more general simultaneous fusion rule later in Proposition~\ref{prop:fused SLE} in Section~\ref{sec:ASY}.
It will be crucial in SLE applications (cf.~Appendix~\ref{app:fused SLE}), though the above simpler result suffices for the bulk of this article.

\subsection{Proof of property \textnormal{(COV)} in Theorem~\ref{thm:CFT properties}}
\label{subsec:covariance of partition functions}

Let us next record a general conformal covariance property for the functions $\PartF_\alpha$, 
with explicit conformal weights of the Kac form $h_{1,s+1} = s (s+1)/2$. 
This gives property \textnormal{(COV)} in Theorem~\ref{thm:CFT properties} as the special case $\domain = \bH$.

In the unfused case $(s_1, \ldots, s_d) = (1, \ldots, 1)$, each determinant --- and hence the entire expression --- on the right-hand side of Equation~\eqref{eq:scaling limit of pinched pertition functions} inherits 
the conformal covariance of the Brownian excursion kernel (Equation~\eqref{eq:BM exc ker covariance}), 
immediately resulting in the asserted covariance property~\eqref{eq: COV general at kappa equals 2} with conformal weights $h_{1,2} = 1$ (i.e., $s_j=1$ for all $j$). 
Although this is not equally straightforward in the general case, the pure partition functions of Equation~\eqref{eq:fused partition function}
are also conformally covariant with conformal weights $(h_{1,s_1+1},\ldots,h_{1,s_d+1})$: 

\begin{lemma}
\label{lem:covariance of fused partition functions}
Let $\domain$ be a simply-connected domain and $p_1, \ldots, p_{d} \in \partial \domain$ distinct counterclockwise boundary points 
on sufficiently regular\footnote{E.g. $C^{1+\epsilon}$ for some $\epsilon>0$ --- we need derivatives of conformal maps defined on the boundary.} 
boundary segments, so that $\PartF_\alpha (\domain; p_1, \ldots, p_d)$ can be defined as in Equation~\eqref{eq:fused partition function}. 
Then, for any conformal map
$\varphi \colon \domain \to \bH$ with $\varphi(p_1)<\cdots<\varphi(p_{d})$, 
we have
\begin{align*}
\PartF_\alpha(\domain; p_1, \ldots, p_d) = \prod_{j=1}^d |\varphi'(p_j)|^{(s_j+1)s_j/2} \times \PartF_\alpha(\bH; \varphi(p_1), \ldots, \varphi(p_{d})) , \qquad \alpha \in \LP_\multii .
\end{align*} 
\end{lemma}

\begin{proof} 
This follows by injecting into the expression for $\PartF_\alpha$ in Lemma~\ref{lem:fusion of fused partition functions} 
the already known covariance property (which is easy to verify by hand, see~\cite[(COV2) in Theorem~4.1]{KKP:Boundary_correlations_in_planar_LERW_and_UST})
\begin{align*}
\PartF_{\iota(\alpha)} (\domain; x_1, \ldots, x_{2N}) 
= \prod_{j=1}^{2N} |\varphi'(x_j)| \times \PartF_{\iota(\alpha)}(\varphi(\domain); \varphi(x_1), \ldots, \varphi(x_{2N})) 
\end{align*}
and simple first-order difference approximations, e.g.,~ $\frac{1}{|x_2 - p_1|} = \frac{|\varphi' (p_1)|}{|\varphi (x_2) - \varphi (p_1)|} \, (1+o(1))$.
\end{proof}

%% file: tex/sec4-CFT.tex
In this section, we proceed with the properties \textnormal{(PDE)}, \textnormal{(LIN)}, and \textnormal{(POS)} in Theorem~\ref{thm:CFT properties}. 
The proof of each property is inductive in the valences, as illustrated in Figure~\ref{fig:fusion}. For the proofs, we will also need to prove the existence of a Frobenius series expansion for $\PartF_\alpha$. 

\subsection{Frobenius series expansion}
\label{subsec:Frobenius series}

We will need the existence and regularity properties of a Frobenius series of the BPZ PDE solutions. Fix valences $\multii = (s_1,\ldots,s_d) \in \bZpos^d$ and let 
\begin{align*}
\PartF \in \SolSp_\multii := \textnormal{span}_\bC \{\PartF_\alpha \; | \; \alpha \in \LP_\multii \} .
\end{align*}
Note that by Equation~\eqref{eq:fused partition function}, the function $\PartF$ is a linear combination of determinants of matrices whose entries are up to constant of the form $(x_\ell- x_k)^{-m}$; 
so $\PartF$ is a rational expression itself --- and in particular it has a meromorphic continuation to a complex neighborhood of the real line. 
Fix now an index $1 \leq j \leq d - 1$ and denote
\begin{align} \label{eq: eps and hatx}
\varepsilon := x_{j+1} - x_j 
\qquad \textnormal{and} \qquad 
\hat{x}_j := \frac{x_j + x_{j+1}}{2} 
\end{align}
(so that $x_j = \hat{x}_j - \tfrac{\varepsilon}{2}$ and $x_{j+1} = \hat{x}_j + \tfrac{\varepsilon}{2}$). 
Fixing $x_1 < \cdots < x_{j-1} < \hat{x}_j <  x_{j+2} < \cdots < x_d$ and varying $\varepsilon$, 
let us investigate the obtained a meromorphic rational function 
\begin{align} \label{eq: function g}
\varepsilon \quad \longmapsto \quad
\PartF(x_1,\ldots,x_{j-1},\hat{x}_j - \tfrac{\varepsilon}{2},\hat{x}_j + \tfrac{\varepsilon}{2},x_{j+2},\ldots,x_d) 
\end{align}
on $\bC$ with finitely many poles or zeroes 
(one of them at $\varepsilon = 0$, and the others corresponding to $x_j$ or $x_{j+1}$ colliding with the rest of the $x_i$:s). 
The explicit expressions for $\PartF_\alpha$ yield an upper bound $\Xi \in \bZ$ for the order\footnote{Here, the order of a zero is non-negative and the order of a pole is non-positive.} 
of the pole or zero at $\varepsilon = 0$ (independent of the fixed variables). 
Analyzing the Taylor coefficients of $\varepsilon^{\Xi} \PartF(x_1,\ldots,x_{j-1},\hat{x}_j - \tfrac{\varepsilon}{2},\hat{x}_j + \tfrac{\varepsilon}{2},x_{j+2},\ldots,x_d)$ and similarly for the derivatives of $\PartF$, one readily concludes the following expansion:

\begin{lemma} \label{lem:Frobenius series}
\textnormal{(Frobenius series).}
Any $\PartF \in \SolSp_\multii$ has the Frobenius (Laurent) series expansion 
\begin{align}
\label{eq:Frobenius series of partition function}
\PartF(\ldots,x_{j-1},\hat{x}_j - \tfrac{\varepsilon}{2},\hat{x}_j + \tfrac{\varepsilon}{2},x_{j+2},\ldots)
= \varepsilon^{\Xi} \sum_{k=0}^\infty \varepsilon^{k} \PartF^{(k)}(\ldots, x_{j-1}, \hat{x}_j, x_{j+2}, \ldots) ,
\end{align} 
which, given any fixed points $x_1 < \cdots < x_{j-1} < \hat{x}_j <  x_{j+2} < \cdots < x_d$, 
converges for all $\varepsilon \in \bC$ with $0 < |\varepsilon| < \max\{ 2 (x_{j+2}- \hat{x}_j), 2 (\hat{x}_j - x_{j-1})\}$. 
The coefficients $\PartF^{(k)}$ of this series are smooth functions of $(x_1, \ldots, x_{j-1}, \hat{x}_j, x_{j+2}, \ldots, x_d)$, 
and the partial derivatives $\partial_\cdot \PartF$ of $\PartF$ 
with respect to these variables also have a series expression with the same punctured disc of convergence:
\begin{align}
\label{eq:Frob derivatives}
\varepsilon^{\Xi} \sum_{k=0}^\infty \varepsilon^{k} \partial_\cdot \PartF^{(k)}(\ldots, x_{j-1}, \hat{x}_j, x_{j+2}, \ldots) .
\end{align}
\end{lemma}

This series expansion leads to a conclusion that is typical in ODE theory.

\begin{lemma}
\label{lem:indical equation}
\textnormal{(Indicial equation).}
Fix valences $\multii = (s_1,\ldots,s_d) \in \bZpos^d$ and suppose that some index $j \in \{1,\ldots, d-1\}$ satisfies $s_{j+1} = 1$.
Let $\PartF \in \SolSp_\multii$ as above. 
Assume that the function $\PartF$ also satisfies the second order PDE
\begin{align} \label{eq:BPZ}
\Bigg[ \pdder{x_{j+1}}
+  \sum_{i \neq j+1} \bigg( \frac{2}{x_i - x_{j+1}} \pder{x_i} - \frac{s_i(s_i+1)}{(x_i - x_{j+1})^2} \bigg)
\Bigg] \PartF(x_1,\ldots, x_d) = 0 .
\end{align}
Then,~\eqref{eq: function g} is either an identically zero function, or the lowest power $\varepsilon^\Xi$ with a non-zero coefficient in the series expansion~\eqref{eq:Frobenius series of partition function} satisfies
\begin{align*}
\Xi =
\begin{cases}
\Xi_+ := s_j , \\
\Xi_- := -(s_j+1) 
\end{cases}
=
\begin{cases}
h_{1,s_j+2} - h_{1,s_j+1} - h_{1,2} , \\
h_{1,s_j} - h_{1,s_j+1} - h_{1,2} .
\end{cases}
\end{align*}
\end{lemma}

Note that the two exponents are obtained from the two possible fusion rules of primary fields in a CFT with central charge $c=-2$ and Kac weights $h_{1,s_j+1} = s_j (s_j+1)/2$.

\begin{proof}
The chain rule of derivatives gives
\begin{align*}
\pder{x_j} = - \pder{\varepsilon} + \frac{1}{2} \pder{\hat{x}_j} 
\qquad \textnormal{and} \qquad 
\pder{x_{j+1}} = \pder{\varepsilon} + \frac{1}{2} \pder{\hat{x}_j} .
\end{align*}
Fixing $x_1 < \cdots < x_{j-1} < \hat{x}_j <  x_{j+2} < \cdots < x_d$ and only varying $\varepsilon$, Equation~\eqref{eq:BPZ} 
can be studied via the expansions~\eqref{eq:Frobenius series of partition function}--\eqref{eq:Frob derivatives}. The lowest powers of $\varepsilon$ originate from the terms
\begin{align*}
\; & \Bigg[ \pdder{x_{j+1}} 
+ \frac{2}{x_j - x_{j+1}} \pder{x_j} - \frac{s_j(s_j+1)}{(x_j - x_{j+1})^2} 
\Bigg] \PartF (x_1,\ldots, x_d) 
\\
= \; &
\left[ \pdder{\varepsilon} + \frac{1}{4} \pdder{\hat{x}_j} + \pder{\varepsilon} \pder{\hat{x}_j} 
+ \frac{2}{\varepsilon} \pder{\varepsilon} - \frac{1}{\varepsilon} \pder{\hat{x}_j} - \frac{s_j(s_j+1)}{\varepsilon^2} \right] \PartF (x_1,\ldots, x_d)
\end{align*}
and equal
\begin{align*}
\Bigg[ \pdder{\varepsilon} + \frac{2}{\varepsilon} \pder{\varepsilon} - \frac{s_j(s_j+1)}{\varepsilon^2} \Bigg] \PartF (x_1,\ldots, x_d).
\end{align*}
The lowest-order terms must sum up to zero, which yields the \emph{indicial equation}
\begin{align*}
\Xi (\Xi-1) + 2\Xi - s_j(s_j+1) = 0 ,
\end{align*}
with solutions $\Xi_+ = s_j$ and $\Xi_- = -(s_j +1 )$.
If there is no lowest-order term, then~\eqref{eq: function g} is identically zero in the region of convergence of the Laurent series, thus everywhere. 
\end{proof}

\subsection{Proof of property \textnormal{(PDE)} in Theorem~\ref{thm:CFT properties}}
\label{subsec:PDEs}

We follow an approach developed by Dub\'edat in~\cite{Dubedat:SLE_and_Virasoro_representations_localization, Dubedat:SLE_and_Virasoro_representations_fusion}, 
which relies on the framework of Virasoro uniformization developed in particular by Kontsevich and Friedrich~\cite{Kontsevich:Virasoro_and_Teichmuller_spaces, 
Kontsevich:CFT_SLE_and_phase_boundaries, Friedrich-Kalkkinen:On_CFT_and_SLE, Friedrich:On_connections_of_CFT_and_SLE}. 
(We have tried to keep the presentation rather short in the present work, 
and refer to~\cite[Section~4]{LPR:Fused_Specht_polynomials_and_c_equals_1_degenerate_conformal_blocks} and references therein for a more detailed exposition.)

\begin{proposition} \label{prop:higher-valency solution PDE}
For each $\alpha \in \LP_\multii$, the function $\PartF_\alpha$ of~\eqref{eq:fused partition function} satisfies the system of BPZ PDEs~\eqref{eq: BPZ PDE at kappa equals 2}. 
\end{proposition}

\begin{proof}
The proof of property \textnormal{(PDE)} in Theorem~\ref{thm:CFT properties}
proceeds by applying Lemma~\ref{lem:lemmafusion}, stated below, inductively: 
starting from the known second order PDEs in the unfused case where $\multii = (1,\ldots,1)$~\cite[Theorem~4.1]{KKP:Boundary_correlations_in_planar_LERW_and_UST} 
and recursively fusing boundary points and increasing the order of the PDEs (see Figure~\ref{fig:fusion} for an illustration). 
We will not explicitly perform the induction here, since the analogous case of $c=1$ 
was detailed in~\cite[Sections~3~\&~4]{LPR:Fused_Specht_polynomials_and_c_equals_1_degenerate_conformal_blocks}.  
\end{proof}

Since the strategy of the proof is similar to~\cite[Sections~3~\&~4]{LPR:Fused_Specht_polynomials_and_c_equals_1_degenerate_conformal_blocks}, which in turn invokes results from~\cite{Dubedat:SLE_and_Virasoro_representations_localization, Dubedat:SLE_and_Virasoro_representations_fusion},
we only outline the strategy here, including the necessary auxiliary results in order to deduce the asserted PDEs~\eqref{eq: BPZ PDE at kappa equals 2}.  
The following result is the gist of the proof: 
starting from a solution of two BPZ PDEs of orders $s_j+1$ and $2$ at $x_j$ and $x_j+1$, respectively, 
having a specific Frobenius series expansion as $|x_{j+1}-x_j| \to 0$ (cf.~Equation~\eqref{eq:asyf}), 
we can construct a solution of a BPZ PDE of order $s_j+2$ at $\hat{x}_j := \frac{x_j+ x_{j+1}}{2}$. 

\begin{lemma} \label{lem:lemmafusion}
Fix $d \geq 2$. 
Fix $\multii=(s_1,\ldots,s_d) \in \bZpos^d$ such that $s_{j}=\ell-1$ and $s_{j+1}=1$ for some index $j \in \{1,\ldots,d - 1\}$. 
Also, let $\PartF \colon \chamber_d \to \bR$ be 
a smooth function satisfying the BPZ PDEs 
\begin{align} \label{eq:saintaubineq}
\sD\super{x_i}_{s_i+1} \PartF(x_1,\ldots,x_d) = 0, 
\qquad \textnormal{for all } i \in \{1,\ldots,d \} , 
\end{align}
with $\sD\super{x_i}_{s_i+1}$ as in~\eqref{eq: BPZ operator at kappa equals 2}. 
Assume that $\PartF$ has the following Frobenius series expansion at $\hat{x}_j := \frac{x_j + x_{j+1}}{2}$\textnormal{:}
\begin{align} \label{eq:asyf}
\PartF(x_1,\ldots,x_d) = (x_{j+1}-x_j)^{\ell-1} 
\sum_{k=0}^\infty (x_{j+1}-x_j)^{k} \; 
\PartF^{(k)}(\ldots, x_{j-1}, \hat{x}_j, x_{j+2}, \ldots) ,
\end{align}
where $\PartF^{(k)}$ are smooth functions on $\chamber_{d-1}$. 
Then, the coefficient $\PartF^{(0)}$ satisfies the BPZ~PDEs 
\begin{align*} 
\sD\super{x_i}_{s_i+1} \PartF^{(0)}(x_1,\ldots,x_{j-1}, \hat{x}_j, x_{j+2},\ldots, x_d) = \; & 0, 
\qquad i \in \{1,\ldots,d \} \setminus \{ j, j+1 \} , \\
\sD\super{x_j}_{\ell+1} \PartF^{(0)}(x_1,\ldots,x_{j-1}, \hat{x}_j, x_{j+2},\ldots, x_d) = \; & 0 ,
\end{align*}
with differential operators 
as in~\eqref{eq: BPZ operator at kappa equals 2} but with variables 
$(x_1,\ldots,x_{j-1}, \hat{x}_j, x_{j+2},\ldots, x_d)$.
\end{lemma}

To finish the proof of Proposition~\ref{prop:higher-valency solution PDE}, 
we may apply Lemma~\ref{lem:lemmafusion} inductively. 
To this end, we first note that the needed Frobenius series follows from 
Lemmas~\ref{lem:Frobenius series} and \ref{lem:indical equation}, 
and the assumed PDEs~\eqref{eq:saintaubineq}.
With $\multii$ as above, from Lemma~\ref{lem:indical equation}  
we see that each solution $\PartF \in \SolSp_\multii$ has either the \emph{leading} 
$\Xi_- = -\ell$ or \emph{subleading asymptotics} 
$\Xi_+ = \ell-1$ of the indicial equation. 
By property \textnormal{(FUS)} (Lemma~\ref{lem:fusion of fused partition functions}), 
$\PartF_\alpha$ has the subleading asymptotics, required by Lemma~\ref{lem:lemmafusion}.

\smallskip

Lemma~\ref{lem:lemmafusion}, which holds for $c=-2$, is analogous to~\cite[Theorem~4]{LPR:Fused_Specht_polynomials_and_c_equals_1_degenerate_conformal_blocks} for $c=1$ 
and to~\cite[Theorem~15]{Dubedat:SLE_and_Virasoro_representations_fusion} for $c$ irrational. 
In fact, almost the entire proof of~\cite[Theorem~15]{Dubedat:SLE_and_Virasoro_representations_fusion} applies to all values of the central charge. 
The only part of the proof which requires $c$ to be irrational is an algebraic result,~\cite[Lemma~1]{Dubedat:SLE_and_Virasoro_representations_fusion}:
the submodule structure of highest-weight representations of 
the Virasoro algebra can be very intricate when $c$ is rational 
(while for $c$ irrational, all Verma modules are irreducible). 
For this reason, the proof of our Lemma~\ref{lem:lemmafusion} consists of 
proving the algebraic Lemma~\ref{lem:analoglemma1Dubedat} stated below, 
which is the analog of~\cite[Lemma~1]{Dubedat:SLE_and_Virasoro_representations_fusion} 
for the present case of $\kappa=2$ and $c=-2$.

In order to state the key algebraic Lemma~\ref{lem:analoglemma1Dubedat}, 
we first need to recall some standard notions about Virasoro-modules. 
The \emph{Virasoro algebra} %$\Vir$ 
\gls{symb:Vir} 
is the infinite-dimensional Lie algebra 
generated by the Virasoro modes %$\{L_n \;|\; n \in \bZ\}$ 
$\{\textnormal{\gls{symb:VirLn}} \;|\; n \in \bZ\}$ 
and the central element %$C$, 
\gls{symb:VirC}, 
\begin{align*}
\Vir = \bC C \oplus \bigoplus_{n \in \bZ} \bC L_n ,
\end{align*}
with the following commutation relations:
\begin{align} \label{eq:comm rel}
\begin{split}
[L_m,L_n] = \; & (m-n) L_{m+n} + \delta_{m,-n} \frac{m^2(m-1)}{12} C, \qquad m,n \in \bZ, \\
[C, \Vir] = \; & 0, \qquad n \in \bZ
\end{split}
\end{align}
(where $\delta_{i,j}$ stands for the Kronecker delta function, equaling zero unless $i=j$). 
It has the triangular decomposition $\Vir = \Vir^- \oplus \mathfrak h \oplus \Vir^+$, 
where $\mathfrak h = \bC C \oplus L_0$ and $\smash{\Vir^\pm = \underset{\pm n>0}{\oplus} \bC L_n}$. 
The universal enveloping algebra of the subalgebra $\Vir^-$ is
\begin{align*}
%\mathcal U(\Vir^-) 
\textnormal{\gls{symb:VirEnv}}
= \bigoplus_{\substack{0 < i_1 \leq \cdots  \leq i_k \\ k \geq 0}} \bC L_{-i_k} \cdots  L_{-i_1} ,
\end{align*}
and it has the \quote{standard basis} 
$\{L_{-i_k} \cdots  L_{-i_1} \;|\; 0< i_1 \leq \cdots  \leq i_k, \; k \geq 0\}$ by the Poincar\'e-Birkhoff-Witt theorem. 
Let us also note that $\mathcal U(\Vir^-)$ is a $\bZ$-graded associative algebra with \emph{degree} $\textnormal{deg}(L_n) := -n$ and $\textnormal{deg}(C) := 0$.
(See~\cite{Iohara-Koga:Representation_theory_of_Virasoro} for more background on $\Vir$.)

Let $V$ be a $\Vir$-module. 
For $(c,h) \in \bC^2$, a $(c,h)$-\emph{highest-weight vector} %$v_{h}^{c} \in V$ 
\gls{symb:Virhwv} $\in V$ 
is an element satisfying $C v_{h}^{c} = c v_{h}^{c}$, $L_0 v_{h}^{c} = h v_{h}^{c}$, 
and $L_n v_{h}^{c} = 0$ for all $n>0$. 
In this context, $c \in \bC$ is called the \emph{central charge} and 
$h \in \bC$ is called the \emph{weight} of $v_{h}^{c}$.  
The \emph{Verma module} %$M_{h}^{c}$ 
\gls{symb:VirVerma} 
is 
\begin{align*}
M_{h}^{c} = \bigoplus_{\ell \geq 0} (M_{h}^{c})_\ell ,
\qquad \textnormal{where} \qquad
(M_{h}^{c})_\ell := \bigoplus_{\substack{0 < i_1 \leq \cdots  \leq i_k \\ i_1+\cdots +i_k = \ell \\ k \geq 0}} \bC L_{-i_k}\cdots L_{-i_1} v_{h}^{c} .
\end{align*}
Note that the dimension of $(M_{h}^{c})_\ell$ is the number of partitions of $\ell$.
Moreover, it follows from the commutation relations~\eqref{eq:comm rel} 
that each element $v \in (M_{h}^{c})_\ell$ satisfies $L_0 v = (h+\ell) v$. 
Hence, we say that each $v \in (M_{h}^{c})_\ell$ is a vector in $M_{h}^{c}$ at \emph{level} (or degree) $\ell$.

A highest-weight vector $w_\ell \in M_{h}^{c}$ of level $\ell>0$ is called a \emph{singular vector}.  
If a non-zero singular vector 
can be found, 
then $M_{h}^{c}$ is said to be \emph{degenerate at level $\ell>0$}, 
and in this case, $w_\ell$ generates a proper submodule of $M_{h}^{c}$ isomorphic to $M_{h+\ell}^{c}$. 
Submodules of Verma modules were classified by B.~Fe{\u\i}gin and D.~Fuchs
\cite{Feigin-Fuchs:Invariant_skew-symmetric_differential_operators_on_the_line_and_Verma_modules_over_Virasoro,
Feigin-Fuchs:Verma_modules_over_Virasoro_book,
Feigin-Fuchs:Representations_of_Virasoro}.
In particular, every submodule of $M_{h}^{c}$ is generated by singular vectors. 
There is an exceptional set of parameters $(c,h)$ for which $M_{h}^{c}$ is not irreducible --- the Kac table~\cite{Kac:Contravariant_form_for_infinite-dimensional_Lie_algebras_and_superalgebras,
Kac:Highest_weight_representations_of_infinite_dimensional_Lie_algebras}.

Fix $c = -2$. 
From now on, we only consider Verma modules of type %$M_{h} := M_{h}^{-2}$.
\gls{symb:VirVerma2} $ := M_{h}^{-2}$. 
If 
\begin{align} \label{eq:conf_weights again} 
%h_{\ell} 
\textnormal{\gls{symb:Kac}}
= h_{1,\ell} := \frac{\ell(\ell-1)}{2} ,
\end{align}
then $M_{h_{\ell}}$ possesses a singular vector at level $\ell>0$ (weights~\eqref{eq:conf_weights again} belong to the Kac table).
Let %$v_{\ell} := \smash{v_{h_{\ell}}^{-2}}$ 
\gls{symb:Virhwv2} $:= \smash{v_{h_{\ell}}^{-2}}$ 
denote the highest-weight vector of $M_{h_{\ell}}$. 
Then, the singular vector at level $\ell$ has the form $w_\ell = \BSAOP_\ell v_{\ell}$, 
where $\BSAOP_\ell \in \mathcal U(\Vir^-)$ is some polynomial in the negative Virasoro generators. 
As the coefficient of $L_{-1}^\ell$ in $\BSAOP_\ell$ cannot vanish~\cite[Section~5.2.1]{Iohara-Koga:Representation_theory_of_Virasoro}, we may normalize it to one.
An explicit formula for the polynomial 
%$\BSAOP_\ell$ 
\gls{symb:BSAOP} 
was found in~\cite{BSA:Degenerate_CFTs_and_explicit_expressions_for_some_null_vectors}\footnote{Note that in~\eqref{eq:defBSAoperator}, the Virasoro generators $L_{-i_j}$ are not ordered.}:
\begin{align} \label{eq:defBSAoperator}
\BSAOP_\ell = \sum_{k=1}^\ell \sum_{\substack{i_1,\ldots ,i_k \geq 1 \\ i_1+\cdots +i_k = \ell}} \frac{(-2)^{\ell-k} (\ell-1)!^2}{\prod_{l=1}^{k-1} (\sum_{j=1}^l i_j)(\sum_{j=l+1}^k i_j)} 
\; L_{-i_1} \cdots  L_{-i_k} .
\end{align}
Observe that $L_0 (\BSAOP_\ell v_{\ell}) = (h_{\ell}+\ell) (\BSAOP_\ell v_{\ell}) = h_{\ell+1} (\BSAOP_\ell v_{\ell})$. 
In fact, the singular vector $\BSAOP_\ell v_{\ell}$ generates a submodule of $M_{h_{\ell}}$ isomorphic to $M_{h_{\ell+1}}$, which is the maximal proper submodule. 
Generally, when $c=-2$, there exists a one-dimensional infinite chain of submodules, 
where each arrow denotes the embedding of $M_{h_{j+1}}$ into $M_{h_{j}}$ giving its maximal proper submodule:
\begin{align} \label{eq:submodules}
M_{h_{\ell}} \hookleftarrow M_{h_{\ell+1}} \hookleftarrow M_{h_{\ell+2}} \hookleftarrow \cdots .
\end{align} 
This structure of the Verma module $M_{h_{\ell}}$ is referred to as \quote{chain} type (see~\cite[Figure~1]{Kytola-Ridout:On_staggered_indecomposable_Virasoro_modules}, 
and~\cite{Feigin-Fuchs:Verma_modules_over_Virasoro_book, Iohara-Koga:Representation_theory_of_Virasoro} for details). 
Let us also remark that the submodule structure of Verma modules can be more intricate 
for other rational values of the central charge~\cite{Feigin-Fuchs:Verma_modules_over_Virasoro, 
Iohara-Koga:Representation_theory_of_Virasoro}.

We will consider particular types of $\Vir$-modules. 
Let $t$ be a formal variable. 
For $\alpha, h\in \bR$, consider the space %$V_{\alpha,h} := \bC [[t]][t^{-1}]t^\alpha$ 
\gls{symb:VOAF} $:= \bC [[t]][t^{-1}]t^\alpha$  
of formal series with finitely many negative terms:
\begin{align*}
t^\alpha \sum_{k\in \bZ} a_k t^k , \qquad 
\textnormal{with} \quad \inf \{k \colon a_k\neq 0\} > -\infty .
\end{align*}
$V_{\alpha,h}$ is a $\Vir$-module with zero central charge $c=0$, 
where each generator $L_n$ acts by 
\begin{align*}
L_n \mapsto %L_n^0 
\textnormal{\gls{symb:VOAFAct}}
:= -t^{n+1}\partial_t-(n+1)ht^n .
\end{align*}
(The role of the parameter $\alpha$ will become clear in the fusion procedure later, see Lemma~\ref{lem:analoglemma1Dubedat}, and also~\cite[Section~8.A]{DMS:CFT}.)
The operators $\{L_n^0 \;|\; n \in \bZ\}$ satisfy the commutation relations 
$[L_m^0,L_n^0] = (m-n) L_{m+n}^0$ of the \emph{Witt algebra} 
(so $V_{\alpha,h}$ is also a Witt-module\footnote{Recall that the Virasoro algebra is a one-dimensional central extension of the Witt algebra.}).

Next, let $W$ be a $\Vir$-module with central charge $c=-2$, 
whose $\Vir$-action is simply denoted by $L_n$. 
Consider the space %$W\otimes V_{\alpha,h}$ 
\gls{symb:VOAFW} 
of formal series with coefficients in $W$:
\begin{align*}
t^\alpha \sum_{k\in \bZ} u_k t^k,\quad u_k\in W , \qquad 
\textnormal{with} \quad \inf \{k \colon u_k\neq 0\} > -\infty .
\end{align*}
$W\otimes V_{\alpha,h}$ is a $\Vir$-module with central charge $c=-2$,
where each generator $L_n$ acts by 
\begin{align*}
L_n \mapsto %\hat{L}_n
\textnormal{\gls{symb:VOAFWAct}}(vt^{\alpha+k}) := (L_nv)t^{\alpha+k}-(\alpha+k+(n+1)h)vt^{\alpha+k+n} ,
\qquad n,k \in \bZ, \; v \in W.
\end{align*}
Let %$\hat{\BSAOP}_\ell$ 
\gls{symb:BSAOPhat} 
be the BPZ operator in~\eqref{eq:defBSAoperator} 
with the substitutions $L_n \mapsto \hat L_n$ for all $n \in \bZ$.

We are now ready to state the key algebraic result (analogous to~\cite[Lemma~1]{Dubedat:SLE_and_Virasoro_representations_fusion}).
\begin{lemma} \label{lem:analoglemma1Dubedat}
Fix $\ell \geq 2$. 
Using the notation from~\eqref{eq:conf_weights again}, let us 
denote $\smash{\Tilde{h}} := h_{2}$, $\smash{\hat{h}} := h_{\ell}$, and $\alpha := h_{\ell+1} - \smash{\hat{h}} - \smash{\Tilde{h}} = \ell - 1$. 
Suppose $w = t^\alpha \sum_{k\geq 0} u_k t^k$ is a highest-weight vector of weight $\smash{\hat{h}}$ such that
\begin{align*}
\hat{\BSAOP}_\ell w = 0 
\qquad \textnormal{and} \qquad
\Tilde{\BSAOP}_2 w=0 ,
\end{align*}
where 
$\Tilde{\BSAOP}_2 := \Tilde{L}_{-1}^2 - \tfrac{1}{2}\Tilde{L}_{-2}$ is defined in terms of
\begin{align}
\label{eq:actiontildeL1} 
\Tilde{L}_{-1}:= \; & \partial_t \\
\label{eq:actiontildeL2} 
\Tilde{L}_{-2}:= \; & -t^{-1}\partial_t + t^{-1}L_{-1} + \smash{\hat{h}} t^{-2} + \sum_{k\geq 0} t^k L_{-2-k} .
\end{align}
Then, the coefficient $u_0$ 
is a highest-weight vector in $W$ of weight $h_{\ell+1}$ which satisfies $\BSAOP_{\ell+1} u_0=0$.
\end{lemma}

Note that $L_{-1}$ and $L_{-2-k}$ act on the module $W$, while $\Tilde{L}_{-1}$ and $\Tilde{L}_{-2}$ act on $W\otimes V_{\alpha,\smash{\Tilde{h}}}$. 

As mentioned above, Lemma~\ref{lem:analoglemma1Dubedat} implies Lemma~\ref{lem:lemmafusion}.

\begin{proof}
The proof consists of two steps. 
The first step is to check that $u_0 \in W$ is a highest-weight vector of weight $h_{\ell+1}$.  
Indeed, we have $\hat{L}_0 w = \smash{\hat{h}} w$ at degree $\alpha  = \ell - 1$, 
which yields $L_0u_0=(\smash{\hat{h}}+\smash{\Tilde{h}}+\alpha)u_0$. 
Moreover, we have $\hat{L}_n w=0$ at degree $\alpha$, for all $n>0$, 
which yields $L_nu_0=0$, for all $n>0$. 
This shows that $u_0$ is a highest-weight vector of weight $h_{\ell+1}$.

The second and last step of the proof is to find an element 
$P_k^0 \in \mathcal U(\Vir^-) \setminus \{0\}$ of degree $k < 2\ell+3$ such that $P_k^0 u_0=0$. 
To see why this is useful, consider the homomorphism $\phi \colon M_{h_{\ell+1}} \to W$ 
of $\Vir$-modules 
which maps the highest-weight vector $v_{\ell+1} \in M_{h_{\ell+1}}$ to $u_0 \in W$. 
The first isomorphism theorem of modules implies that $\textnormal{Ker}(\phi)$ is a proper submodule of $M_{h_{\ell+1}}$. 
Using the chain~\eqref{eq:submodules} of Verma modules, we obtain
\begin{align*}
M_{h_{\ell+1}} \hookleftarrow M_{h_{\ell+2}} \hookleftarrow M_{h_{\ell+3}} \hookleftarrow \cdots ,
\end{align*} 
where the image of $M_{h_{\ell+2}}$ is generated by $\BSAOP_{\ell+1} v_{\ell+1}$, 
the image of $M_{h_{\ell+3}}$ is generated by $\BSAOP_{\ell+2} (\BSAOP_{\ell+1} v_{\ell+1})$, etc. 
Now, if there exists $P_k^0 \in \mathcal U(\Vir^-) \setminus \{0\}$ of degree $k$ such that $P_k^0 u_0=0$, then it follows that $P_k^0 v_{\ell+1} \in \textnormal{Ker}(\phi)$. 
In particular, we have $\textnormal{Ker}(\phi) = M_{h_{\ell+2}}$ 
if $k<2\ell+3$, 
in which case we may conclude that $0 = \phi(\BSAOP_{\ell+1} v_{\ell+1}) = \BSAOP_{\ell+1} u_0$, as desired.

It now remains to construct such a $P_k^0$, which we will find from the expansion of $\hat{\BSAOP}_\ell$. 
To this end, consider first the assumption $\Tilde{\BSAOP}_2 w = 0$. Expanding by degree, we obtain
\begin{align*}
%0  = \; & \rho(\alpha)u_0 , \\
0  = \; & \rho(\alpha +1)u_1-L_{-1}u_{0},  \\
0  = \; & \rho(\alpha +k)u_k-\sum_{j=1}^{k}L_{-j}u_{k-j} ,
\end{align*}
where $\rho(a)=(a+\ell)(a-\ell+1)$ has roots $\alpha$ and $h_{\ell-1} - \smash{\hat{h}} - \smash{\Tilde{h}} = -\ell < \alpha$.
Thus, we have $\rho(\alpha+k)\neq 0$ for all $k>0$, 
and there are $R_0,R_1,\ldots \in \mathcal{U}(\Vir^-)$ such that $u_k = R_ku_0$ for all~$k$.

Next, consider the assumption $\hat{\BSAOP}_\ell w=0$:
\begin{align*}
\hat{\BSAOP}_\ell \Big( t^\alpha \sum_{k\geq 0} t^kR_ku_0 \Big)=0.
\end{align*}
Write
\begin{align*}
\hat{\BSAOP}_\ell \Big( t^\alpha \sum_{k\geq 0} t^kR_k \Big) = t^{\alpha-\ell} \sum_{k\geq 0} t^kP_k ,
\end{align*}
for some polynomials $P_k\in \mathcal{U}(\Vir^-)$ of degree $k$ such that $P_ku_0=0$, for all $k\geq 0$. 
We first focus on the coefficients of $L_{-1}^k$ of $P_k$ decomposed in the standard basis.  If $P,Q\in \mathcal{U}(\Vir^-)$ are homogeneous and such that $P=aL_{-1}^k+\cdots $ and $Q=bL_{-1}^{k'}+\cdots $ in the standard basis, then $PQ=abL_{-1}^{k+k'}+\cdots $ in the standard basis. 
(This holds because the commutation relations of $\textnormal{Vir}^-$ do not produce any monomial in $L_{-1}$.) 
We then see inductively that 
\begin{align*}
R_k=\frac{1}{\varrho(1)\cdots \varrho(k)}L_{-1}^k+\cdots , \qquad k \geq 1 ,
\end{align*}
with $\varrho(k)=\rho(\alpha+k)$.
Next, we write
\begin{align*}
\hat{\BSAOP}_\ell = \sum_{\substack{i+j+k=\ell \\ i,j,k\geq 0}} b_{i,j,k}t^{-i}\partial_t^jL_{-1}^k+\cdots , 
\end{align*}
where the remainder does not contain any monomial in $L_{-1}$. Note that $b_{0,0,\ell}=1$.

We now finally show that there exists an element $P_k^0 \in \mathcal U(\Vir^-) \setminus \{0\}$ of degree $d < 2\ell+3$ such that $P_k^0 u_0=0$. 
To this end, we assume towards a contradiction that no $P_k$ has a nonzero monomial 
in $L_{-1}$ for $k\leq 2\ell+2$. Then, we have
\begin{align*}
t^{\alpha-\ell} \sum_{k\geq 0} t^kP_k 
= \; & \hat{\BSAOP}_\ell \Big( t^\alpha \sum_{k\geq 0} t^kR_k \Big) \\
= \; & 
\bigg( \sum_{\substack{i+j+k=\ell \\ i,j,k\geq 0}} b_{i,j,k}t^{-i}\partial_t^jL_{-1}^k\bigg)\bigg( t^\alpha \sum_{l\geq 0} \frac{t^lL_{-1}^l}{\varrho(1)\cdots  \varrho(l)}  \bigg)+\cdots 
\end{align*}
For each $d=0,\ldots ,\ell$, let $Q_d$ be the polynomial of degree at most $d$ (determined by explicit differentiation) such that
\begin{align*}
\bigg( \sum_{\substack{i+j=d \\ i,j\geq 0}} b_{i,j,\ell-d}t^{-i}\partial_t^j\bigg)t^{\alpha+m+d} = Q_d(m)t^{\alpha+m} , \qquad m \geq -d .
\end{align*}
Note that $Q_0 = b_{0,0,\ell}=1$. 
Now, we have 
\begin{align*}
t^{\alpha-\ell} \sum_{k\geq 0} t^kP_k 
= \sum_{d=0}^\ell \sum_{j\geq-d}\frac{Q_d(j)}{\varrho(1)\cdots  \varrho(j+d)}t^{\alpha+j}L_{-1}^{\ell+j}+\cdots .
\end{align*}
By assumption, we know that the coefficients of monomials in $L_{-1}$ of degree $\ell+j$ for $j \in 
\{-\ell,-\ell+1,\ldots , \ell+2\}$ are vanishing.
Thus, we obtain
\begin{align*}
0 = \; & Q_\ell(-\ell) , \\
0 = \; & \frac{Q_\ell(-\ell+1)}{\varrho(1)}+Q_{\ell-1}(-\ell+1) , \\
0 = \; & \frac{Q_\ell(-\ell+2)}{\varrho(1)\varrho(2)}+\frac{Q_{\ell-1}(-\ell+2)}{\varrho(1)}+Q_{\ell-2}(-\ell+2) , \\
\vdots \; & \\
0 = \; & \frac{Q_\ell(0)}{\varrho(1)\cdots  \varrho(\ell)}+\frac{Q_{\ell-1}(0)}{\varrho(1)\cdots  \varrho(\ell-1)}+\cdots  +Q_0(0) , \\
\vdots\; & \\
0 = \; & \frac{Q_\ell(\ell+2)}{\varrho(1)\cdots  \varrho(2\ell+2)}+\frac{Q_{\ell-1}(\ell+2)}{\varrho(1)\cdots  \varrho(2\ell+1)}+\cdots  +\frac{Q_0(\ell+2)}{\varrho(1)\cdots  \varrho(\ell+2)} .
\end{align*}
Multiplying the $i$-th equation by $\varrho(1)\cdots  \varrho(i-1)$, we obtain 
\begin{align*}
0 = \; & Q_\ell(-\ell) , \\
0 = \; & Q_\ell(-\ell+1)+Q_{\ell-1}(-\ell+1)\varrho(1) , \\
0 = \; & Q_\ell(-\ell+2)+Q_{\ell-1}(-\ell+2)\varrho(2)+Q_{\ell-2}(-\ell+2)\varrho(1)\varrho(2) , \\
\vdots \; & \\
0 = \; & Q_\ell(0)+Q_{\ell-1}(0)\varrho(\ell)+\cdots  Q_0(0)\varrho(1)\cdots  \varrho(\ell) , \\
\vdots \; & \\
0 = \; & Q_\ell(\ell+2)+Q_{\ell-1}(\ell+2) \varrho(2\ell+2)+\cdots  +Q_0(\ell+3)\varrho(\ell+4)\cdots  \varrho(2\ell+2) .
\end{align*}
Since $\varrho(0)=0$, we find that for all $m \in \{-\ell,-\ell+1,\ldots , \ell+2\}$
\begin{align*}
0 = Q_\ell(m)+Q_{\ell-1}(m)\varrho(\ell+m)+\cdots  +Q_0(m)\varrho(m+1)\cdots  \varrho(\ell+m) 
=: \; & O(m) .
\end{align*}

On the one hand, since $Q_0=1$, the last term is non-vanishing and of degree $2\ell$, 
while all the other terms are of degree at most $2\ell-1$. 
Thus, $O$ is not the zero polynomial. 
On the other hand, since $O$ is a polynomial of degree at most $2\ell$ with $2\ell+3$
zeroes, we infer that $O \equiv 0$. 
This is a contradiction. 
Hence, we conclude that there exists an element $P_k^0 \in \mathcal U(\Vir^-) \setminus \{0\}$ of degree $d < 2\ell+3$ such that $P_k^0 u_0=0$. 
This concludes the proof.
\end{proof}

\subsection{Proof of properties \textnormal{(LIN)} and \textnormal{(POS)} in Theorem~\ref{thm:CFT properties}}
\label{subsec:linear independence  of partition functions}

\begin{proposition} \label{prop:higher-valency solution LIN}
The functions $\{\PartF_\alpha \; | \; \alpha \in \LP_\multii\}$ of~\eqref{eq:fused partition function} are linearly independent. 
\end{proposition} 

\begin{proof}
Fix a target valence $\multii = (s_1, \ldots, s_d ) \in \bZpos^d$. 
The proof proceeds by induction over the fusion from the left, as illustrated in Figure~\ref{fig:fusion}, that is, 
over the chain of valences recursively obtained as
\begin{align*}
(1, 1, \ldots, 1); \; (2, 1, \ldots, 1); \; \ldots; \;(s_1, 1, 1, \ldots, 1); \; (s_1, 2, 1, \ldots, 1); \; \ldots (s_1, s_2, \ldots, s_d) .
\end{align*}

The base case, i.e., the \quote{unfused} partition functions $\{\PartF_\alpha \; | \; \alpha \in \LP_N= \LP_{(1, \ldots, 1)}\}$ being linearly independent, was handled in~\cite[Theorem~4.1]{KKP:Boundary_correlations_in_planar_LERW_and_UST}. 
Hence, we can directly proceed to the induction step. 
To this end, let us take two subsequent valence vectors 
\begin{align*} 
\multii_1 := (s_1,\ldots,s_{j-1}, r, 1, \ldots, 1) \in \bZpos^t 
\qquad \textnormal{and} \qquad
\multii_2 := (s_1,\ldots, s_{j-1}, r+1, 1, \ldots, 1)  \in \bZpos^{t-1} ,
\end{align*}
in the chain above, 
and assume that the functions $\{\PartF_\alpha \; | \; \alpha \in \LP_{\multii_1}\}$ do not have any non-trivial linear relations. Moreover, assume, towards a contradiction, that 
$\{\PartF_{\hat{\alpha}} \colon \hat{\alpha} \in \LP_{\multii_2}\}$ do: for some coefficients $c_{\hat{\alpha}} \in \bC$ not all zero, we have
\begin{align}
\label{eq:assuming an impossible linear relation}
\sum_{\hat{\alpha} \in \LP_{\multii_2}} c_{\hat{\alpha}} \, \PartF_{\hat{\alpha}} (x_1, \ldots, x_{t-1}) = 0 , \qquad \textnormal{for all } x_1 < \cdots < x_{t-1}. 
\end{align}
For each $\hat{\alpha} \in \LP_{\multii_2}$, 
let $\alpha \in \LP_{\multii_1}$ be the unique $\multii_1$-valenced link pattern with $\imath(\alpha) = \imath (\hat{\alpha})$ (informally, 
$\alpha$ is obtained from $\hat{\alpha}$ by \quote{opening up} the $j$:th index with valence $r+1$ into two indices of valences $r$ and $1$). 
Set 
\begin{align}
\label{eq:create linear relation}
\PartF (x_1, \ldots, x_{t}) := \sum_{\hat{\alpha} \in \LP_{\multii_2}} c_{\hat{\alpha}} \, \PartF_{\alpha}(x_1, \ldots, x_{t}).
\end{align}
Now, using the notation from~\eqref{eq: eps and hatx}, 
fix for a moment $x_1, \ldots, x_{j-1}, \hat{x}_j, x_{j+2}, \ldots, x_t$ and study the limit of $\PartF$ when $\varepsilon \to 0$. 
By property \textnormal{(FUS)} (Lemma~\ref{lem:fusion of fused partition functions}), we know that in this limit, 
\begin{align*}
\PartF (x_1,\ldots, x_t) 
 = \frac{\varepsilon^{\Xi_+}}{r !} \sum_{\hat{\alpha} \in \LP_{\multii_2}} c_{\hat{\alpha}} \, \PartF_{\hat{\alpha}} (x_1,\ldots, x_{j-1}, \hat{x}_j, x_{j+2}, \ldots, x_t) + o(\varepsilon^{\Xi_+ }) = o(\varepsilon^{\Xi_+ }),
\end{align*}
where we also used~\eqref{eq:assuming an impossible linear relation}. 
By Lemma~\ref{lem:indical equation}, for these particular $x_1, \ldots, x_{j-1}, \hat{x}_j, x_{j+2}, \ldots, x_t$ and any $\varepsilon$, we have $\PartF = 0$. 
As this argument holds for any fixed $x_1, \ldots, x_{j-1}, \hat{x}_j, x_{j+2}, \ldots, x_t$, 
$\PartF (x_1,\ldots, x_t) $ given in~\eqref{eq:create linear relation} is simply the zero function. 
But now, we have deduced a non-trivial linear relation
between the partition functions $\{ \PartF_{{\beta}}: \beta \in \LP_{\multii_1} \}$, a contradiction.
\end{proof}

\begin{proposition} \label{prop:higher-valency solution POS}
For each $\alpha \in \LP_\multii$, the function $\PartF_\alpha$ of~\eqref{eq:fused partition function} is positive: $\PartF_\alpha(\chamber_{d}) \subset (0,\infty)$.
\end{proposition}

\begin{proof}
Fix $\alpha$ and let $\alpha_1 = \imath(\alpha), \ldots, \alpha_M = \alpha$ be link patterns, 
with valences as in the chain in the previous proof, and all with the same corresponding \quote{unfused} pattern $\imath (\alpha)$.
We proceed by induction over this chain of link patterns.
Again, the base case was proven in~\cite[Theorem~4.1]{KKP:Boundary_correlations_in_planar_LERW_and_UST}. 
At the induction step, let $\alpha_i$ and $\alpha_{i+1}$ be two subsequent link patterns; the induction hypothesis is $\PartF_{\alpha_i} > 0$, and assume towards a contradiction that for some fixed $x_1, \ldots, x_{j-1}, \hat{x}_j, x_{j+2}, \ldots, x_t$ using the notation from~\eqref{eq: eps and hatx}, we have
\begin{align*}
\PartF_{\alpha_{i+1}} (x_1, \ldots, x_{j-1}, \hat{x}_j, x_{j+2}, \ldots, x_t) = 0.
\end{align*}
Let us study $\PartF_{\alpha_i }(x_1, \ldots, x_t)$ only varying $\varepsilon$, where
$x_j = \hat{x}_j - \tfrac{\varepsilon}{2}$ and 
$x_{j+1} = \hat{x}_j + \tfrac{\varepsilon}{2}$.
Now, due to an identical argument as above, 
$\varepsilon \mapsto \PartF_{\alpha_i}$ is the zero function, and fixing any small enough value for $\varepsilon$,
we thus have found variables $x_1 < \cdots < x_t$
for which $\PartF_{{\alpha_i}}(x_1, \ldots, x_t) = 0$ equals zero, yielding a contradiction.
Hence, $\PartF_{\alpha_{i+1}}(\chamber_{d-1}) \subset (0,\infty)$, as desired.
\end{proof}

Collecting Lemmas~\ref{lem:fusion of fused partition functions}~\&~\ref{lem:covariance of fused partition functions} and Propositions~\ref{prop:higher-valency solution PDE},~\ref{prop:higher-valency solution LIN}~\&~\ref{prop:higher-valency solution POS},  
we obtain Theorem~\ref{thm:CFT properties}.

%% file: tex/sec5-alg.tex
The objective of this section is to investigate in detail algebraic properties of the space 
\begin{align} \label{eq: sol space sec 5}
\SolSp_\multii := \textnormal{span}_\bC \{\PartF_\alpha \; | \; \alpha \in \LP_\multii \} 
\end{align}
of correlation functions. 
We will begin with the unfused case of $\multii = (1^{2N}) =  (1, \ldots, 1)$.
In fact, it turns out that the algebraic properties of $\SolSp_\multii$ for any $\multii$ are are naturally inherited from those of $\SolSp\sub{1^{2N}}$ by fusion. 
Notice that for the latter space, we also know from~\cite{KKP:Boundary_correlations_in_planar_LERW_and_UST}
that 
\begin{align} \label{eq:det basis solsp}
%\SolSp_N 
\textnormal{\gls{symb:SolSpN}} 
= \SolSp\sub{1^{2N}} = \textnormal{span}_\bC\{\Delta_\beta^{\mathfrak K} \; | \; \beta \in \LP_N \} 
\end{align}
where $\LP_N = \LP\sub{1^{2N}}$ contain ordinary link patterns. 
In this case, the results of this section hold for any kernel $\mathfrak K$. 
The starting point is the key observation that for any $\beta \in \LP_N$, the functions $\Delta_\beta^\mathfrak{K}$ can be written as minors of a certain $(2N \times 2N)$-matrix. Then, we can invoke a result of Lascoux in \cite{Lascoux:Pfaffians_and_representations_of_the_symmetric_group} to prove that $\SolSp_N$ is an irreducible (simple) module over the group algebra 
$\bC[\SymGrp_{2N}]$ of the symmetric group, where $\SymGrp_{2N}$ acts by permuting the variables.

We then prove that as a representation of $\bC[\SymGrp_{2N}]$, the space $\SolSp_N$ actually descends to a representation of 
a quotient of $\bC[\SymGrp_{2N}]$ 
known as the \emph{Temperley-Lieb} (TL) \emph{algebra} 
$\TL_{2N} = \TL_{2N}(-2)$ with fugacity parameter $\fugacity = -2$ (and deformation parameter $q=1$), 
see Proposition~\ref{prop:descendtoTL}. 
We then proceed with the case of arbitrary valences $\multii$, 
proving that each $\SolSp_\multii$ is a simple module for the \emph{valenced-TL algebra} 
$\vTL_\multii = \vTL_\multii(-2)$~\cite{Flores-Peltola:Generators_projectors_and_the_JW_algebra, Flores-Peltola:Standard_modules_radicals_and_the_valenced_TL_algebra}
(i.e., it is non-zero and does not have any non-trivial irreducible submodules), 
see Proposition~\ref{prop:5p12}. 

\begin{remark}
The results of this section for arbitrary valences $\multii$ could well hold for more general kernels $\mathfrak{K}(x,y)$. 
Our proofs of certain key auxiliary results would, however, become more intricate. 
It would be nevertheless interesting to investigate such a generalization. 
\end{remark}

\begin{remark}
A space of correlation functions in a CFT with central charge $c=1$ and 
with algebraic properties similar to those of $\SolSp_\multii$ was considered recently by three of the authors in~\cite{LPR:Fused_Specht_polynomials_and_c_equals_1_degenerate_conformal_blocks}. 
Algebraically, the difference in these two setups is reflected in the deformation parameter of the TL algebra, equaling $q=+1$ in the present work and $q=-1$ in~\cite{LPR:Fused_Specht_polynomials_and_c_equals_1_degenerate_conformal_blocks}.
Note that in both cases, there is really no \quote{quantum deformation} present, 
and the associated quantum group $U_q(\mathfrak{sl}_2)$ 
just reduces to the classical Lie algebra $\mathfrak{sl}_2$. 
\end{remark}

\subsection{Symmetric group representations}

Throughout this section, 
we let $\Summed \in \bZpos$ be an integer and 
%$\lambda \vdash \Summed$ 
\gls{symb:partition} $\vdash \Summed$  
a partition of $\Summed$, that is, 
$\lambda=(\lambda_1,\lambda_2,\ldots ,\lambda_l)$ such that $\lambda_1 \geq \lambda_2 \geq\cdots \geq\lambda_l \geq 0$ and $\lambda_1+\lambda_2+\cdots +\lambda_l = \Summed$. 
Consider the symmetric group algebra %$\bC[\SymGrp_\Summed]$, 
\gls{symb:CSymGrp} 
generated by the transpositions %$\tau_i = (i,i+1) \in \SymGrp_\Summed$ 
\gls{symb:transposition} $= (i,i+1) \in $ \gls{symb:SymGrp} 
for $i\in\{1,\ldots ,\Summed-1\}$ with relations
\begin{equation*}
\begin{aligned}
\tau_i^2 = \; & 1 , && \textnormal{for } i\in\{1,\ldots ,\Summed-1\} , \\
\tau_i\tau_{i+1}\tau_i = \; & \tau_{i+1}\tau_i\tau_{i+1} , && \textnormal{for } i\in\{1,\ldots ,\Summed-2\} , \\
\tau_i\tau_j = \; & \tau_j\tau_i , && \textnormal{for } |j-i|>1.
\end{aligned}
\end{equation*}

\subsubsection{Irreducible modules for the symmetric group}
\label{subsubsec:Specht}

A \emph{Young diagram} of shape $\lambda$ is a finite collection of boxes arranged in $l$ left-justified rows with row lengths being, from top to bottom, $\lambda_1,\ldots ,\lambda_l$. 
A \emph{numbering} of a Young diagram is obtained by placing the numbers $1,\ldots ,\Summed$ in the $\Summed$ boxes of the Young diagram. 
A \emph{standard Young tableau} is a numbering which is strictly increasing across each row and down each column. 
The sets of numberings and of standard Young tableaux of shape $\lambda$ will be denoted 
%$\NB$ 
\gls{symb:NB} 
and %$\SYT$, 
\gls{symb:SYT}, 
respectively.

\begin{example}
For clarity, we illustrate a Young diagram $D$ of shape $\lambda =(4,3,2)$, 
as well as a numbering $A \in \NB$ and a standard Young tableau $T \in \SYT$:
\begin{align*}
D \; = \; \ydiagram{4,3,2}
, \qquad\quad
A \; = \; 
{\ytableausetup{nosmalltableaux}\begin{ytableau} 
2 & 1 & 4 & 7 \\
8 & 5 & 3 \\
9 & 6 
\end{ytableau}} 
, \qquad\quad
T \; = \; 
{\ytableausetup{nosmalltableaux}\begin{ytableau} 
1 & 3 & 4 & 7 \\
2 & 6 & 8 \\
5 & 9
\end{ytableau}} .
\end{align*}
\end{example}

The group $\SymGrp_\Summed$ acts on $\NB$ by letter permutations; 
the action of $\sigma \in \SymGrp_\Summed$ on a numbering $N \in \NB$ is denoted $\sigma.N$. 
For $A \in \NB$, let 
%$\Rows(A)$ 
\gls{symb:Rows} 
%(resp.~$\Columns(A)$) 
(resp.~\gls{symb:Columns}) 
be the subgroup of $\SymGrp_\Summed$ which preserves the set of entries of each of its rows (resp.~columns). 
A \emph{tabloid} 
%$\{A\}$ 
\gls{symb:tabloid} 
is an equivalence class of numberings defined by $\{A'\} = \{A\}$ 
if and only if $A'=\sigma.A$ for some $\sigma \in \Rows(A)$. 
The $\bC$-vector space spanned by tabloids of shape $\lambda$,
\begin{align*}
%M^\lambda 
\textnormal{\gls{symb:Mlambda}}
:= \textnormal{span}_\bC \{ \{A\} \; | \; A \in \NB\} ,
\end{align*}
carries a natural $\SymGrp_\Summed$-action denoted by $\sigma.\{A\} := \{\sigma.A\}$. 
\emph{Simple modules} of $\SymGrp_\Summed$ are
modules that are nonzero and have no nontrivial submodules. 
They can be realized in various ways as subspaces of $M^\lambda$, e.g., in terms of \quote{polytabloids} discussed below.

For each numbering $A \in \NB$, the \emph{column antisymmetrizer}
\begin{align*} 
%\epsilon_A 
\textnormal{\gls{symb:antisym}}
:= \sum_{\sigma \in \Columns(A)} \sign(\sigma) \, \sigma \; \in \; \bC[\SymGrp_\Summed]
\end{align*}
defines the associated \emph{polytabloid} %$v_A := \epsilon_A.\{A\} \in M^\lambda$. 
\gls{symb:polytabloid} $:= \epsilon_A.\{A\} \in M^\lambda$. 
Note that $\{A\} = \{A'\}$ does not necessarily imply that $v_A$ and $v_{A'}$ would be equal, 
since the actions of row and column permutations (the subgroups $\Rows(A)$ and $\Columns(A)$) do not commute in general.
Consider the $\bC$-vector space spanned by the polytabloids:
\begin{align*} 
%V^\lambda 
\textnormal{\gls{symb:Vlambda}}
:= \textnormal{span}_\bC \{ v_A \; | \; A \in \NB \} \; \subset \; M^\lambda .
\end{align*}
Note that $\sigma.v_A = v_{\sigma.A}$, 
for $\sigma \in \SymGrp_\Summed$ and $A \in \NB$, 
which implies that $V^\lambda$ has a $\SymGrp_\Summed$-module structure.
In fact, $V^\lambda$ is simple
--- and these modules form a complete set $\{V^\lambda\}_{\lambda \vdash n}$ 
of pairwise non-isomorphic simple modules of the algebra $\bC[\SymGrp_\Summed]$~\cite{Specht:Die_irreduziblen_Darstellungen_der_symmetrischen_Gruppe}.
Moreover, $V^\lambda$ has an important \emph{polytabloid basis} $\{ v_T \; | \; T \in \SYT \}$ indexed by the standard Young tableaux.

\subsubsection{Action on $\SolSp_N$} 

From now on, let %$\pi$ 
\gls{symb:pin} 
denote the rectangular Young diagram with $2N$ boxes and $N$ rows. Let $\mathfrak{K} \colon \{1,...,2N\} \times \{1,...,2N\} \to \bC$ be a symmetric function (kernel), 
that is, $\mathfrak{K}(i,j) = \mathfrak{K}(j,i)$ for all $i$ and $j$. 
Let $M^\mathfrak{K}$ be a $(2N \times 2N)$-matrix with entries
\begin{align*}
M_{i,j}^\mathfrak{K} := 
\begin{cases} 0, & i=j, \\
\mathfrak{K}(x_i,x_j), & i \neq j. 
\end{cases}
\end{align*}
By a \emph{minor}, we mean the determinant of the $N\times N$ matrix whose $(i,j)$th element is $\smash{M^{\mathfrak K}_{a_i,b_j}}$ for $i=1,\ldots, N$, $j=1,\ldots, N$. 
For a numbering $A \in \NBof{\pi}$, denote by $[M^\mathfrak{K}]_A$ the minor of $M^\mathfrak{K}$ consisting of rows $a_1,\ldots,a_N$ and columns $b_1,\ldots,b_N$ where $a_1,\ldots,a_N$ (resp.~$b_1,\ldots,b_N$) are the numbers in the first (resp.~second) column of $A$ read from top to bottom.

As before, we consider link patterns $\beta = \{ \link{a_1\;}{\;b_1\,}, \ldots, \link{a_N\;}{\;b_N\,} \} \in \LP_N$ with the left-to-right orientation $a_1 < \cdots < a_N$ and $a_k < b_k$ for all $k$. 
One obtains a numbering from $\beta$ by filling the first (resp.~second) column with $(a_k)_{k=1}^N$ (resp.~$(b_k)_{k=1}^N$). 
By abuse of notation, we write
\begin{align*}
\Delta^\mathfrak{K}_\beta = [M^\mathfrak{K}]_\beta.
\end{align*}
where $\Delta^\mathfrak{K}_\beta$ is the Fomin type determinant defined in Equation~\eqref{eq: LPdet}.
Conversely, for a numbering $A$, because one can permute entries in each column of $A$ in a way as to obtain a link pattern $\beta$, 
the minor $[M^\mathfrak{K}]_A$ is, up to a sign, 
equal to $\Delta^\mathfrak{K}_\beta$. 
With~\eqref{eq:det basis solsp}, this implies that 
\begin{align}
\SolSp_N=\textnormal{span}_\bC\{ [M^\mathfrak{K}]_A \; | \; A\in \NBof{\pi} \}.
\end{align}

Now, the symmetric group $\SymGrp_{2N}$ acts on $\SolSp_N$ by permuting the entries of a numbering, 
or, equivalently, by permuting the arguments of $f \in \SolSp_N$. 
We now recall the following result.

\begin{proposition}[\cite{Lascoux:Pfaffians_and_representations_of_the_symmetric_group} Proposition~2.4] 
\label{prop:Lascoux}
The map that, for each numbering $A$, 
sends $v_A\in V^\pi$ to $[M^\mathfrak{K}]_A\in \SolSp_N$, is an isomorphism of $\bC[\SymGrp_{2N}]$-modules.
\end{proposition}

Using this knowledge, we will next show that this representation descends to a representation of a certain quotient of $\bC[\SymGrp_{2N}]$: namely, the Temperley-Lieb algebra.

\begin{definition}
The \emph{Temperley-Lieb algebra} $\TL_{\Summed}(\fugacity)$ with 
fugacity $\fugacity := -q - q^{-1} \in \bC$ parameterized by 
$q \in \bC \setminus \{0\}$
is generated by $e_i \in \TL_{\Summed}(\fugacity)$ for $i\in\{1,\ldots ,\Summed-1\}$ with relations
\begin{equation} \label{eq:TLpres}
\begin{aligned}
e_i^2 = \; & \fugacity \, e_i , && \textnormal{for } i\in\{1,\ldots ,\Summed-1\} , \\
e_ie_{i+1}e_i = \; &  e_i , && \textnormal{for }  i\in\{1,\ldots ,\Summed-2\} , \\
e_ie_{i-1}e_i = \; &  e_i , && \textnormal{for } i\in\{2,\ldots ,\Summed-1\}, \\
e_ie_j = \; & e_je_i, && \textnormal{for } |j-i|>1.
\end{aligned}
\end{equation}
Setting $q=1$ (i.e., $\fugacity = -2$) and $\tau_k=1+e_k$ for all $k$,  
the defining relations~\eqref{eq:TLpres} of $\TL_\Summed = \TL_{\Summed}(-2)$ can be written in the form
\begin{equation*}
\begin{aligned}
\tau_i^2 = \; & 1 , && \textnormal{for } i\in\{1,\ldots ,\Summed-1\} , \\
\tau_i\tau_{i+1}\tau_i = \; & \tau_{i+1}\tau_i\tau_{i+1} , && \textnormal{for } i\in\{1,\ldots ,\Summed-2\}, \\
\tau_i\tau_j = \; & \tau_j\tau_i , && \textnormal{for } |j-i|>1, \\
1+\tau_i+\tau_{i+1}+\tau_i\tau_{i+1}+\tau_{i+1}\tau_i+\tau_i\tau_{i+1}\tau_i = \; &  0 , && \textnormal{for } 
i\in \{1,\ldots ,\Summed-2\} ,
\end{aligned}
\end{equation*}
which implies that $\TL_\Summed$ is a quotient of the group algebra $\bC[\SymGrp_\Summed]$. 
The first three relations are the defining relations of $\bC[\SymGrp_\Summed]$, and
the fourth corresponds to the vanishing of the symmetrizers on three consecutive sites.
This identification also corresponds to the skein relation for the Kauffman bracket polynomial~\cite{Kauffman:State_models_and_the_Jones_polynomial}  with $q=1$.
\end{definition} 

\begin{remark}
The Temperley-Lieb algebra is usually presented as a quotient of $\bC[\SymGrp_\Summed]$, where the extra relation is 
the vanishing of the antisymmetrizers (instead of symmetrizers) on three consecutive sites. The two presentations give rise to isomorphic algebras.
\end{remark}

\begin{proposition} \label{prop:descendtoTL}
As a representation of $\bC[\SymGrp_{2N}]$, $\SolSp_N$ descends to a representation of $\TL_{2N}$. 
\end{proposition}

\begin{proof}
The key argument is that $\SolSp_N\cong V^{\pi}$, which is provided to us by Proposition~\ref{prop:Lascoux}. 
Indeed, since $\pi$ is a Young diagram with two columns, at least two of the three sites must lie in the same column of $\pi$. Therefore, any $v_A \in V^\pi$ is antisymmetric in at least two variables among $\{x_k,x_{k+1},x_{k+2}\}$. 
Thus, symmetrizing on three consecutive sites equals zero.
\end{proof}

\begin{remark}
The Temperley-Lieb algebra also admits a diagrammatic representation~\cite{Kauffman:State_models_and_the_Jones_polynomial}. 
From this viewpoint, the space $\SolSp_N$ is isomorphic to a module whose basis is given by link patterns in $\LP_N$ (cf.~\cite{Flores-Peltola:Standard_modules_radicals_and_the_valenced_TL_algebra}). 
We will not discuss this fact further in the present article.
\end{remark}

\subsection{Valenced Temperley-Lieb action}

\begin{definition}
\emph{Jones-Wenzl idempotents}~\cite{Wenzl:On_sequences_of_projections} 
in the Temperley-Lieb algebra $\TL_\Summed(\fugacity)$ are
nonzero elements %$\textnormal{JW}_{i,j} \neq 0$ 
\gls{symb:JW} $\in \TL_\Summed(\fugacity) \setminus \{ 0 \}$ 
for $i,j \in \{ 1,\ldots ,\Summed \}$ with $i<j$, defined recursively: 
\begin{align*}
\textnormal{JW}_{i,j}\textnormal{JW}_{i,j} = \; & \textnormal{JW}_{i,j}, \\
e_k\;\textnormal{JW}_{i,j} = \; & \textnormal{JW}_{i,j}\;e_k=0 , \qquad\textnormal{for all } k\in \{ i,i+1,\ldots ,j-1\}. 
\end{align*}
In the case $q=1$ (i.e., $\fugacity=-2$), the Jones-Wenzl idempotents are given by the \emph{antisymmetrizers}
(or rather, their images under the quotient map):
\begin{align*}
\textnormal{JW}_{i,j} = \frac{1}{(j-i+1)!} 
\sum_{\sigma \in \langle \tau_i,\tau_{i+1},\ldots ,\tau_{j-1}\rangle} \textnormal{sgn}(\sigma) \, \sigma 
\end{align*}
(where $\langle \tau_i,\tau_{i+1},\ldots ,\tau_{j-1}\rangle$ denotes the subgroup of $\SymGrp_{2N}$ generated by $\tau_i,\tau_{i+1},\ldots ,\tau_{j-1}$).
\end{definition}

The \quote{colored symmetric group} %$\SymGrp_{s_1} \times \cdots \times \SymGrp_{s_d}$,
\gls{symb:clSYM},  
where $\SymGrp_{s_i} = \langle\tau_{\summ_{i-1}+1}, \dots, \tau_{\summ_i-1}\rangle$, 
is a subgroup of $\SymGrp_{\Summed}$ giving rise to the $\multii$-\emph{antisymmetrizer} idempotent 
\begin{align}\label{eq:idempotent}
%\idpt 
\textnormal{\gls{symb:idpt}}
:= \frac{1}{s_1!\cdots s_d!} \prod_{k=1}^d \sum_{\sigma \in \SymGrp_{s_k}} \sign(\sigma) \, \sigma 
\; \in \; \bC[\SymGrp_{s_1} \times \cdots \times \SymGrp_{s_d}] 
\; \subset \; \bC[\SymGrp_\Summed] ,
\end{align}
Abusing notation, we denote also by $\idpt$ the corresponding image in the $\TL_{2N}$ quotient:
\begin{align} \label{eq:def_idpt}
\idpt = \prod_{k=1}^d \textnormal{JW}_{\summ_{k-1}+1,\summ_k}.
\end{align}

\begin{definition} \label{def:vTL}
(\cite{Flores-Peltola:Standard_modules_radicals_and_the_valenced_TL_algebra, Flores-Peltola:Generators_projectors_and_the_JW_algebra}).
The \emph{valenced Temperley-Lieb algebra} $\vTL_\multii = \vTL_\multii(-2)$ with fugacity $-2$  
is defined as the associative algebra 
$\vTL_\multii := \idpt \TL_\Summed(-2) \idpt$ with unit\footnote{Notice in particular that $\vTL_\multii$ is not a unital subalgebra of $\TL_\Summed(-2)$, since it has unit $\idpt$.} $\idpt$.
\end{definition}

We next investigate the $\vTL_\multii$-module $\idpt (\SolSp_N)$.
For this purpose, let %$\RSYTof{\pi}$ 
\gls{symb:RSYT} 
denote the set of row-strict Young tableaux of shape $\pi$ and content $\multii$, 
that is, Young tableaux whose entry $i$ appears $s_i$ times, 
and whose entries increase strictly along rows and weakly down columns.

\begin{proposition} \label{prop:5p12} 
If $\underset{1 \leq i \leq d}{\max} s_i \leq N$, 
then $\idpt(\SolSp_N)$ is a simple $\vTL_\multii$-module with 
\begin{align*}
\dim \idpt( \SolSp_N) = |\RSYTof{\pi}| .
\end{align*}
\end{proposition}

\begin{proof}
This follows from classification of simple modules of the \quote{fused Hecke algebra} 
$\idpt \bC[\SymGrp_{2N}] \idpt$ given in~\cite[Theorem~6.5]{Crampe-Poulain-d-Andecy:Fused_braids_and_centralisers_of_tensor_representations_of_Uq_gln}. (See also~\cite[Section~2]{LPR:Fused_Specht_polynomials_and_c_equals_1_degenerate_conformal_blocks} for $q=-1$.) 
\end{proof}

From now on, we restrict ourselves to the 
kernels $\mathfrak{K}(x,y)=\ExcKH(x,y)$ and 
investigate algebraic properties of the space $\SolSp_\multii$ that we will derive from the case $\multii = (1^{2N})$.

\begin{lemma} \label{lem:lemma5p12}
For each $F \in \SolSp_N$, we have 
\begin{align}
(\idpt.F)(x_1,\ldots,x_d) = \left( \prod_{k=1}^d \prod_{\summ_{k-1}< i<j\leq \summ_{k}} (x_j-x_i)\right) g(x_1,\ldots,x_d) ,
\end{align} \label{eq:5p12}
where $g$ is a rational function with no pole at $x_i=x_j$, for $i,j\in \{\summ_{k-1},\ldots,\summ_{k}\}$, $k \in \{1, \ldots, d\}$.
\end{lemma}

\begin{proof}
Any basis element $F=\Delta^\ExcKH_\beta$, with $\beta \in \LP_N$, is a rational function which does not possess any singular terms of odd order at $x_i=x_j$, for $i,j\in \{\summ_{k-1},\ldots,\summ_{k}\}$, $k \in \{1, \ldots, d\}$. 
Therefore, $\idpt.F$ is a rational function which is analytic at those loci, since any singular term of even order of $F$ vanishes upon antisymmetrization. Finally, by antisymmetry, $\idpt.F$ has zeros at these points, resulting in the prefactor in~\eqref{eq:5p12}.
The result follows by linearity. 
\end{proof}

For valences $\multii=(s_1,\ldots,s_d) \in \bZpos^d$ such that $s_1+\cdots+s_d = 2N$, denote 
\begin{align*}
%\mathfrak{D}_\multii 
\textnormal{\gls{symb:Dmultii}} 
:= \big\{ (x_1,\ldots, x_{2N}) \in \bC^{2N} \; | \;  
x_{\summ_{k-1}+1}=x_{\summ_{k-1}+2}=\cdots =x_{\summ_{k}}\textnormal{ for all } k \in \{ 1, \ldots, d\} \big\}
\; \subset \; \bC^{2N} ,
\end{align*}
and for a function $f \colon U \to \bC$ defined on a domain $U\subset \bC^{2N}$ which can be continuously extended 
to a subset of $\mathfrak{D}_\multii$, we shall write
\begin{align} \label{eq: eval notation}
%[f]_\textnormal{eval} 
\textnormal{\gls{symb:feval}} 
\colon \bC^d \to \bC
\end{align}
for the function obtained from 
$f(x_1,\ldots,x_{2N})$ by the evaluations of variables (projection) 
$x_{\summ_{k-1}+1}=x_{\summ_{k-1}+2}=\cdots =x_{\summ_{k}}$ for all $k \in \{ 1, \ldots,d\}$. 
We denote the variables of $f$ and $[f]_\textnormal{eval}$ respectively by $(x_1,\ldots, x_{2N}) \in \bC^\Summed$ and $(\xi_1,\ldots ,\xi_d) \in \bC^d$. 

\begin{proposition} \label{prop:5p15}
The solution space $\SolSp_\multii$ of~\eqref{eq: sol space sec 5} equals
\begin{align} \label{eq:5p15}
\SolSp_\multii = 
\left\{ \left[ \frac{\idpt.F}{\prod_{k=1}^d\prod_{\summ_{k-1}< i<j\leq \summ_{k}} (x_j-x_i)}\right]_{\textnormal{eval}} , \; F \in \SolSp_N \right\},
\end{align}
where the right-hand side is well-defined by Lemma~\ref{lem:lemma5p12}. 
Moreover, the map 
$\psi \colon \idpt(\SolSp_N) \to \SolSp_\multii$,
\begin{align} \label{eq:psimap}
\idpt.F \longmapsto \left[ \frac{\idpt.F}{\prod_{k=1}^d\prod_{\summ_{k-1}< i<j\leq \summ_{k}} (x_j-x_i)}\right]_{\textnormal{eval}} ,
\end{align}
is an isomorphism of vector spaces.
\end{proposition}

\begin{proof}
Denote the right-hand side of the asserted equality~\eqref{eq:5p15} by $\mathcal T_\multii$. 
Fix $\alpha \in \LP_\multii$. 
Note that $\PartF_{\iota(\alpha)} \in \SolSp_N$. 
Moreover, for any $j_k \in [\summ_{k-1}+1,\summ_k-1]$ and for any $k \in \{1,\ldots,d\}$, 
we have $\link{j_k}{j_k+1} \notin \iota(\alpha)$ (because otherwise, $\alpha \notin \LP_\multii$). 
Therefore, it follows from~\eqref{eq:simple anti-symmetry of inv Fomin} that $\PartF_{\iota(\alpha)}$ is completely antisymmetric with respect to $(x_{\summ_{k-1}+1}, \ldots, x_{\summ_{k}})$, or, equivalently, 
\begin{align} \label{eq:pZ=Z}
\idpt .\PartF_{\iota(\alpha)} = \PartF_{\iota(\alpha)} .
\end{align}
Next, we rewrite the iterated limit~\eqref{eq:iterated limit expression} 
(proven in Lemma~\ref{lem:fusion of fused partition functions})
in the following form:
\begin{align*} 
\PartF_\alpha (\Lambda; \xi_1, \ldots, \xi_d) = \; & \bigg( \prod_{j=1}^d ( 0! \cdot 1! \cdot \ldots \cdot (s_j-1)!) \bigg) \left(
 \lim_{x_{\summ_d} \to \xi_d} 
 \ldots
 \lim_{x_{\summ_{d-1} + 2} \to \xi_d} 
 \lim_{x_{\summ_{d-1} + 1} \to \xi_d}
 \right) \\
& \; \ldots \;
\left( \lim_{x_{\summ_1} \to \xi_1} 
 \ldots \lim_{x_2 \to \xi_1} \lim_{x_1 \to \xi_1} \right)
\frac{\PartF_{\iota(\alpha)} (\Lambda; x_1, \ldots, x_{2N})}{\prod_{k=1}^d \prod_{j=1}^{s_k} |x_{\summ_{k-1}+j} - \xi_k|^{j-1}}.
\end{align*}
We now rewrite the ratio as follows:
\begin{align*}
\frac{\PartF_{\iota(\alpha)} (\Lambda; x_1, \ldots, x_{2N})}{\prod_{k=1}^d \prod_{j=1}^{s_k} |x_{\summ_{k-1}+j} - \xi_k|^{j-1}} = \frac{\PartF_{\iota(\alpha)} (\Lambda; x_1, \ldots, x_{2N})}{\prod_{k=1}^d \prod_{\summ_{k-1}< i<j\leq \summ_{k}} (x_j-x_i)} \frac{\prod_{k=1}^d \prod_{\summ_{k-1}< i<j\leq \summ_{k}} (x_j-x_i)}{\prod_{k=1}^d \prod_{j=1}^{s_k} |x_{\summ_{k-1}+j} - \xi_k|^{j-1}}.
\end{align*}
A careful inspection shows that the iterated limit of the second ratio equals one, 
while due to~\eqref{eq:pZ=Z} the iterated limit of the first ratio equals 
\begin{equation*}
\left[ \frac{\idpt.\PartF_{\iota(\alpha)}}{\prod_{k=1}^d\prod_{\summ_{k-1}< i<j\leq \summ_{k}} (x_j-x_i)}\right]_{\textnormal{eval}}.
\end{equation*} By the product rule of limits, we conclude that $\PartF_\alpha \in \mathcal T_\multii$ and that
\begin{align*}
\SolSp_\multii \subset \mathcal T_\multii.
\end{align*} 

Next, note that the sets $\LP_\multii$ and $\RSYTof{\pi}$ are in bijection: 
indeed, to each $\beta \in \LP_\multii$ we can uniquely associate $T \in \RSYTof{\pi}$ by listing the beginnings (resp.~endings) of links in the first (respectively second) column of $\pi$. 
Thus, from Proposition~\ref{prop:higher-valency solution LIN} we infer that $\dim \SolSp_\multii = |\RSYTof{\pi}|$. 
Thus, to finish the proof, 
it suffices to show that $\dim \mathcal T_\multii \leq |\RSYTof{\pi}|$. To this end, note that the map
\begin{align*}
\idpt.F \mapsto \left[ \frac{\idpt.F}{\prod_{k=1}^d\prod_{\summ_{k-1}< i<j\leq \summ_{k}} (x_j-x_i)}\right]_{\textnormal{eval}}
\end{align*}
is surjective onto $\mathcal T_\multii$ by definition. 
Therefore, we conclude that $\dim \mathcal T_\multii \leq |\RSYTof{\pi}|$.
\end{proof}

We now conclude with the statement which implies Theorem~\ref{thm:vTLrep}.

\begin{corollary} \label{cor:corollaryvTL}
The map $\psi$ in~\eqref{eq:psimap} induces an action of $\vTL_\multii := \idpt \TL_\Summed \idpt$
on $\SolSp_\multii$ as
\begin{align*}
(\idpt a\idpt) .\left[ \frac{\idpt.F}{\prod_{k=1}^d\prod_{\summ_{k-1}< i<j\leq \summ_{k}} (x_j-x_i)}\right]_{\textnormal{eval}}=\left[ \frac{\idpt a\idpt.F}{\prod_{k=1}^d\prod_{\summ_{k-1}< i<j\leq \summ_{k}} (x_j-x_i)}\right]_{\textnormal{eval}} ,
\end{align*}
for all $a\in \TL_{2N}$. 
Moreover, $\SolSp_\multii$ is a simple $\vTL_\multii$-module. 
\end{corollary}

\begin{proof}
The map in Proposition~\ref{prop:5p15} is an isomorphism of vector spaces from 
$\idpt(\SolSp_N)$ to $\SolSp_\multii$, and the former space is a simple $\vTL_\multii$-module
by Proposition~\ref{prop:5p12}.
\end{proof}

\begin{remark} 
The $\vTL_\multii$-action on $\SolSp_\multii$ defined in Corollary~\ref{cor:corollaryvTL} is natural in the sense that it commutes with fusion. 
More precisely, as a $\TL_{2N}$-module, we have
\begin{align*}
\SolSp_N = \idpt (\SolSp_N ) \oplus \textnormal{ker}(\idpt) ,
\end{align*}
since $\idpt$ is a projector that respects the $\TL_{2N}$-action.  
The fusion map $\psi\super{\textnormal{fus}}$ defined by
\begin{align*}
\psi\super{\textnormal{fus}} \colon & \; \SolSp_N \to \SolSp_\multii , \\
F \mapsto  & \;  \left[ \frac{\idpt.F}{\prod_{k=1}^d\prod_{\summ_{k-1}< i<j\leq \summ_{k}} (x_j-x_i)}\right]_{\textnormal{eval}} ,
\end{align*}
acts as $\psi$ on $\idpt (\SolSp_N )$ and as $\psi\super{\textnormal{fus}} \equiv 0$ on $\textnormal{ker}(\idpt)$. 
Thus, it is a morphism of representations.
\end{remark} 

\begin{remark} \label{rem:commute}
Just like for the Temperley-Lieb algebra, the valenced Temperley-Lieb algebra can be defined in terms of diagrams. From this viewpoint, $\SolSp_\multii$ is isomorphic to a $\vTL_\multii$-module whose basis elements are valenced link patterns in $\LP_\multii$ --- see~\cite{Flores-Peltola:Generators_projectors_and_the_JW_algebra, Flores-Peltola:Standard_modules_radicals_and_the_valenced_TL_algebra}.
\end{remark}

%% file: tex/sec6-asy.tex
The purpose of this section is to derive asymptotic properties of the pure partition functions as neighboring endpoints are fused together. 
The asymptotics can be interpreted as specific \emph{fusion rules} for the corresponding CFT fields $\Phi_{1,s_j+1}$. 
In the unfused case, where all the fields are $\Phi_{1,2}$, 
these asymptotics properties are very important for singling out the pure partition functions as \emph{unique} solutions to a certain boundary value problem 
(see~\cite{BBK:Multiple_SLEs_and_statistical_mechanics_martingales,FSKZ:A_formula_for_crossing_probabilities_of_critical_systems_inside_polygons,Peltola-Wu:Global_and_local_multiple_SLEs_and_connection_probabilities_for_level_lines_of_GFF,Peltola:Towards_CFT_for_SLEs}
and references therein.) 
Analogous properties for CFT correlation functions with irrational central charges have been proven in~\cite{Peltola:Basis_for_solutions_of_BSA_PDEs_with_particular_asymptotic_properties}. 

In the last Section~\ref{subsec:ASY simul}, we furthermore derive a simultaneous fusion property for the simplest, unfused partition functions, generalizing Proposition~\ref{prop:higher-valence solution limit}.

For notation, fix $\multii = (s_1, \ldots, s_d) \in \bZpos^d$ with $s_1 + \cdots + s_d = 2N$, $j \in \{1,2,\ldots,d-1 \}$, and $\alpha \in \LP_\multii$; 
we will be analyzing $\PartF_\alpha(x_1,\ldots, x_d)$ as $ x_j, x_{j+1} \to \xi \in (x_{j-1},x_{j+2})$. 
Write $\bs{x} = (x_1,\ldots, x_d) \in \chamber_d$.
Denote the fused valences and coordinates by $\hat{\multii}$ and $\hat{\bs{x}}$, respectively. 
These can be expressed in terms of 
the multiplicity\footnote{Here, we crucially treat the valenced link pattern $\alpha$ as a multiset, not via its \quote{unfused} link pattern $\imath (\alpha)$.} 
%$\mult_{j, j+1}$ 
of the link $\link{j}{j+1}$ in $\alpha$: 
\begin{itemize}[leftmargin=1.5em]
\item 
if $\mult_{j, j+1} =  s_{j} = s_{j+1} $, i.e., 
if all of the links from $x_j$ and $x_{j+1}$ connect these two points 
(this is the case of the \emph{leading asymptotics} for the function $\PartF_\alpha$), 
then 
\begin{align*}
\begin{cases}
\hat{\multii} = (s_1, \ldots, s_{j-1}, s_{j+2}, \ldots, s_d) \in \bZpos^{d-2} , \\
\hat{\bs{x}} = (x_1, \ldots, x_{j-1}, x_{j+2}, \ldots, x_d)  \in \chamber_{d-2} ,
\end{cases}
\end{align*}
\item 
and in all other (\emph{subleading}) cases, we have
\begin{align*}
\begin{cases}
\hat{\multii} = (\ldots, s_{j-1}, s_j + s_{j+1} - 2 \mult_{j, j+1}, s_{j+2}, \ldots)\in \bZpos^{d-1} , \\
\hat{\bs{x}} = (\ldots, x_{j-1}, \xi, x_{j+2}, \ldots)   \in \chamber_{d-1} . 
\end{cases}
\end{align*}
\end{itemize}
Denote by $\hat{\alpha} \in \LP_{\hat{\multii}}$ the valenced link pattern obtained from $\alpha$ by removing $\mult_{j, j+1}$ links $\link{j}{j+1}$ and re-labeling the remaining indices appropriately (see Figure~\ref{fig:ASY patterns}).

\begin{figure}
\raisebox{-0.12\textwidth}{
\includegraphics[width=0.35\textwidth]{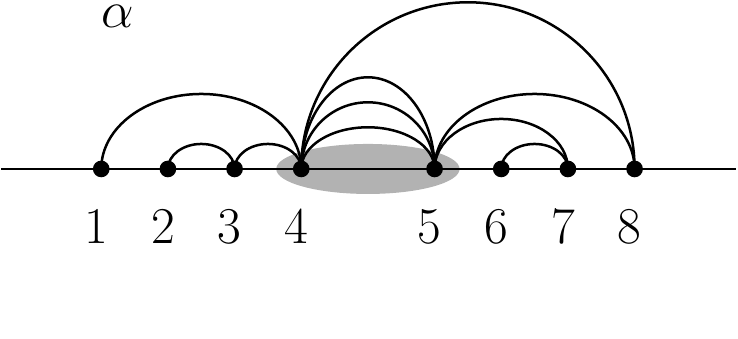}
}
$\rightsquigarrow$
\raisebox{-0.12\textwidth}{
\includegraphics[width=0.35\textwidth]{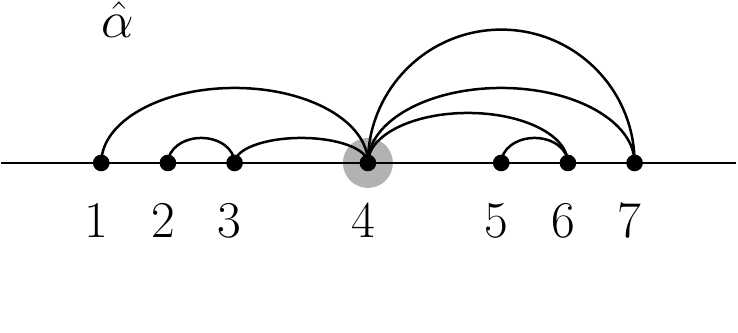}
}
\caption{\label{fig:ASY patterns}
Link pattern $\hat{\alpha}$ obtained from $\alpha$ by fusing two endpoints. 
}
\end{figure}

\begin{figure}
\includegraphics[width=0.35\textwidth]{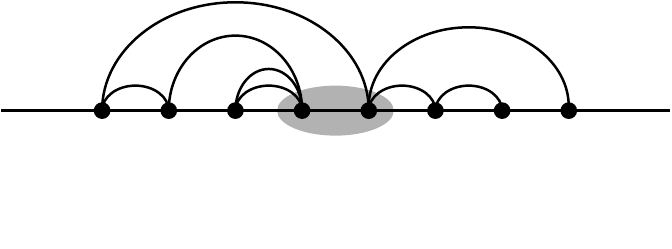} \qquad
\includegraphics[width=0.35\textwidth]{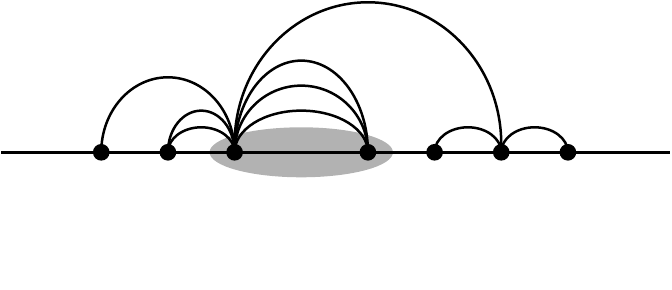}
\caption{\label{fig:ASY patterns special cases}
Example link patterns $\alpha$ in the two special cases $\mult_{j,j+1} =0$ (left) and $\mult_{j,j+1} =\min(s_{j},s_{j+1})$ (right), 
which lay the basis for the proof of Theorem~\ref{thm:ASY}, and for which $C(s_j, s_{j+1}, \mult_{j,j+1}) > 0$ is obtained explicitly.
}
\end{figure}

\begin{theorem}
\label{thm:ASY}
\textnormal{\bf \textnormal{(ASY)} Asymptotics:} 
There exist constants %$C(s_j, s_{j+1}, \mult_{j,j+1}) \geq 0$ 
\gls{symb:const} $ \geq 0$ 
such that
\begin{align}\label{eq: asymptotic properties}
\lim_{x_j,x_{j+1}\to\xi}
\frac{\PartF_\alpha(\bs x)}{|x_{j+1}-x_j|^{\Delta^{s_{j},s_{j+1}}_{\mult_{j,j+1}}}} 
= \; & C(s_j, s_{j+1}, \mult_{j,j+1}) \; 
\PartF_{\hat{\alpha}}(\hat{\bs{x}})
\end{align}
where we set $\PartF_\emptyset \equiv 1$ for the empty link pattern $\emptyset \in \LP_0$ and (using the Kac conformal weights)
\begin{align*}
%\Delta^{s_{j},s_{j+1}}_m 
\textnormal{\gls{symb:pow}}
= \; & h_{1,s_j+s_{j+1}+1-2m} - h_{1,s_j+1} - h_{1,s_{j+1}+1} \\
= \; & 2 m^2 - m (2 s_j + 2 s_{j+1} + 1 ) + s_j s_{j+1}.
\end{align*}
\end{theorem}

If either $\mult_{j,j+1} =0$ or $\mult_{j,j+1} =\min(s_{j},s_{j+1})$ (see Figure~\ref{fig:ASY patterns special cases}), then we also obtain an explicit, strictly positive expression for $C(s_j, s_{j+1}, \mult_{j,j+1})$\footnote{In principle, all $C(s_j, s_{j+1}, \mult_{j,j+1})$ could be found by computing the limit~\eqref{eq: asymptotic properties} for some example pattern $\alpha$ with the correct values of $s_j, s_{j+1}, \mult_{j,j+1}$. However, because $\PartF_\alpha$ is defined via combinatorial numbers with no known closed-form expression, we do not expect to obtain one for $C(s_j, s_{j+1}, \mult_{j,j+1})$, either.}. 
We believe that the constant is nonzero in all cases, 
just as in the generic situation in~\cite[Theorem~5.3]{Peltola:Basis_for_solutions_of_BSA_PDEs_with_particular_asymptotic_properties}. 
The constants $C(s_j, s_{j+1}, \mult_{j,j+1})$ can be thought of as \emph{structure constants} 
for chiral conformal blocks with central charge $c=-2$ and degenerate conformal weights 
$h_{1,s_j+1}, h_{1,s_{j+1}+1}, h_{1,s_j+s_{j+1}+1-2m}$.

The next Sections~\ref{subsec:ASY no links case}--\ref{subsec: asy gen} constitute the proof of this result. 
We will first explicate the easy cases $\mult_{j,j+1} =0$ or $\mult_{j,j+1} =\min(s_{j},s_{j+1})$. The general case is a combination of these two. 
The last Section~\ref{subsec:ASY simul} concerns the generalization of Proposition~\ref{prop:higher-valence solution limit}.

\subsection{Asymptotics, case $\mult_{j, j+1} = 0$}
\label{subsec:ASY no links case}

By symmetry, we may assume that $s_{j+1} \leq s_j$.

\begin{proposition} \label{prop:ASY no links case}
In the setup of Theorem~\ref{thm:ASY}, assuming $\mult_{j, j+1} = 0$ and  $s_{j+1} \leq s_j$, we have
\begin{align} \label{eq:ASY no links case}
\lim_{x_{j+1}, x_j \to \xi} 
\frac{\PartF_{\alpha} (\bs x)}{|x_{j+1} - x_j|^{s_j s_{j+1}}} = \; & C(s_j, s_{j+1}) \, \PartF_{\hat{\alpha}} (\hat{\bs{x}}) ,
\\ \nonumber
\textnormal{where} 
\qquad 
C(s_j, s_{j+1}) = C(s_j, s_{j+1}, 0) = \; & \frac{\prod_{k = 1}^{s_{j+1}} (k-1)!}{\prod_{\ell = 1}^{s_{j+1}} (s_j + \ell -1)!} .
\end{align}
\end{proposition}

\begin{proof}
For notational simplicity, we only consider the limit as $x_{j+1} \downarrow x_j$, so $\xi = x_j$. 
Denote the sets of indices to be fused by $I:= \{ \summ_{j-1} + 1, \ldots, \summ_j \}$ and $K := \{ \summ_{j} + 1, \ldots, \summ_{j+1} \} $ 
and their elements, in order, by $i_1, \ldots, i_{s_j}$ and by $k_1, \ldots, k_{s_{j+1}}$, respectively. 
The partition function $\PartF_{\alpha} (\bs x)$ is, by its definition in Theorem~\ref{thm:scaling limit of pinched pertition functions}, 
and inverse Fomin type sum $\mathfrak{Z}^\mathfrak{K}_\alpha(\bs x)$. Furthermore, being currently in the case $\mult_{j, j+1} = 0$,
where we may additionally assume, by~Lemma~\ref{lem:Fomin properties}\ref{item:general zero-replacing rule b}, 
that the kernel $\mathfrak{K}$ is modified to be
\begin{align*}
\mathfrak{K}(i_p, k_q) = 0 , \qquad \textnormal{for all } 1 \leq p \leq s_j \textnormal{ and } 1 \leq q \leq s_{j+1},
\end{align*}
and the rest of the kernel entries are as in Equation~\eqref{eq: valenced kernel}.
In particular, for any $1 \leq b \leq 2N$, there exist functions $f_b$ 
(either zero, or suitable derivatives of the Brownian excursion kernel), 
locally smooth in some open interval containing $[x_j, x_{j+1}]$, and such that
\begin{align*}
& \mathfrak{K}(i_1, b) = f_b(x_j),  \quad
\mathfrak{K}(i_2, b) = f_b'(x_j),  \quad \ldots,  \quad
\mathfrak{K}(i_{s_j}, b) = f_b^{(s_j - 1)}(x_j) \qquad \textnormal{and} \\
& \mathfrak{K}(k_1, b) = f_b(x_{j+1}),  \quad
\mathfrak{K}(k_2, b) = f_b'(x_{j+1}),  \quad \ldots,  \quad
\mathfrak{K}(k_{s_{j+1}}, b) = f_b^{(s_{j+1} - 1)}(x_{j+1}).
\end{align*}
Likewise, $\PartF_{\hat{\alpha}} (\hat{\bs{x}})$ is another inverse Fomin type sum $\mathfrak{Z}^{\hat{\mathfrak{K}}}_\alpha(\hat{\bs{x}})$ with 
\begin{align}
\label{eq:fused kernels}
& \hat{\mathfrak{K}} (k_1, b) = f_b^{(s_{j} )}(x_{j}),  \quad
\hat{\mathfrak{K}} (k_2, b) = f_b^{(s_{j} + 1 )}(x_{j}),  \quad \ldots,  \quad
\hat{\mathfrak{K}} (k_{s_{j+1}}, b) = f_b^{(s_{j} + s_{j+1} -1 )}(x_{j}),
\end{align}
and $\hat{\mathfrak{K}} (a, b) = \mathfrak{K} (a, b)$ if $a, b \not \in K$.

Set $\delta := x_{j+1} - x_j$,  
and Taylor expand until the derivative $f_b^{(s_{j} + s_{j+1} - 1)}(x_{j})$ appears:
\begin{align} \label{eq:Taylor series for kernels}
\begin{split}
\mathfrak{K}(k_1, b) = \; & f_b(x_{j}) + \delta f_b'(x_{j}) + \cdots 
+ \frac{\delta^{s_j + s_{j+1} -1}}{(s_j + s_{j+1} -1) !}  \, f_b^{(s_{j} + s_{j+1} - 1)}(x_{j}) + O(\delta^{s_j + s_{j+1}}) \\
\mathfrak{K}(k_2, b) = \; & f_b'(x_{j}) + \delta f_b''(x_{j}) + \cdots 
+ \frac{\delta^{s_j + s_{j+1} -2}}{(s_j + s_{j+1} -2) !} \, f_b^{(s_j + s_{j+1} -1)}(x_{j}) + O(\delta^{s_j + s_{j+1} -1})  \\
\; & \vdots \\
\mathfrak{K}(k_{s_{j+1}}, b) = \; & f_b^{(s_{j+1} - 1)}(x_{j}) + \delta f_b^{(s_{j+1} )} (x_{j}) + \cdots 
+ \frac{\delta^{s_j}}{s_j !} \, f_b^{(s_{j} + s_{j+1} - 1)}(x_{j}) + O(\delta^{s_j + 1}).
\end{split}
\end{align}
By linearity in the kernels $\mathfrak{K}(k_q, \cdot)$, we can think of $\mathfrak{Z}^\mathfrak{K}_\alpha$ as 
a sum of different inverse Fomin type sums $\mathfrak{Z}^{\widetilde{\mathfrak{K}}}_\alpha$, 
obtained by choosing each $\widetilde{\mathfrak{K}}(k_q, \cdot)$ to be one of the terms in 
the Taylor expansion of $\mathfrak{K}(k_q, \cdot)$ above (and $\widetilde{\mathfrak{K}} (a, b) = \mathfrak{K} (a, b)$ if $a, b \not \in K$). 
Now, by Lemma~\ref{lem:Fomin properties}\ref{item:general zero-replacing rule a} (which applies as $\mult_{j, j+1} = 0$), $\mathfrak{Z}^{\widetilde{\mathfrak{K}}}_\alpha$ 
equals zero if some factor $f_b(x_j), \ldots, f_b^{(s_{j} + s_{j+1} - 1)}(x_{j})$ appears 
twice as $\widetilde{\mathfrak{K}}(i_p, \cdot)$ or $\widetilde{\mathfrak{K}}(k_q, \cdot)$, 
for some $1 \leq p \leq s_j$ or $1 \leq q \leq s_{j+1}$.  
Hence, the only terms $\mathfrak{Z}^{\widetilde{\mathfrak{K}}}_\alpha$ that contribute to the leading order in $\delta$ are of the form
\begin{align*}
\begin{array}{l l l l l}
\widetilde{\mathfrak{K}}(k_1, b) & \in & \Big\{ \dfrac{ \delta^{s_j}}{s_j !} \, f_b^{(s_j)}(x_{j}), & \ldots, & \dfrac{\delta^{s_j + s_{j+1} -1 }}{(s_j + s_{j+1} -1) !} \, f_b^{(s_{j} + s_{j+1} - 1)}(x_{j}) \Big\}\\
\widetilde{\mathfrak{K}}(k_2, b)  & \in & \Big\{ \dfrac{\delta^{s_j-1}}{(s_j -1)!} \, f_b^{(s_j)}(x_{j}) , & \ldots , & \dfrac{\delta^{s_j + s_{j+1} -2}}{(s_j + s_{j+1} -2) !} \, f_b^{(s_j + s_{j+1} -1)}(x_{j})  \Big\} \\
& \vdots  &&&\\
\widetilde{\mathfrak{K}}(k_{s_{j+1}}, b)  & \in  & \Big\{  \dfrac{\delta^{s_j - s_{j+1} +1}}{(s_j - s_{j+1} +1 )!} \, f_b^{(s_{j} )}(x_{j}), &  \ldots, & \dfrac{ \delta^{s_j}}{s_j !} \, f_b^{(s_{j} + s_{j+1} - 1)}(x_{j}) \Big\} ,
\end{array}
\end{align*}
where we have also removed 
the subleading error terms $O(\delta^{s_j + s_{j+1}}), \ldots, O(\delta^{s_j + 1})$ 
(and the first term on the last line is present since $s_{j} \geq s_{j+1}$).
Note also that in the table above, each $\widetilde{\mathfrak{K}} (k_{q}, b)$, with $1 \leq q \leq s_{j+1}$, 
has $s_{j+1}$ alternative forms to take, 
the $r$:th one being 
\begin{align*}
\frac{\delta^{s_j - q + r}}{(s_j - q + r)!} \, f_b^{(s_j - 1 + r)}(x_{j}), \qquad 1 \leq r \leq s_{j+1}.
\end{align*}
If the higher-order derivatives $f_b^{(s_{j} )}(x_{j}), \ldots , f_b^{(s_{j} + s_{j+1} - 1)}(x_{j})$ appear twice as kernels $\widetilde{\mathfrak{K}} (k_q, b)$, then $\mathfrak{Z}^{\widetilde{\mathfrak{K}}}_\alpha = 0$. 
Therefore, the kernels $\widetilde{\mathfrak{K}}$ in the table above that potentially yield a non-zero inverse Fomin sum $\mathfrak{Z}^{\widetilde{\mathfrak{K}}}_\alpha$ are indexed by permutations $\tau$ of $ s_{j+1} $ elements, and given by
\begin{align*}
\widetilde{\mathfrak{K}}_{\tau} (k_q, b) 
= \frac{\delta^{s_j - q + \tau(q)}}{(s_j - q + \tau(q))!} \, f_b^{(s_j - 1 + \tau(q))}(x_{j}), \qquad \textnormal{for all } 1 \leq q \leq s_{j+1}.
\end{align*}
For instance, when $\tau = \mathrm{id}$ is an identity permutation, we obtain
\begin{align*}
\widetilde{\mathfrak{K}}_{\mathrm{id}} (k_1, b) 
\; = \; \frac{\delta^{s_j }}{s_j!} \, f_b^{(s_j )}(x_{j}), \, \ldots, \, \widetilde{\mathfrak{K}}_{\mathrm{id}} (k_{s_{j+1}}, b) 
\; = \; \frac{\delta^{s_j }}{s_j!} \, f_b^{(s_j + s_{j+1} -1 )}(x_{j}).
\end{align*}
Note that, up to the factors $\delta^{s_j }/s_j!$, this coincides with the kernel $\hat{\mathfrak{K}}$ for $\PartF_{\hat{\alpha}} (\hat{\bs{x}})$ given in~\eqref{eq:fused kernels}. 
In particular, by linearity of inverse Fomin type sums, we obtain
\begin{align*}
\mathfrak{Z}^{\widetilde{\mathfrak{K}}_{\mathrm{id}} }_\alpha (\hat{\bs{x}})
= \Big( \frac{\delta^{s_j }}{s_j!} \Big)^{s_{j+1}} \mathfrak{Z}^{\hat{\mathfrak{K}} }_\alpha (\hat{\bs{x}})
= \Big( \frac{\delta^{s_j }}{s_j!} \Big)^{s_{j+1}} \PartF_{\hat{\alpha}} (\hat{\bs{x}}).
\end{align*}
In general, via a similar linearity argument and the antisymmetry property~\eqref{eq:simple anti-symmetry of inv Fomin}, we obtain
\begin{align*}
\mathfrak{Z}^{\widetilde{\mathfrak{K}}_\tau }_\alpha (\hat{\bs{x}})
= \; & \sign(\tau) \, \prod_{q=1}^{s_{j+1}}
\frac{\delta^{s_j - q + \tau(q)}}{(s_j - q + \tau(q))!} \times \mathfrak{Z}^{\hat{\mathfrak{K}} }_\alpha (\hat{\bs{x}})
\\
= \; & \delta^{s_j s_{j+1}} \,  \sign(\tau) \, \prod_{q=1}^{s_{j+1}}
\frac{1}{(s_j - q + \tau(q))!} \times \PartF_{\hat{\alpha}} (\hat{\bs{x}}).
\end{align*}
Summing the contributions of the different permutations $\tau$, we obtain the expansion 
\begin{align*}
\PartF_{\alpha} (\bs x) 
\; = \;\; & \delta^{s_j s_{j+1}} \, \PartF_{\hat{\alpha}} (\hat{\bs{x}}) 
\; \sum_{\tau} \sign(\tau) \prod_{q=1}^{s_{j+1}}
\frac{1}{(s_j - q + \tau(q))!} \; + \; O(\delta^{s_j s_{j+1} + 1})\\
\; = \;\; & 
 \delta^{s_j s_{j+1}} \, \PartF_{\hat{\alpha}} (\hat{\bs{x}}) 
 \det
 \bigg(
 \frac{1}{(s_j - \ell + k)!}
 \bigg)_{\ell, k = 1}^{s_{j+1}} \; + \; O(\delta^{s_j s_{j+1} + 1}) .
\end{align*}
The asserted asymptotics~\eqref{eq:ASY no links case} 
then follows by inserting to the above the value of the determinant,
which will be evaluated explicitly in Lemma~\ref{lem:first determinant evaluation} in Appendix~\ref{app:det}.
\end{proof}

\subsection{Asymptotics, case $\mult_{j, j+1} = \min \{ s_j, s_{j+1}\}$}

To treat this case, we will need some further combinatorial facts, discussed next.

\subsubsection{More on inverse Fomin type sums}

We begin with yet another property of general inverse Fomin type sums.
For background, let us recall Fomin's original observation~\cite{Fomin:LERW_and_total_positivity}: 
for the $N$-link \quote{rainbow pattern} 
%$\rainbow{N} := \{ \{ 1, 2N\}, \{ 2, 2N -1 \}, \ldots, \{ N, N+1\} \}$ 
\gls{symb:rainbow} $\{ \{ 1, 2N\}, \{ 2, 2N -1 \}, \ldots, \{ N, N+1\} \}$  
(see Figure~\ref{fig:watermelon UST}), 
the inverse Fomin type sum becomes just a single determinant, with rows labelled by the left endpoints of the links, 
$\{1, 2, \ldots, N\}$, and columns by the right ones, $\{2N, 2N -1, \ldots, N+1\}$:
\begin{align*}
\mathfrak{Z}_\rainbow{N}^{\mathfrak{K}} = \det \big(\mathfrak{K}(i, 2N + 1 - j) \big)_{i, j = 1}^N.
\end{align*}
More generally, for a set $L =( \ell_0 + 1, \ldots, \ell_0 + 2m ) = (\ell_1, \ldots, \ell_{2m}) \subset \{ 1, \ldots, 2N\}$
of $2m \leq 2N$ consecutive integers, let us denote
\begin{align*}
\mathfrak{Z}_\rainbow{L}^{\mathfrak{K}} = \det \big(\mathfrak{K}(\ell_i, \ell_{2m + 1 - j}) \big)_{i, j = 1}^m.
\end{align*}

For another piece of notation, 
given two disjoint subsets $I,K \subset \{ 1, \ldots, 2N\}$ of size $m$, let us denote by $[\mathfrak{Z}_\alpha^{\mathfrak{K}}]_{I,K}$ the terms in the polynomial $\mathfrak{Z}_\alpha^{\mathfrak{K}}$ that contain a product $\prod_{r=1}^m \mathfrak{K}(i_r, k_{\tau(r)})$, for some permutation $\tau$ of $m$ elements. 
(One may think of these kernel entries as diverging in the considered asymptotics.) 
We can now state the key combinatorial lemma.

\begin{lemma}
\label{lem:diverging asymptotics}
Let $L = (\ell_1, \ldots, \ell_{2m}) \subset \{1, \ldots, 2N \}$ be a set of $2m$ subsequent integers. 
Suppose that $\alpha \in \LP_N$ contains the rainbow links $\{ \{ \ell_1, \ell_{2m}\},  \{ \ell_2, \ell_{2m-1}\},\ldots, \{ \ell_m, \ell_{m+1} \} \}$ on $L$, and denote $I = \{ \ell_1, \ldots, \ell_m\}$ and $K = \{ \ell_{m+1}, \ldots, \ell_{2m} \}$. 
Then, we have
\begin{align*}
[\mathfrak{Z}_\alpha^{\mathfrak{K}}]_{I,K} = \mathfrak{Z}_\rainbow{L}^{\mathfrak{K}} \; \mathfrak{Z}_{\alpha \setminus \alpha \setminus \rainbow{L}}^{\hat{\mathfrak{K} }},
\end{align*}
where $\alpha \setminus \rainbow{L}$ is the link pattern of $N-m$ links 
obtained by removing the links $ \rainbow{L}$ from $\alpha$ and re-labeling the remaining boundary points in order, 
and $\hat{\mathfrak{K} }$ is the kernel on $\{1, \ldots, 2(N-m) \}^2$ obtained by removing the entries $L$ from $\mathfrak{K}$ and re-labeling the remaining indices in order.
\end{lemma}

\begin{proof}
The case $m=1$ can be found in~\cite[Proposition~2.21]{KKP:Boundary_correlations_in_planar_LERW_and_UST}:
\begin{align*}
[\mathfrak{Z}_\alpha^{\mathfrak{K}}]_{\{ j\}, \{ j+1\} } 
= \mathfrak{K}(j, j+1) \, \mathfrak{Z}_{\alpha \setminus \{j, j+1 \}}^{\mathfrak{K}\setminus \{j, j+1 \} }.
\end{align*}
Using this case iteratively, one finds the terms in the polynomial $\mathfrak{Z}_\alpha^{\mathfrak{K}}$ containing the product 
\begin{align*}
\mathfrak{K}(\ell_m, \ell_{m+1}) \, \mathfrak{K}(\ell_{m-1}, \ell_{m+2}) \, \cdots \, \mathfrak{K}(\ell_1, \ell_{2m}) = \prod_{r=1}^m \mathfrak{K}(\ell_r, \ell_{2m+1-r}) 
\end{align*} 
as
\begin{align*}
 \mathfrak{Z}_{\alpha \setminus \rainbow{L}}^{ \hat{\mathfrak{K} } }  \prod_{r=1}^m \mathfrak{K}(\ell_r, \ell_{2m+1-r}) , \qquad m \geq 1 .
\end{align*}
More generally, let us identify permutations $\tau$ that only permute $K=\{ \ell_{m +1}, \ldots, \ell_{2m} \}$ with permutations $\tau'$ of $r \in \{ 1, 2, \ldots, m\}$ so that $\ell_{2m+1-r}$ is mapped by $\tau$ to $\ell_{2m+1-\tau'(r)}$. 
We denote by $\mathfrak{K}^\tau$ the symmetric kernel obtained as $\mathfrak{K}^\tau (i, j) = \mathfrak{K}(\tau(i), \tau(j))$. 
Then, we have 
\begin{align}
\label{eq:kernel entry product 2}
\prod_{r=1}^m \mathfrak{K}(\ell_r, \ell_{2m+1-\tau'(r)}) = \prod_{r=1}^m \mathfrak{K}^\tau(\ell_r, \ell_{2m+1-r}).
\end{align}
It now follows from the antisymmetry property~\eqref{eq:simple anti-symmetry of inv Fomin} that
\begin{align*}
\mathfrak{Z}_\alpha^{\mathfrak{K}^\tau} = \sign(\tau) \, \mathfrak{Z}_\alpha^{\mathfrak{K}} = \sign(\tau') \, \mathfrak{Z}_\alpha^{\mathfrak{K}}.
\end{align*}
In particular, the terms in the polynomial $\mathfrak{Z}_\alpha^{\mathfrak{K}}$ 
containing the product~\eqref{eq:kernel entry product 2} are $\sign(\tau')$ times the terms of $\mathfrak{Z}_\alpha^{\mathfrak{K}^\tau}$ 
containing $\prod_{r=1}^m \mathfrak{K}^\tau(\ell_r, \ell_{2m+1-r})$, i.e.,
\begin{align*}
\sign(\tau') \, \mathfrak{Z}_{\alpha \setminus \rainbow{L}}^{ \widehat{\mathfrak{K}^\tau } }  \, \prod_{r=1}^m \mathfrak{K}^\tau (\ell_r, \ell_{2m+1-r}) 
= \sign(\tau') \, \mathfrak{Z}_{\alpha \setminus \rainbow{L}}^{ \hat{\mathfrak{K} } }  \, \prod_{r=1}^m \mathfrak{K} (\ell_r, \ell_{2m+1-\tau'(r)}).
\end{align*}
Finally, $[\mathfrak{Z}_\alpha^{\mathfrak{K}}]_{I,K}$ is by definition obtained by summing over $\tau'$, i.e.,
\begin{align*}
[\mathfrak{Z}_\alpha^{\mathfrak{K}}]_{I,K} 
= \; & \mathfrak{Z}_{\alpha \setminus \rainbow{L}}^{ \hat{\mathfrak{K} } }  \sum_{\tau'} \sign(\tau') \prod_{r=1}^m \mathfrak{K} (\ell_r, \ell_{2m+1-\tau'(r)}) \\
= \; &
\mathfrak{Z}_{\alpha \setminus \rainbow{L}}^{ \hat{\mathfrak{K} } } \det \big(\mathfrak{K} (\ell_i, \ell_{2m+1-j}) \big)_{i,j=1}^m \\
= \; &
\mathfrak{Z}_{\alpha \setminus \rainbow{L}}^{ \hat{\mathfrak{K} } } \mathfrak{Z}_\rainbow{L}^{\mathfrak{K}}.
\end{align*}
This finishes the proof.
\end{proof}

\subsubsection{Application to Theorem~\ref{thm:ASY}}

By symmetry, we may assume that $s_{j+1} \leq s_j$.

\begin{proposition}
In the setup of Theorem~\ref{thm:ASY}, assuming $\mult_{j, j+1} = s_{j+1} \leq s_j$, we have
\begin{align*}
\lim_{x_{j+1}, x_j \to \xi} 
\frac{\PartF_{\alpha} (\bs{x}) }{|x_{j+1} - x_j |^{ -s_j s_{j+1} - s_{j+1} } }
= \; & C(s_j, s_{j+1}) \, \PartF_{\hat{\alpha}} (\hat{\bs{x}}) ,
\\ \nonumber
\textnormal{where} 
\qquad 
C(s_j, s_{j+1}) = C(s_j, s_{j+1},s_{j+1}) 
= \; &  \frac{1}{\pi^{s_{j+1}}} \prod_{\ell=1}^{s_{j+1}} (s_{j} - s_{j+1} + \ell)! \times \prod_{r=1}^{s_{j+1}} (r-1)! .
\end{align*}
\end{proposition}

\begin{proof}
By the inverse Fomin type expression for $\PartF_{\alpha}(\bs{x})$ in~\eqref{eq: valenced kernel}, 
it is exactly the kernels $\mathfrak{K}(\iota, k)$, $\iota \in \{ \summ_{j-1} + 1 \, \ldots, \summ_j \}$ and $k \in K$, 
that diverge in the limit $x_{j+1}, x_j \to \xi$. Furthermore, the higher the values of $\iota$ and $k$, the worse the divergence. Hence, with $s_{j+1} \leq s_j$, the leading asymptotics of $\PartF_{\alpha} (\bs{x})$ in this limit 
is determined by $[\mathfrak{Z}_\alpha^{\mathfrak{K}}]_{I,K}$.
Now, Lemma~\ref{lem:diverging asymptotics} yields
\begin{align*}
[\mathfrak{Z}_\alpha^{\mathfrak{K}}]_{I,K} 
= \mathfrak{Z}_\rainbow{L}^{\mathfrak{K}} \, \mathfrak{Z}_{\alpha \setminus \rainbow{L}}^{\hat{\mathfrak{K} }},
\end{align*}
where one readily observes that 
$\mathfrak{Z}_{\alpha \setminus \rainbow{L}}^{\hat{\mathfrak{K} }} (\hat{\bs{x}}) = \PartF_{\hat{\alpha}} (\hat{\bs{x}})$. Also $\mathfrak{Z}_\rainbow{L}^{\mathfrak{K}}$ can be made fully explicit.

First, note that
\begin{align*}
\ExcKH (x_{j+1}, x_j) = \frac{1}{\pi} \frac{1}{(x_{j+1} - x_j)^2}, \qquad
\partial_{j+1}^{m_{j+1}} \partial_{j}^{m_{j}} \ExcKH (x_{j+1}, x_j) 
= \frac{1}{\pi} \frac{(-1)^{m_{j+1}} (m_j + m_{j+1} +1)!}{(x_{j+1} - x_j)^{m_j + m_{j+1} + 2 }} .
\end{align*}
Denoting the elements of $I$ and $K$ in order by $i_p$ and $k_q$, for $1 \leq p, q \leq s_{j+1}$, 
and denoting $\delta := x_{j+1} - x_j$ and $s_j - s_{j+1} = t$, we have
\begin{align*}
\mathfrak{K}(i_p, k_q) = \partial_{j+1}^{q-1} \partial_{j}^{t + p -1 } \ExcKH (x_{j+1}, x_j) =
\frac{1}{\pi} \frac{(-1)^{q-1} (p + q +t - 1)!}{\delta^{p+q + t}} .
\end{align*}
Elementary operations on determinants now yield 
(denoting $s_{j+1} =: m$ for lighter notation)
\begin{align*}
\mathfrak{Z}_\rainbow{L}^{\mathfrak{K}}
= \; & \det (\mathfrak{K}(i_r, k_{m-s}))_{r, s = 1}^m 
\\
= \; & (-1)^{\lfloor m/2 \rfloor} \det (\mathfrak{K}(i_r, k_{s}))_{r, s = 1}^m 
\\
= \; & (-1)^{\lfloor m/2 \rfloor} \det \Big( \frac{1}{\pi} \frac{(-1)^{s-1} (r+s+t - 1)!}{\delta^{r+s+ t}} \Big)_{r, s = 1}^m \\
= \; & \frac{1}{\pi^m} \, \delta^{-tm - m(m+1)} \det ((r+s+t - 1)!)_{r, s = 1}^m
\end{align*}
The determinant above has been evaluated in Lemma~\ref{lem:second determinant evaluation}, which completes the proof.
\end{proof}

\subsection{Asymptotics, general case}
\label{subsec: asy gen}

The idea of the proof in the general case is to first analyze a reference term that we expect to be one of those that exhibit the worst divergence as $x_j, x_{j+1} \to \xi$, and then relate it to other diverging terms.

\subsubsection{The reference term}

The idea is to introduce a permutation $\tau$ that collects the worst divergence to the entries $\mathfrak{K}^\tau(i, k)$, with 
\begin{align*}
i \in \{ \summ_j -\mult_{j, j+1} + 1, \ldots, \summ_j \} =:I
\qquad \textnormal{and} \qquad 
k \in \{ \summ_j + 1, \ldots, \summ_j + \mult_{j, j+1} \} =: K .
\end{align*}
Note that by the antisymmetry property~\eqref{eq:simple anti-symmetry of inv Fomin}, 
we have $\mathfrak{Z}^{\mathfrak{K}^\tau}_\alpha = \sign(\tau) \mathfrak{Z}^{\mathfrak{K}}_\alpha$, 
so we may analyze $\mathfrak{Z}^{\mathfrak{K}^\tau}_\alpha$ equally well. 
Now, one such permutation $\tau $ only acts on the indices $ \summ_j + 1, \ldots, \summ_{j+1} $, permuting them to $\summ_{j+1} -\mult_{j, j+1} + 1, \summ_{j+1} -\mult_{j, j+1} + 2, \ldots, \summ_{j+1} , \summ_j + 1, \summ_j + 2, \ldots, \summ_{j+1} -\mult_{j, j+1}$. 
By Lemma~\ref{lem:diverging asymptotics}, we have
\begin{align*}
[\mathfrak{Z}_\alpha^{\mathfrak{K}^\tau}]_{I,K} 
= \mathfrak{Z}_\rainbow{L}^{\mathfrak{K}^\tau} \, \mathfrak{Z}_{\alpha \setminus \rainbow{L}}^{\widehat{\mathfrak{K}^\tau }},
\end{align*}
where by the choice of $\tau$ and then by Proposition~\ref{prop:ASY no links case}, we have
\begin{align}
\begin{split}
\mathfrak{Z}_{\alpha \setminus \rainbow{L}}^{\widehat{\mathfrak{K}^\tau }}
= \; & \PartF_{\alpha \setminus \rainbow{L}} (\bs{x}) \\
\label{eq:ASY general case easy factor}
= \; & (1+ O (\delta)) \, C(s_j - \mult_{j, j+1}, s_{j+1} - \mult_{j, j+1}) \,  \delta^{(s_j - \mult_{j, j+1})(s_{j+1} - \mult_{j, j+1})} \,  \PartF_{\hat{\alpha}} (\hat{\bs{x}}) ,
\end{split}
\end{align}
where $\delta := x_{j+1} - x_j$. 
On the other hand, denoting 
\begin{align*}
i_p =: \summ_j -\mult_{j, j+1} + p
\qquad \textnormal{and} \qquad 
k_p = \summ_{j+1} -\mult_{j, j+1} + p, \quad 1 \leq p \leq \mult_{j, j+1} , 
\end{align*}
we have
\begin{align*}
\mathfrak{K}(i_p, k_q) 
= \; & \partial_{j+1}^{s_j -\mult_{j, j+1} + q - 1} \partial_{j}^{s_{j+1} -\mult_{j, j+1} + p - 1} \ExcKH (x_{j+1}, x_j) \\
= \; &
\frac{1}{\pi} \frac{(-1)^{s_j -\mult_{j, j+1} + q-1} ( s_j -\mult_{j, j+1} + q + s_{j+1} -\mult_{j, j+1} + p - 1)!}{\delta^{s_j + s_{j+1} -2\mult_{j, j+1} + p + q }} ,
\end{align*}
and by definition, we have
\begin{align*}
\mathfrak{Z}_\rainbow{L}^{\mathfrak{K}^\tau} 
= \; & (-1)^{\lfloor \mult_{j, j+1} /2 \rfloor} \det(\mathfrak{K}(i_p, k_q))_{p, q =1}^{\mult_{j, j+1}} \\
= \; & (-1)^{s_j -\mult_{j, j+1}} \delta^{-\mult_{j, j+1}(s_j + s_{j+1} -2\mult_{j, j+1}) - \mult_{j, j+1} (\mult_{j, j+1} + 1 )} \\
\; & \times \det (( s_j -\mult_{j, j+1} + q + s_{j+1} -\mult_{j, j+1} + p - 1)!)_{p, q =1}^{\mult_{j, j+1}}.
\end{align*}
Without analyzing the determinant here, we can already conclude that
\begin{align} \label{eq:ASY general case reference term conclusion}
[\mathfrak{Z}_\alpha^{\mathfrak{K}^\tau}]_{I,K}
= C(s_j, s_{j+1}, \mult_{j, j+1}) \, \delta^{s_j s_{j+1}-2\mult_{j, j+1}(s_j + s_{j+1} -\mult_{j, j+1}) - \mult_{j, j+1} } \, \PartF_{\hat{\alpha}} (\hat{\bs{x}}).
\end{align}
Note that this is the power of $\delta$ is as claimed in Theorem~\ref{thm:ASY}.

\subsubsection{A general term}

The treatment of a general term is inductive. 
Suppose that $[\mathfrak{Z}_\alpha^{\mathfrak{K}}]_{I,K}$ has been analyzed for some subsets $I \subset \{ \summ_{j-1} + 1, \ldots, \summ_j \}$ and $ K \subset \{ \summ_{j} + 1, \ldots, \summ_{j+1} \}$ 
(above, we considered $I= \{ \summ_j -\mult_{j, j+1} + 1, \ldots, \summ_j \} $ and $K= \{ \summ_j + 1, \ldots, \summ_j + \mult_{j, j+1} \} $) yielding~\eqref{eq:ASY general case reference term conclusion} (possibly with the constant zero).
Next, consider the change of lowering one index in one of the sets by one unit; denote the new index sets by $I'$ and $K'$. 
Now, similarly to the reference term, there is a permutation $\tau$ such that $\tau(I') = \{ \summ_j -\mult_{j, j+1} + 1, \ldots, \summ_j \}$ and $\tau(K') =  \{ \summ_j + 1, \ldots, \summ_j + \mult_{j, j+1} \}$ and 
\begin{align*}
[\mathfrak{Z}_\alpha^{\mathfrak{K}^\tau}]_{I',K'} 
= \mathfrak{Z}_\rainbow{L}^{\mathfrak{K}^\tau}\, \mathfrak{Z}_{\alpha \setminus \rainbow{L}}^{\widehat{\mathfrak{K}^\tau }}.
\end{align*}
One then readily verifies that, 
when changing from $I,K$ to $I', K'$, 
the expression $\mathfrak{Z}_{\alpha \setminus \rainbow{L} }$ 
will still be a negative-power monomial in $\delta$, the power of $\delta$ increasing by one, i.e., 
the monomial exhibits a milder divergence upon $\delta \downarrow 0$ than the earlier one, 
and the monomial is multiplied by a constant that depends on $s_j$, $s_{j+1}$, $\mult_{j, j+1}$, $I'$, and $K'$.

Let us then consider $\mathfrak{Z}_{\alpha \setminus \rainbow{L}}^{\widehat{\mathfrak{K}^\tau }}$. 
Looking into how~\eqref{eq:ASY general case reference term conclusion} was derived in~\eqref{eq:Taylor series for kernels} 
and below it, this means that one less order of differentiation will be needed in the Taylor polynomials, i.e., 
the result will be the same (cf.~\eqref{eq:ASY general case easy factor}) but converging to zero with 
a polynomial rate in $\delta$ that is one power lower\footnote{The order of the Taylor polynomial, of course, cannot decrease below zero --- in that case, the rate of convergence is not decreased, and $I', K'$ will not contribute to the renormalized limit. The induction hypothesis~\eqref{eq:ASY general case reference term conclusion} is still true, but with the constant being zero.} 
(viz.~converges to zero slower). 
Combining the above analyses concludes the induction step, and the entire proof of Theorem~\ref{thm:ASY}.

\subsection{Simultaneous fusion}
\label{subsec:ASY simul}

Recall that in Proposition~\ref{prop:higher-valence solution limit} we obtained, for $\alpha \in \LP_\multii$, expressions for iterated limits of the unfused partitions functions $\PartF_{\imath(\alpha)}$ when fused to the valence $\multii$. 
(Here, $\imath$ is the unfusing map, see Figures~\ref{fig: UST fused} and~\ref{fig:fusion}.)
The expression was an inverse Fomin type sum, with a kernel denoted $\mathfrak{K}'$ 
below, and up to a constant factor it also coincided with the fused partition function $\PartF_{\alpha}$ (see Lemma~\ref{lem:fusion of fused partition functions}). 
Next, we derive the same expression as a suitably renormalized limit of $\PartF_{\imath(\alpha)}$ under a simultaneous fusion. 
This provides an SLE interpretation of the fused partition functions $\PartF_{\alpha}$ in Appendix~\ref{subsec:Fusing endpoints}.

\begin{proposition}
\label{prop:fused SLE}
In the setup of Proposition~\ref{prop:higher-valence solution limit}, assume furthermore that for any $b \not \in J =J_j$,
the kernel $\mathfrak{K}(a, b) = C_b(x_b - x_a)^{-p_b}$
is given by the same expression for all $a \in J$, where $C_b \neq 0$  and $p_b \geq 0$ are constants depending only on $b$. 
Denote by $\mathfrak{K}'$ the modified kernel described in the end of Proposition~\ref{prop:higher-valence solution limit} \textnormal{(}with $\xi =  x_{\summ_{j-1}+1}$\textnormal{)}, 
and denote $\delta_k := x_{\summ_{j-1}+k}-x_{\summ_{j-1}+1}$. 
Then, we have
\begin{align*}
\mathfrak{Z}^{\mathfrak{K}}_\alpha (x_1, \ldots, x_{2N}) = 
\prod_{p=2}^{s_j} \delta_p \prod_{2 \leq q < r \leq s_j} (\delta_r-\delta_q) \times \mathfrak{Z}^{\mathfrak{K}' }_\alpha ( \hat{\bs{x}} )  \; (1+o(1)) , \qquad \textnormal{as } \delta_{s_j} \downarrow 0 ,
\end{align*}
where $\hat{\bs{x}} := (x_1, \ldots, x_{\summ_{j-1}}, x_{\summ_{j-1}+1}, x_{\summ_{j}+1}, \ldots, x_{2N}) $ and the Landau term $o(1)$ converges to zero as $\delta_{s_j} \downarrow 0$ uniformly over $\hat{\bs{x}} \in K$ for any fixed compact set $ K \subset \chamber_{2N-s_j+1} $.
\end{proposition}

Starting from the \quote{unfused kernel} $\mathfrak{K}(a, b) = \tfrac{1}{\pi}(x_b - x_a)^{-2}$ and looking at the definition of $\mathfrak{K}'$ in Proposition~\ref{prop:higher-valence solution limit}, it is not hard to see that 
this result can be used iteratively to handle multiple sets of points to be fused, simultaneously or iteratively.

\begin{proof}[Proof of Proposition~\ref{prop:fused SLE}]
Denote by $i_1 := \summ_{j-1} + 1, \, \ldots, \, i_{s_j} := \summ_{j-1} + s_j = \summ_{j}$  the boundary indices to be fused. Denote also $J^c := \{ 1,2, \ldots, 2N\} \setminus J$. 
Since $\alpha \in \LP_\multii$, Lemma~\ref{lem:Fomin properties}\ref{item:general zero-replacing rule b} guarantees that  
we may assume that $\mathfrak{K}(a, b) = 0$ for all $a, b \in J$.
Let us then define 
\begin{align}\label{eq:Kprimeprime}
\mathfrak{K}''(a, b) :=  \prod_{c \in J^c} \bigg( \frac{(x_c-x_a)^{p_c}}{C_c}\bigg)^{\onesmall\{ a \in J \}} \bigg(  \frac{(x_c-x_b)^{p_c}}{C_c}\bigg)^{\onesmall\{ b \in J \}}  \mathfrak{K}(a, b) ,
\end{align}
where $\one\{a \in J\} := 1$ when $a \in J$ and $\one\{a \in J\} := 0$ otherwise
--- in other words, we have
\begin{align}
\label{eq:Kprimeprime equiv def}
\mathfrak{K}''(a, b) = \prod_{\substack{c \in J^c \setminus \{ b \} }}  \frac{(x_c-x_a)^{p_c} }{C_c} , \qquad \textnormal{for } a \in J, \, b \in J^c ,
\end{align}
and $\mathfrak{K}''(a, b) = 0$, for $a, b \in J$; and $\mathfrak{K}''(a, b) = \mathfrak{K}(a, b)$, for $a, b \in J^c$.

On the one hand, the first definition~\eqref{eq:Kprimeprime} of $\mathfrak{K}''$ above can be seen as scaling $\mathfrak{K}(a, b)$ for one factor for $a$ and another similar factor for $b$.
From the linearity of determinants with respect to each row and column, one then readily concludes that
\begin{align*}
\mathfrak{Z}_\alpha^{\mathfrak{K}''} = \prod_{b \in J} \prod_{c \in J^c}  \frac{(x_c-x_b)^{p_c}}{C_c}  \, \mathfrak{Z}_\alpha^{\mathfrak{K}}.
\end{align*} 
We will thus analyze the fusion asymptotics via $\mathfrak{Z}^{\mathfrak{K}''}_\alpha$.
On the other hand, we observe from~\eqref{eq:Kprimeprime equiv def} that regarding $(x_c)_{c \in J^c}$ as constants in the fusion limit, for each fixed $b$ and all $a \in J$, we have
\begin{align*}
\mathfrak{K}''(a,b) = \mathfrak{K}''(b,a) 
=f_b (x_{a}), \qquad b \in \{1,2,\ldots,2N\} , \, a \in J ,
\end{align*}
where $f_b$ is some polynomial that depends on $b$ and on $(x_c)_{c \in J^c}$ but not on $(x_{a'})_{a' \in J, a' \neq a}$.
Since the function $f_b$ (although not its argument $x_a$) is the same for all $a \in J$, we Taylor expand for all $i_p$, $1 \leq p \leq s_j$:  
\begin{align} \label{eq:Taylor series for kernels 2}
\begin{split} 
\mathfrak{K}''(i_p, b) = \; & f_b(x_{i_1})  + \delta_p f_b'(x_{i_1}) + \tfrac{\delta_p^2}{2} f_b''(x_{i_1}) + \cdots 
+
\tfrac{\delta_p^{M}}{M!} f_b^{(M)}(x_{i_1}) , \qquad \delta_p = x_{i_p}-x_{i_1} .
\end{split}
\end{align}
(It will be crucial in what follows that $f_b$ is a polynomial, i.e., this expansion terminates at some finite order $M$.)
Similarly as in the proof of Proposition~\ref{prop:ASY no links case}, we can split $\mathfrak{Z}^{\mathfrak{K}''}_\alpha$ as 
a sum of different inverse Fomin type sums $\mathfrak{Z}^{\widetilde{\mathfrak{K}}}_\alpha$, 
obtained by choosing $\widetilde{\mathfrak{K}}(i_p, \cdot)$ for each $p$ to be one of the terms in 
this Taylor expansion (and $\widetilde{\mathfrak{K}} (a, b) = \mathfrak{K}'' (a, b)$ if $a, b \in J^c$). 
Again, similarly to Proposition~\ref{prop:ASY no links case}, the sums where the same order of derivative appears as $\widetilde{\mathfrak{K}}(i_p, \cdot)$ for two different values of $p$ are zero by Lemma~\ref{lem:Fomin properties}\ref{item:general zero-replacing rule b}. In other words, the inverse Fomin type sums $\mathfrak{Z}^{\widetilde{\mathfrak{K}}}_\alpha$ with (at least potentially) nonzero contribution can be indexed by increasing sets $0 \leq \ell_1 < \cdots < \ell_{s_j} \leq M$, describing the appearing derivative orders, and permutations $\tau $ of $s_j$ elements, describing which $i_p$ corresponds to derivative of order $\ell_{\tau(p)}$. 
In other words:
\begin{align*}
\widetilde{\mathfrak{K}}_\tau (i_p, b) = \frac{\delta_p^{\ell_{\tau(p)}}}{\ell_{\tau(p)} !} \, f_b^{(\ell_{\tau(p)})}(x_{i_1}),
\end{align*}
where only the dependence on $\tau$ is now explicated (and $\delta_1^0=1$).

Let us now fix a sequence $\boldsymbol{\ell} = (\ell_1 < \cdots < \ell_{s_j})$ and define yet another kernel depending on these indices as 
$\mathfrak{K}_{\boldsymbol{\ell}} (i_p, b) = \tfrac{1}{\ell_{p} !} \, f_b^{(\ell_{p})}(x_{i_1})$ if $i_p \in J$, and $\mathfrak{K}_{\boldsymbol{\ell}} (a, b) = \widetilde{\mathfrak{K}} (a, b)$ 
 if $a, b \in J^c$. Hence, if $\tau=\mathrm{id}$ is the identity permutation, we observe that
\begin{align*}
\mathfrak{Z}^{ \widetilde{\mathfrak{K}}_\mathrm{id} }_\alpha = \prod_{p=1}^{s_j} \delta_p^{\ell_p} \times \mathfrak{Z}^{\mathfrak{K}_{\boldsymbol{\ell}}}_\alpha.
\end{align*}
Arguing with the antisymmetry property~\eqref{eq:simple anti-symmetry of inv Fomin} as in the proof of Proposition~\ref{prop:ASY no links case}, we see that for a general permutation $\tau \in \SymGrp_{s_j}$, where 
$\SymGrp_{s_j} := \langle \tau_{i_1}, \dots, \tau_{i_{s_j-1}} \rangle \subset \SymGrp_{2N}$ is the subgroup generated by the transpositions associated to the indices $(i_1, \dots, i_{s_j} )$,  
we have
\begin{align}
\label{eq:Ztau}
\mathfrak{Z}^{ \widetilde{\mathfrak{K}}_\tau }_\alpha
= \; & \sign(\tau) \, \prod_{p=1}^{s_{j}} \delta_p^{\ell_{\tau(p)}} \times \mathfrak{Z}^{\mathfrak{K}_{\boldsymbol{\ell}}}_\alpha .
\end{align}
The contribution of a fixed index sequence $\boldsymbol{\ell}$ to $\mathfrak{Z}^{\mathfrak{K}''}_\alpha$ is thus given by
\begin{align*}
\sum_{\tau \in \SymGrp_{s_j}}  \sign(\tau)  \prod_{p=1}^{s_{j}}
\delta_p^{\ell_{\tau(p)}} \times \mathfrak{Z}^{\mathfrak{K}_{\boldsymbol{\ell}} }_\alpha 
= \det \big( \delta_{k}^{\ell_m} \big)_{k, m = 1}^{s_j} \times \mathfrak{Z}^{\mathfrak{K}_{\boldsymbol{\ell}} }_\alpha  .
\end{align*}
From the linearity of the determinant, it is now clear that the index sequence $\boldsymbol{\ell}^*$ where $\ell^*_m = m-1$ dominates all other ones in the limit 
(provided that $\mathfrak{Z}^{\mathfrak{K}_{\boldsymbol{\ell}^*} }_\alpha \neq 0$, which will become clear in the end of the proof). 
As there are only finitely many index sequences, 
\begin{align*}
\mathfrak{Z}^{\mathfrak{K}'' }_\alpha 
= \; & \det \big( \delta_{k}^{m-1} \big)_{k, m = 1}^{s_j}  \times  \mathfrak{Z}^{\mathfrak{K}_{\boldsymbol{\ell}^*} }_\alpha \;  (1+O(\delta_{s_j}))
\\
= \; & \det \big( \delta_{k}^{m-1} \big)_{k, m = 2}^{s_j} \times  \mathfrak{Z}^{\mathfrak{K}_{\boldsymbol{\ell}^*} }_\alpha  \; (1+O(\delta_{s_j})) \\
= \; &  \prod_{p=2}^{s_j} \delta_p \prod_{2 \leq q < r \leq s_j} (\delta_r-\delta_q) \times  \mathfrak{Z}^{\mathfrak{K}_{\boldsymbol{\ell}^*} }_\alpha \;  (1+O(\delta_{s_j})) , \qquad \textnormal{as } \delta_{s_j} \downarrow 0 ,
\end{align*}
where we use Vandermonde's determinant formula in the last step, and the Landau term $O(\delta_{s_j})$ 
can be bounded as $O(\delta_{s_j}) \leq C \delta_{s_j}$, where $C$ is a uniform constant over compact sets 
(actually, the Landau term is a polynomial in $\delta_1, \ldots, \delta_{s_j}$ with no constant term).

Summarizing everything so far, we have obtained the approximation
\begin{align}
\nonumber
\mathfrak{Z}^{\mathfrak{K}}_\alpha 
= \; & \prod_{b \in J} \prod_{c \in J^c}  \frac{C_c}{(x_c-x_b)^{p_c}}  \mathfrak{Z}^{\mathfrak{K}''}_\alpha \\
\label{eq:blabla}
= \; &  \prod_{b \in J} \prod_{c \in J^c} \frac{C_c}{(x_c-x_b)^{p_c}} \prod_{p=2}^{s_j} \delta_p \prod_{2 \leq q < r \leq s_j} (\delta_r-\delta_q) \times \; \mathfrak{Z}^{\mathfrak{K}_{\boldsymbol{\ell}^*} }_\alpha \;  (1+O(\delta_{s_j})) .
\end{align}
Now, the entries of the matrices in $\mathfrak{Z}^{\mathfrak{K}_{\boldsymbol{\ell}^*} }_\alpha$ are rational functions and do not depend at all on $x_{i_2}, \ldots, x_{i_{s_j}}$ 
--- in particular, they do not diverge in the fusion limit.
Hence, because Landau term is a polynomial, the right-hand side above is well defined and continuous up to and including the fusion point. 
A direct computation then gives the iterated limit
\begin{align*}
& \bigg( \lim_{x_{i_{\summ_j}} \to x_{i_1 }} 
|x_{i_{\summ_j} } - x_{i_1}|^{-s_j+1} \ldots \lim_{ x_{i_{\summ_{j-1}+2}} \to x_{i_1 } } |  x_{i_{\summ_{j-1}+2}} -  x_{i_1 } |^{-1} \lim_{ x_{i_{\summ_{j-1}+1}} \to  x_{i_1 } } \bigg)
\mathfrak{Z}^{\mathfrak{K}}_\alpha (x_1, \ldots, x_{2N}) \\
&=
\bigg( \prod_{b \in J} \prod_{c \in J^c} \frac{C_c}{(x_c-x_b)^{p_c} } \, \mathfrak{Z}^{\mathfrak{K}_{\boldsymbol{\ell}^*} }_\alpha(x_1, \ldots, x_{2N})
\bigg)\bigg|_{x_{i_1} \, = \, \cdots \, = \, x_{i_{s_j} }}
\end{align*}
On the other hand, by Proposition~\ref{prop:higher-valence solution limit}, this limit is also given by $\mathfrak{Z}^{\mathfrak{K}'}_\alpha (\hat{\bs{x}})$, where $\mathfrak{K}'$ is the kernel described in the end of that proposition. 
We thus conclude that 
\begin{align*}
\bigg( \prod_{b \in J} \prod_{c \in J^c} \frac{C_c}{(x_c-x_b)^{p_c} } \, \mathfrak{Z}^{\mathfrak{K}_{\boldsymbol{\ell}^*} }_\alpha
\bigg)\bigg|_{x_{i_1} \, = \, \cdots \, = \, x_{i_{s_j} }} = \mathfrak{Z}^{\mathfrak{K}'}_\alpha (\hat{\bs{x}}).
\end{align*}

Finally, because the left-hand side above is continuous around the evaluation point, 
perturbation of its variables within a compact set $K$, as in the statement, makes a uniformly small error. 
Furthermore, the right-hand side is continuous and positive by Theorem~\ref{thm:CFT properties}\textnormal{(POS)}, so this is also yields a uniformly small \textit{relative} error. Thus, in such compact sets, we can approximate uniformly
\begin{align*}
\prod_{b \in J} \prod_{c \in J^c} \frac{C_c} {(x_c-x_b)^{p_c} } \, \mathfrak{Z}^{\mathfrak{K}_{\boldsymbol{\ell}^*} }_\alpha
=\mathfrak{Z}^{\mathfrak{K}'}_\alpha (\hat{\bs{x}}) \;  (1+o(1)) , \qquad \textnormal{as } \delta_{s_j} \downarrow 0 .
\end{align*}
Plugging this into~\eqref{eq:blabla}, we obtain
\begin{align*}
\mathfrak{Z}^{\mathfrak{K}}_\alpha =  
 \prod_{p=2}^{s_j} \delta_p \prod_{2 \leq q < r \leq s_j} (\delta_r-\delta_q)  \mathfrak{Z}^{\mathfrak{K}'}_\alpha (\hat{\bs{x}}) \; (1+o(1)) , \qquad \textnormal{as } \delta_{s_j} \downarrow 0 .
\end{align*}
This concludes the proof.
\end{proof}

%% file: tex/app-det.tex
\begin{lemma}
\label{lem:first determinant evaluation}
Fix two natural numbers $s_{j+1} \leq s_j$, and consider the matrix $A \in \bR^{s_{j+1} \times s_{j+1}}$ 
with entries $A_{\ell, k} = \frac{1}{(s_j - \ell + k)!}$. Then, we have
\begin{align*}
\det A  =  \frac{ \prod_{k = 1}^{s_{j+1}} (s_{j+1} - k)! }{\prod_{\ell = 1}^{s_{j+1}} (s_j + s_{j+1} - \ell )!}.
\end{align*}
\end{lemma}

\begin{proof}
Row-wise and column-wise multiplication of $A$ gives
\begin{align*}
B_{\ell, k} := \; & \frac{ (s_j + s_{j+1} - \ell )! }{(s_{j+1} - k)!} \, A_{\ell, k} = \binom{s_j + s_{j+1} - \ell }{s_{j+1} - k},
\\ \qquad \det B = \; & (\det A ) \,\frac{\prod_{\ell = 1}^{s_{j+1}} (s_j + s_{j+1} - \ell )!}{\prod_{k = 1}^{s_{j+1}}(s_{j+1} - k)!} 
\end{align*}
after which a permutation (reversal) of both rows and columns gives $\det C = \det B$, where
\begin{align*}
C_{\ell, k} := B_{s_{j+1} + 1 - \ell, s_{j+1} + 1 - k} = \binom{s_j + \ell - 1}{k-1} .
\end{align*}
It thus suffices to compute $\det C$.
To this end, we find its LU-decomposition by Vandermonde's convolution formula\footnote{In combinatorial terms, the formula states that to choose $k-1$ colored balls from a pool of $s_j$ red balls and $\ell - 1$ blue balls, one chooses $r'$ blue balls, $0 \leq r' \leq \min\{k-1, \ell -1 \}$, and then $k-1-r'$ red balls.}:
\begin{align*}
C_{\ell, k} 
\; = \;  \binom{s_j + \ell - 1}{k-1}
\; = \;  \sum_{r'=0}^{\min \{k-1, \ell -1 \}} \binom{\ell - 1}{r'} \binom{s_j}{k-1-r'}
\; = \;  \sum_{r=1}^{s_{j+1}} \binom{\ell - 1}{r-1} \binom{s_j}{k-r},
\end{align*}
where in the second step we adopted for lighter notation the convention that $\binom{a}{b} =0$ if $b < 0 $ or $b>a$. Defining (in this convention) the matrices $D, E \in \bR^{s_{j+1} \times s_{j+1}}$ via $D_{\ell, r}  = \binom{\ell - 1}{r-1}$ and $E_{r,k} = \binom{s_j}{k-r}$, we observe that
\begin{align*}
C_{\ell, k} = \sum_{r=1}^{s_{j+1}}  D_{\ell, r} E_{r,k} \qquad \textnormal{and} \qquad \det C = (\det D)( \det E).
\end{align*}
Now, note that 
$D$ (resp.~$E$) is lower-unitriangular (resp.~upper-unitriangular), so $\det D =  \det E = 1$. 
The claim now follows from the determinant identities displayed above.
\end{proof}

\begin{lemma}
\label{lem:second determinant evaluation}
Fix $t \in \mathbb{Z}_{\geq -1}$, and consider the matrix $A \in \bR^{m \times m}$ 
with entries $A_{r, s} = (r+s+t - 1)!$. Then, we have
\begin{align*}
\det A  =  \prod_{s=1}^m (s+t)! \times \prod_{r=1}^m (r-1)! .
\end{align*}
\end{lemma}
\begin{proof}
Row-wise and column-wise multiplication yields
\begin{align*}
B_{r,s} := \; & \frac{A_{r, s}}{(s+t)!}\frac{1}{(r-1)!} = \binom{r+s+t-1}{r-1} ,
\\ \qquad \det B = \; & \frac{\det A}{ \left( \prod_{s=1}^m (s+t)! \right) \left( \prod_{r=1}^m (r-1)! \right) } .
\end{align*}
Next, we again apply Vandermonde's convolution formula:
\begin{align*}
\binom{r+s+t-1}{r-1} 
\; = \; \sum_{\ell'=0}^{\min\{r-1, s+t \} } \binom{s+t}{\ell'} \binom{r-1}{r-1-\ell'} 
\; = \; \sum_{\ell=1}^{m } \binom{s+t}{\ell - 1} \binom{r-1}{r-\ell},
\end{align*}
where in the second step we again used the convention that $\binom{a}{b} =0$ if $b < 0 $ or $b>a$. Applying the same formula and notation again, we obtain
\begin{align*}
\binom{s+t}{\ell - 1} 
\; = \; \sum_{\iota' = 0}^{\min\{\ell - 1, s-1 \} } \binom{s - 1}{\iota'} \binom{t + 1}{\ell - 1 - \iota'} 
\;  = \; \sum_{\iota = 1}^{m } \binom{s - 1}{\iota - 1} \binom{t + 1}{\ell - \iota},
\end{align*}
Combining these and defining suitable matrices $C, D, E \in \bR^{m \times m}$, we obtain
\begin{align*}
B_{r,s} \; = \; \binom{r+s+t-1}{r-1} 
\; = \; \sum_{\ell=1}^{m} \sum_{\iota = 1}^{m } \binom{r-1}{r-\ell} \binom{t + 1}{\ell - \iota} \binom{s - 1}{\iota - 1} 
\; = \; \sum_{\ell=1}^{m} \sum_{\iota = 1}^{m } C_{r, \ell} D_{\ell, \iota} E_{\iota, s}.
\end{align*}
Once again, we see that $C, D, E$ are all triangular matrices with all ones on the diagonal.
\end{proof}

%% file: tex/app-coblo.tex
In this short combinatorial appendix, we shall find an alternative, completely explicit basis for the covariant solution space 
$\SolSp_\multii$ of the BPZ PDEs.

Fixing notation, given a link pattern $\alpha \in \LP_N$, we identify it with a non-negative integer walk (i.e., a Dyck path), abusively also denoted as $\alpha \colon \{ 0, 1, \ldots, 2N\} \to \bZnn$ 
and determined by setting $\alpha(0) = 0$, and for $1 \leq j \leq 2N$,
\begin{align*}
\alpha(j)-\alpha(j-1) 
= 
\begin{cases}
1 , & j \textnormal{ is a left link endpoint in the corresponding link pattern,}\\
-1, &  j \textnormal{ is a right link endpoint.}
\end{cases}
\end{align*}
(See, e.g.,~\cite{KKP:Boundary_correlations_in_planar_LERW_and_UST} 
and~\cite[Section~2.4]{Peltola-Wu:Global_and_local_multiple_SLEs_and_connection_probabilities_for_level_lines_of_GFF} 
for illustrations and a thorough discussion.) 
Then, given valences $\multii = (s_1, \ldots, s_d) \in \bZpos^d$ with $s_1 + \cdots + s_d = 2N$, 
for each pair $\alpha, \beta \in \LP_N$ of link patterns, 
we say that $\beta \geq_\multii \alpha$ if (interpreting both $\alpha$ and $\beta$ as Dyck paths) 
\begin{align} \label{eq: partial order}
\alpha \leq \beta  
\qquad \textnormal{and in addition,} \quad \alpha(\summ_j)=\beta(\summ_j) ,
\quad \textnormal{for all } j\in \{ 0, 1, \ldots, d \} .
\end{align}
Note that this is a partial order on $\LP_N$.

Lastly, for each $\multii$-valenced link pattern $\alpha \in \LP_\multii$, 
we define a function %$\CobloF_\alpha \colon \chamber_d \to \bR$ 
\gls{symb:CobloF} $\colon \chamber_d \to \bR$ 
by
\begin{align} \label{eq:fused conformal block function}
\CobloF_\alpha( x_1, \ldots, x_d) 
:= \; & \sum_{ \substack{\beta \in \LP_N \\ \beta \geq_\multii \alpha}}
\Delta^{\mathfrak{K}}_\beta (x_1, \ldots, x_{d}) , 
\end{align}
where $\mathfrak{K}$ is the same kernel as in the definition~\eqref{eq:fused partition function} of $\PartF_\alpha$. 
Note that for the totally unfused case $\multii = (1, \ldots, 1)$, we have $\CobloF_\alpha =  \Delta_\alpha$. 
The main result of this appendix now reads:

\begin{theorem} \label{thm:CFT properties coblo} 
The functions $\{ \PartF_\alpha \; | \; \alpha \in \LP_\multii \}$ and
$\{ \CobloF_\beta \; | \; \beta \in \LP_\multii \}$ are related via 
an invertible linear system of equations:
\begin{align*}
\PartF_\alpha = \sum_{ \substack{\beta \in \LP_\multii \\ \beta \succeq \alpha}} \# \mathcal{C} (\alpha/\beta) \, \CobloF_\beta, \qquad \textnormal{for all } \alpha \in \LP_\multii .
\end{align*}
Moreover, the functions 
$\CobloF_\alpha$ 
satisfy \textnormal{(PDE)}, \textnormal{(COV)}, \textnormal{(LIN)} 
in Theorem~\ref{thm:CFT properties}, and the fusion property
\begin{align*} 
\CobloF_\alpha %(p_1, \ldots, p_d) 
(\boldsymbol{p})
= \; &  \bigg( \prod_{j=1}^d ( 0! \cdot 1! \cdot \cdots \cdot (s_j-1)!) \bigg) \\
& \times \bigg(
 \lim_{x_{\summ_d} \to p_d} |x_{\summ_d} - p_d|^{-s_d+1}
 \ldots
 \lim_{x_{\summ_{d-1} + 2} \to p_d} |x_{\summ_{d-1} + 2} - p_d|^{-1}
 \lim_{x_{\summ_{d-1} + 1} \to p_d}
 \bigg) \\
& \; \ldots \;
\bigg( \lim_{x_{\summ_1} \to p_1} 
|x_{\summ_1} - p_1|^{-s_1+1} \ldots \lim_{x_2 \to p_1} |x_2 - p_1|^{-1} \lim_{x_1 \to p_1} \bigg)
\sum_{ \substack{\beta \in \LP_N \\ \beta \geq_\multii \alpha}}
\Delta^{\mathfrak{K}'}_\beta %( x_1, \ldots, x_{2N})
(\boldsymbol{x}) ,
\end{align*} 
with $\boldsymbol{p} = (p_1, \ldots, p_d)$, $\boldsymbol{x} = ( x_1, \ldots, x_{2N})$, 
and $\mathfrak{K}'(a, b) = \frac{1}{\pi(x_a-x_b)^2}$ the \quote{totally unfused kernel.}
\end{theorem}

The rest of the appendix constitutes the proof of Theorem~\ref{thm:CFT properties coblo}, which is entirely combinatorial.
We use the notation from Section~\ref{subsec:inv Fomin}. 
For any symmetric kernel $\mathfrak{K}$, set 
\begin{align} \label{eq: define Ufrak}
%\mathfrak{U}^{\mathfrak{K}}_\alpha 
\textnormal{\gls{symb:CobloFdisc}}
:= \sum_{ \substack{\beta \in \LP_N \\ \beta \geq_\multii \alpha}} \Delta^{\mathfrak{K}}_\beta , \qquad \alpha \in \LP_\multii .
\end{align}
Let us also recall that the partial order $\DPgeq$ on $\LP_N$ is defined in~\cite[Definition~2.1]{KKP:Boundary_correlations_in_planar_LERW_and_UST}.

\begin{lemma} \label{lem: app sums}
For any $\gamma \in \LP_\multii$, we have
\begin{align*}
\mathfrak{Z}^{\mathfrak{K}}_\gamma = \sum_{\substack{\alpha \in \LP_\multii \\ \alpha \DPgeq \gamma}} \# \mathcal{C} (\gamma/\alpha) \, \mathfrak{U}^{\mathfrak{K}}_\alpha.
\end{align*}
Furthermore, this system of linear relations is invertible and upper-triangular in the partial order $\DPgeq$, i.e., 
there exists a matrix $(M_{\alpha, \gamma})_{\alpha, \gamma \in \LP_\multii}$ such that 
\begin{align*}
\mathfrak{U}^{\mathfrak{K}}_\alpha = \sum_{\substack{\gamma \in \LP_\multii \\ \gamma \DPgeq \alpha}} M_{\alpha, \gamma} \, \mathfrak{Z}^{\mathfrak{K}}_\gamma , \qquad \alpha \in \LP_\multii.
\end{align*}
\end{lemma}

\begin{proof}
We begin from the definition~\eqref{eq: inverse Fomin-type sum} of the inverse Fomin type sum (see Section~\ref{subsec:inv Fomin}): 
\begin{align} \label{eq: app sum}
\mathfrak{Z}^{\mathfrak{K}}_\gamma = \sum_{\substack{\beta \in \LP_N \\ \beta \DPgeq \gamma}} \# \mathcal{C} (\beta/\gamma) \, \Delta^{\mathfrak{K}}_\beta.
\end{align}
Note that, given $\beta$ and $\multii$, 
there is a at most one $\alpha \in \LP_\multii$ such that $\beta \geq_\multii \alpha$.
Indeed, fix the values $\alpha(\summ_j)=\beta(\summ_j)$, and between the points $(\summ_j, \beta(\summ_j))$, 
fill the curve $\alpha$ with the lowest possible integer walk, i.e., 
a \quote{slope} or a \quote{V-shape}. 
This gives $\alpha$.
Note that such an $\alpha$ does not exist if the integer walk is a V-shape that does not remain non-negative. 
However, as we have here restricted ourselves to $\beta \DPgeq \gamma$, 
we may note that $\gamma$ also interpolates with the lowest possible integer walk between the points $(\summ_j, \gamma(\summ_j))$, 
and we have $\gamma(\summ_j) \leq \beta (\summ_j) = \alpha (\summ_j) $. 
Hence, such an $\alpha$ exists --- and furthermore $\alpha \DPgeq \gamma$. 
Conversely, given that $\alpha \DPgeq \gamma$ and $\beta \geq_\multii \alpha$, we have $\beta \DPgeq \gamma$. 
In conclusion, we can split the sum in~\eqref{eq: app sum} as
\begin{align*}
\mathfrak{Z}^{\mathfrak{K}}_\gamma = \sum_{\substack{\alpha \in \LP_\multii \\ \alpha \DPgeq \gamma}} 
\sum_{\substack{\beta \in \LP_N \\ \beta \geq_\multii \alpha}} \#  \mathcal{C} (\gamma/\beta) \, \Delta^{\mathfrak{K}}_\beta.
\end{align*}

Now, to obtain the claimed formula, note that given $\multii$ and $\alpha \in \LP_\multii$, 
the collection of all $\beta \geq_\multii \alpha$ can be obtained iteratively by \quote{wedge-lifting} local minima that do not occur at time instants $\summ_j$ (defined just above Lemma~2.11 in \cite{KKP:Boundary_correlations_in_planar_LERW_and_UST}). 
In particular, by~\cite[Lemma~2.15]{KKP:Boundary_correlations_in_planar_LERW_and_UST}, 
we see that $\# \mathcal{C}(\gamma/\beta)$ is constant over $\beta \geq_\multii \alpha$. 
This implies that 
\begin{align*}
\# \mathcal{C} (\gamma/\beta) = \# \mathcal{C} (\gamma/\alpha), 
\end{align*}
and the first claim follows. 
For the second claim, we just observe that the first linear relation is upper-triangular, 
and that the coefficients $\# \mathcal{C} (\gamma/\alpha)$, for $\gamma \DPgeq \alpha$, are all positive.
\end{proof}

\begin{remark}
The matrix $M$ can be actually found relatively explicitly.
Namely, from~\cite{KKP:Boundary_correlations_in_planar_LERW_and_UST}, using the combinatorial notations $\KWleq$ and $T_0$ of~\cite[Definition~2.6]{KKP:Boundary_correlations_in_planar_LERW_and_UST}, we obtain
\begin{align*}
\Delta^{\mathfrak{K}}_\beta = \sum_{\substack{\gamma \in \LP_\multii \\ \beta \KWleq \gamma}} (-1)^{|T_0 (\beta / \gamma)|} \, \mathfrak{Z}^{\mathfrak{K}}_\gamma , \qquad \beta  \in \LP_\multii ,
\end{align*}
and thus, from the definition~\eqref{eq: define Ufrak}, we obtain the formula 
\begin{align*}
\mathfrak{U}^{\mathfrak{K}}_\alpha 
= \; & \sum_{ \substack{\beta \in \LP_N \\ \beta \geq_\multii \alpha}} \sum_{\substack{\gamma \in \LP_N \\ \beta \KWleq \gamma}} (-1)^{|T_0 (\beta / \gamma)|} \mathfrak{Z}^{\mathfrak{K}}_\gamma \\
= \; & \sum_{ \substack{\beta \in \LP_N \\ \beta \geq_\multii \alpha}} \sum_{\substack{\gamma \in \LP_\multii \\ \beta \KWleq \gamma}} (-1)^{|T_0 (\beta / \gamma)|} \mathfrak{Z}^{\mathfrak{K}}_\gamma, \qquad \textnormal{for all } \alpha \in \LP_\multii,
\end{align*}
using also the fact that 
$\{ \mathfrak{Z}^{\mathfrak{K}}_\gamma \; | \; \gamma \in \LP_N \}$
are linearly independent and $\mathfrak{U}^{\mathfrak{K}}_\alpha$ 
lies in the subspace spanned by $\{ \mathfrak{Z}^{\mathfrak{K}}_\gamma \; | \; \gamma \in \LP_\multii \}$, 
so all of the other terms must cancel out in the double sum.
\end{remark}

\begin{proof}[Proof of Theorem~\ref{thm:CFT properties coblo}]
We apply Lemma~\ref{lem: app sums} with fixed valences $\multii$, points $(x_1, \ldots, x_d) \in \chamber_d$, 
and the kernel $\mathfrak{K}$ defined in Theorem~\ref{thm:scaling limit of pinched pertition functions},
to obtain 
\begin{align*}
\mathfrak{Z}^{\mathfrak{K}}_\alpha(x_1, \ldots, x_d)  = \PartF_\alpha (x_1, \ldots, x_d) 
\qquad \textnormal{and} \qquad
\mathfrak{U}^{\mathfrak{K}}_\alpha(x_1, \ldots, x_d)  = \CobloF_\alpha(x_1, \ldots, x_d), \qquad \alpha \in \LP_\multii .
\end{align*}
The claimed linear relation is then directly given in Lemma~\ref{lem: app sums}. 
The properties \textnormal{(PDE)}, \textnormal{(COV)}, \textnormal{(LIN)} in Theorem~\ref{thm:CFT properties coblo} 
follow directly from those of $\PartF_\alpha$ (Theorem~\ref{thm:CFT properties}) by inverting the linear relation. 
The fusion property also follows by the linear relation
\begin{align*}
\mathfrak{U}^{\mathfrak{K'}}_\alpha = \sum_{\substack{\beta \in \LP_\multii \\ \beta \DPgeq \alpha}} M_{\alpha, \beta} \, \mathfrak{Z}^{\mathfrak{K'}}_\beta  = \sum_{ \substack{\beta \in \LP_N \\ \beta \geq_\multii \alpha}} \Delta^{\mathfrak{K'}}_\beta,
\end{align*}
with $\mathfrak{K'}$ in Theorem~\ref{thm:CFT properties coblo}.
Taking the renormalized limit of the middle expression yields 
\begin{align*}
\sum_{ \substack{\beta \in \LP_\multii \\ \beta \DPgeq \alpha}} M_{\alpha, \beta} \, \PartF_\beta = \CobloF_\alpha ,
\end{align*}
by virtue of Theorem~\ref{thm:CFT properties}. 
This concludes the proof.
\end{proof}

%% file: tex/app-ex.tex
\subsection{A generic example}

Suppose that $\alpha \in \LP_\multii$ is the following valenced link pattern of $N=5$ links between $d=5$ endpoints with valences $\multii = (1, 1, 4, 2, 2)$:
\begin{align*}
\alpha \quad = \quad \{ \{ 1, 3\}, \{ 2, 3\}, \{ 3, 4\}, \{ 3, 5\}, \{ 4, 5\} \} 
\quad = \quad
\raisebox{-0.04\textwidth}
{\includegraphics[width=0.20\textwidth]{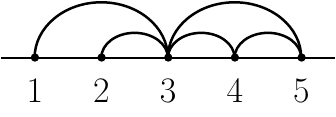}
} .
\end{align*}
We aim to find the expressions for the pure partition function $\PartF_\alpha ( x_1, \ldots, x_5)$ in~\eqref{eq:fused partition function} 
and the explicit basis function $\CobloF_\alpha ( x_1, \ldots, x_5)$ in~\eqref{eq:fused conformal block function}. 
Since both are linear combinations of the determinants $\{\Delta_\beta^\mathfrak{K} \;|\; \beta \in \LP_N \}$, 
we only exemplify one determinant. 
For instance, for the link pattern $\beta$ illustrated below, 
both\footnote{Here, we use the notation~\eqref{eq: partial order} from Appendix~\ref{app:coblo_det} and 
the notation $\imath (\alpha)$ from Figure~\ref{fig: UST fused}.} 
$\beta \geq_\multii \imath (\alpha) $ and $\beta \succeq \imath (\alpha)$ hold, 
so the determinant $\Delta_\beta^\mathfrak{K}$ appears in both $\PartF_\alpha$ and $\CobloF_\alpha$. 
To find $\Delta_\beta^\mathfrak{K}$ explicitly, 
we next tabulate the left-to-right orientation of $\beta$, 
as well as the corresponding boundary points and derivative operators determined by $\multii$. 
It is important to note that once the valences $\multii$ are fixed, 
the choice of the valenced link pattern $\alpha \in \LP_\multii$ 
does not enter the definition of $\Delta_\beta^\mathfrak{K}$. 

\bigskip
\begin{center}
\begin{tabular}{l L{0.7cm} L{0.5cm} L{0.5cm} L{0.5cm} L{0.5cm} L{0.5cm} L{0.5cm} L{0.5cm} L{0.5cm} L{0.5cm} L{0.5cm} L{0.5cm}}
\multicolumn{1}{r}{\raisebox{0.5cm}{$\beta = \; \; $}}& \multicolumn{12}{c}{{\includegraphics[width=9.9cm]{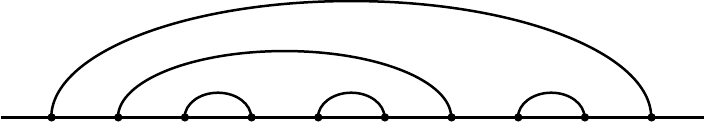}}}\\
\multicolumn{1}{l|}{index of boundary point} &&
 1 & 2 & 3 & 4 & 5 & 6 & 7 & 8 & 9 & 10 \\
\multicolumn{1}{l|}{left-to-right label ($\beta$)} & & 
$ a_1$ & $a_2$ & $a_3$ & $b_3$ & $a_4$ & $b_4$ & $b_2$ & $a_5$ & $b_5$ & $b_1$ & \\
\multicolumn{1}{l|}{corresp. boundary point ($\multii$)} & &
 $x_1$ & $x_2$ & $x_3$ & $x_3$ & $x_3$ & $x_3$ & $x_4$ & $x_4$ & $x_5$ & $x_5$ &  \\
\multicolumn{1}{l|}{corresp. derivative ($\multii$)} & &
 - & - & - & $\partial_3$ & $\partial_3^2$ & $\partial_3^3$ & - & $\partial_4$ & - & $\partial_5$ & 
\end{tabular}
\end{center}
\bigskip

Next, recall that $\Delta_\beta^\mathfrak{K} = \det (\mathfrak{K}(a_i, b_j))_{i,j=1}^N$, 
where the kernel $\mathfrak{K}(a_i, b_j)$ is defined in terms of the derivatives of Brownian excursion kernel $\ExcKH (y, z) = \tfrac{1}{\pi}\tfrac{1}{(z-y)^2}$ at the marked boundary points $x_i$, with $1 \leq i \leq d$. 
In this particular case, the matrix reads: 
\begin{center}
\begin{tabular}{l | L{0.3cm} C{0.5cm} C{0.5cm} C{0.5cm} C{0.5cm} C{0.5cm} }
&& $b_1$ & $b_2$ & $b_3$ & $b_4$ & $b_5$  \\
\hline
$a_1$ &&
$\partial_5 \ExcKH (x_1, x_5)$ & $\ExcKH (x_1, x_4)$ & $\partial_3 \ExcKH (x_1, x_3)$ & $\partial_3^3 \ExcKH (x_1, x_3)$ & $\ExcKH (x_1, x_5)$
 \\
$a_2$ & & 
$\partial_5 \ExcKH (x_2, x_5)$ & $\ExcKH (x_2, x_4)$ & $\partial_3 \ExcKH (x_2, x_3)$ & $\partial_3^3 \ExcKH (x_2, x_3)$ & $\ExcKH (x_2, x_5)$
 \\
$a_3$ & &
$\partial_5 \ExcKH (x_3, x_5)$ & $\ExcKH (x_3, x_4)$ & $0$ & $0$ & $\ExcKH (x_3, x_5)$
 \\
$a_4$ & &
$ \partial_3^2 \partial_5 \ExcKH (x_3, x_5)$ & $\partial_3^2 \ExcKH(x_3, x_4)$ & $0$ & $0$ & $\partial_3^2 \ExcKH (x_3, x_5)$ 
\\
$a_5$ &&
$\partial_4 \partial_5 \ExcKH (x_4, x_5)$ & $0$ & $\partial_4 \partial_3 \ExcKH (x_4, x_3)$ & $\partial_4 \partial_3^3 \ExcKH (x_4, x_3)$ & $\partial_4 \ExcKH (x_4, x_5)$
\end{tabular}
\end{center}
and plugging this into the concrete expressions, we obtain
\begin{align*}
\Delta_\beta^\mathfrak{K}
= \frac{1}{\pi^5}
\det
\left(
\begin{array}{c c c c c}
-\frac{2 }{(x_5 - x_1)^3} & \frac{ 1}{(x_4 - x_1)^2} & -\frac{2 }{(x_3 - x_1)^3} & -\frac{ 4! }{(x_3 - x_1)^5} & \frac{ 1}{(x_5 - x_1)^2}
 \\
-\frac{2 }{(x_5 - x_2)^3} & \frac{1}{(x_4 - x_2)^2} & -\frac{2}{(x_3 - x_2)^3} & -\frac{ 4! }{(x_3 - x_2)^5} & \frac{ 1}{(x_5 - x_2)^2}
 \\
-\frac{2 }{(x_5 - x_3)^3} & \frac{ 1}{(x_4 - x_3)^2} & 0 & 0 & \frac{ 1}{(x_5 - x_3)^2}
 \\
-\frac{4!}{(x_5 - x_3)^5} & \frac{3!}{(x_4 - x_3)^4} & 0 & 0 & \frac{3!}{(x_5 - x_3)^4} 
\\
- \frac{3! }{(x_5 - x_4)^4} & 0 & - \frac{3! }{(x_3 - x_4)^4} & -\frac{5! }{(x_3 - x_4)^6} & \frac{2 }{(x_5 - x_4)^3} 
\end{array}
\right).
\end{align*}

\subsection{The rainbow pattern}
\label{app:rainbow}
Our main results do not address (and do not hold true) for generic determinants 
$\{ \Delta^\mathfrak{K}_\beta \; | \; \beta \in \LP_N \}$ of the above type, 
but only for their suitable \emph{linear combinations} such as 
$\{ \PartF_\alpha \; | \; \alpha \in \LP_\multii \}$ or 
$\{ \CobloF_\alpha \; | \; \alpha \in \LP_\multii \}$. 
A very simple special case ---  already observed in the discrete model by Fomin~\cite{Fomin:LERW_and_total_positivity} 
--- is the $N$-link \quote{rainbow pattern} 
$\rainbow{N} := \{ \{ 1, 2N\}, \{ 2, 2N -1 \}, \ldots, \{ N, N+1\} \}$ (see Figure~\ref{fig:watermelon UST}).
Suppose that $\imath(\alpha) = \rainbow{N}$ and $\alpha$ is obtained from $\rainbow{N}$ 
via any (compatible) fusion with valences $\multii$. 
Then, the inverse Fomin type sum becomes just a single determinant, and 
\begin{align*}
\PartF_\alpha (x_1, \ldots, x_d)
= \CobloF_\alpha (x_1, \ldots, x_d)
= \Delta^\mathfrak{K}_{\alpha}.
\end{align*}
This case provides a simple sanity check for our results.

\bigskip

To consider a particularly symmetric setup, 
let us suppose that $\multii = (1, 1, \ldots, 1, N)$ and denote $x_{N+1} =:y$. 
Let $\alpha$ be the unique $\multii$-valencend link pattern obtained from $\rainbow{N}$ by fusing 
the last $N$ indices together. We tabulate: 
\begin{center}
\begin{tabular}{l L{0.3cm} L{0.5cm} L{0.5cm} L{0.5cm} L{0.5cm} L{0.5cm} L{0.5cm} L{0.5cm} L{0.5cm} }
\multicolumn{1}{l|}{index of boundary point ($\rainbow{N}$)} &&
$1$ & $2$ & \ldots & $N$ & $N+1$ & $N+2$ & \ldots & $2N$ \\
\multicolumn{1}{l|}{left-to-right label ($\rainbow{N}$)} & & 
$ a_1$ & $a_2$ & \ldots & $a_N$ & $b_N$ & $b_{N-1}$ & \ldots & $b_1$ \\
\multicolumn{1}{l|}{corresp. boundary point ($\multii$)} & &
 $x_1$ & $x_2$ & \ldots & $x_N$ & $y$ & $y$ & \ldots & $y$  \\
\multicolumn{1}{l|}{corresp. derivative ($\multii$)} & &
 - & - & \ldots & - & - & $\partial_y$ & \ldots &  $\partial_y^{N-1}$ 
\end{tabular}
\end{center}
From this, we obtain
\begin{align*}
\mathfrak{K}(a_i, b_j) \; = \; \partial_y^{N-j} \frac{\pi^{-1}}{(y-x_i)^2}
\; = \; \frac{(-1)^{N-j} (N-j + 1)! \pi^{-1} }{(y - x_i)^{N-j + 2}} ,
\end{align*}
and consequently, 
\begin{align*}
\PartF_\alpha (x_1, \ldots, x_N, y)
= & \; \Delta^\mathfrak{K}_\alpha = \det (\mathfrak{K}(a_i, b_j))_{i,j=1}^N \\
= & \; (-1)^{\lfloor  N / 2 \rfloor} \det ( \mathfrak{K}(a_i, b_{N-j+ 1}) )_{i,j=1}^N \\
= & \; (-1)^{\lfloor  N / 2 \rfloor} \det \bigg( \frac{(-1)^{j-1} j! }{ \pi \, (y - x_i)^{j+1}} \bigg)_{i,j=1}^N \\
= & \; \frac{1}{\pi^{N}} \prod_{k=1}^N \frac{k!}{(y-x_k)^2} \times \det \Big( \frac{ 1 }{(y - x_i)^{j-1}} \Big)_{i,j=1}^N.
\end{align*}
The determinant can be simplified by Vandermonde's determinant formula:
\begin{align} \label{eq:watermelon partition function}
\PartF_\alpha (x_1, \ldots, x_N, y) 
= \frac{1}{\pi^{N}}  \prod_{k=1}^N \frac{k!}{(y-x_k)^{N+1}} \times  \prod_{1 \leq i < j \leq N} (x_j - x_i) .
\end{align}
As a sanity check, the second-order BPZ PDEs~\eqref{eq:BPZ} can now be verified directly by hand.
This expression also plays a crucial role in the application of Appendix~\ref{subsec:watermelon SLE}.

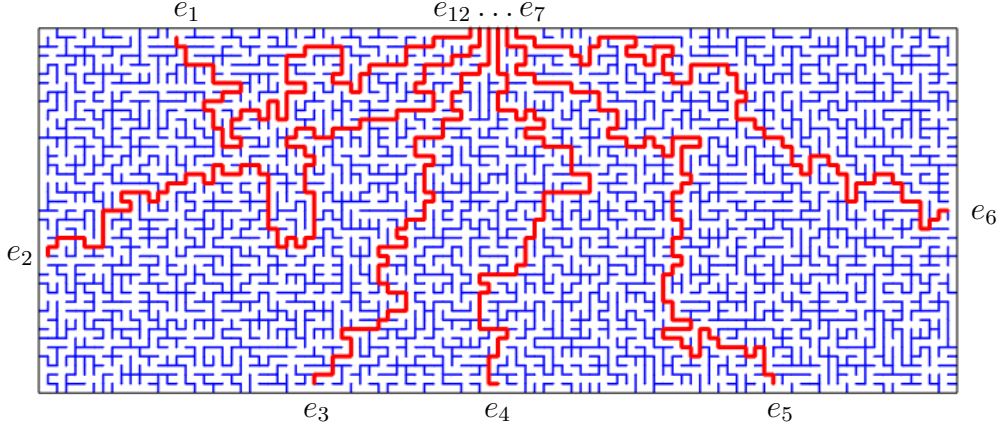
\begin{figure}
\begin{overpic}[width=0.8\textwidth]{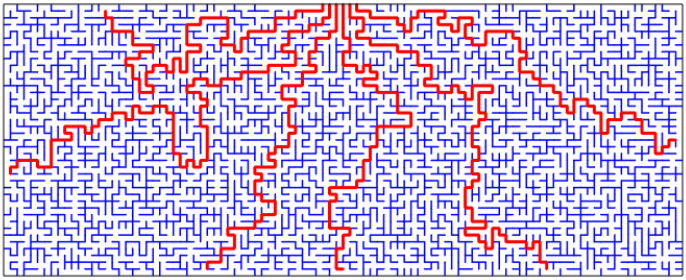}
 \put (29,-2) {\large $e_3$}
 \put (48.5,-2) {\large $e_4$}
 \put (79,-2) {\large $e_5$}
 \put (101,19.5) {\large $e_6$}
 \put (-3,15) {\large $e_2$}
 \put (15,41.5) {\large $e_1$}
 \put (43,41.5) {\large $e_{12}\ldots e_7$}
\end{overpic}
\caption{
\label{fig:watermelon UST}
The discrete geometry pointed out by Fomin~\cite{Fomin:LERW_and_total_positivity} and studied (in the limiting unfused case) 
early on in the SLE context~\cite{Kozdron-Lawler:Configurational_measure_on_mutually_avoiding_SLEs}, here transferred to spanning trees. 
The boundary branches from the interior vertices of $e_1, \ldots, e_N$ (here $N=6$) reach the boundary via the edges $e_{N+1}, \ldots, e_{2N}$, each using a different edge. 
The only topologically possible link pattern is then the rainbow pattern ${\rainbow{N}}$, 
and the resulting boundary-to-boundary branches have the same law as the odd-to-even index branches studied in this article, 
forming the rainbow pattern~\cite[Lemma~3.1]{KKP:Boundary_correlations_in_planar_LERW_and_UST}. 
The discussion in Appendices~\ref{app:rainbow} and~\ref{subsec:watermelon SLE} correspond to this setup, 
when the edges $e_1, \ldots, e_N$ are macroscopically apart, 
while the edges $e_{N+1}, \ldots, e_{2N}$ are only one lattice step apart, as illustrated in this figure.
}
\end{figure}

%% file: tex/app-fused.tex
We now briefly sketch two applications to the  theory of Schramm-Loewner evolutions (SLE). 
This appendix is aimed for readers with probabilistic background and interests, and we assume at least superficial familiarity with SLE. 
Due to the advanced state of the literature, we do not give detailed proofs (though such could be produced with technical work).
We first discuss multiple SLE curves growing from the same point in Section~\ref{subsec:Fusing endpoints}.
The second application concerns scaling limits of the UST branches studied in this article.
Specifically, in Section~\ref{subsec:watermelon SLE} 
we explain how one can characterize the scaling limits of \quote{fused rainbow,} or \quote{half-watermelon,} 
type random curves illustrated in Figure~\ref{fig:watermelon UST} above. 

\subsection{Fusing endpoints of multiple SLEs}
\label{subsec:Fusing endpoints}

Let us explain how the \quote{valenced} pure partition functions 
describe the limit of local multiple SLEs as some of the marked points (other than the growth point) tend together
(see Figure~\ref{fig:NSLE fusion}). 
Recall first that, for local multiple SLEs of $N$ curves in $(\bH; x_1, \ldots, x_{2N})$, 
with the pairing $\alpha \in \LP_N$ and parameter $\kappa \in (0,8)$, 
the curve initial segment growing from $x_j$ is described by the \emph{Loewner equation}
\begin{align} \label{eq:Loewner ODE}
\begin{cases}
g_0 (z) = z , & z \in \overline{\bH} ,
\\
\partial_t g_t (z) = \dfrac{2}{g_t(z) - X\super{j}_t} , & z \in \overline{\bH} , \, t \geq 0 ,
\end{cases}
\end{align}
where the driving function is determined by $X\super{j}_0 = x_j$ and
\begin{align}
\label{eq:multiple SLE}
\diff X\super{j}_t = \sqrt{\kappa} \, \diff B_t + \kappa \,  \frac{\partial_{j} \PartF_\alpha (X\super{1}_t, \ldots, X\super{N}_t)}{\PartF_\alpha (X\super{1}_t, \ldots, X\super{2N}_t)} \diff t, 
\qquad X\super{i}_t = g_t (x_i) , \textnormal{ for } i \neq j.
\end{align}
Conventionally, we consider this process stopped at the exit time of a \quote{localization} neighborhood of $x_j$.
The results of this appendix are specific for $\kappa = 2$, but we will give the equations in the form that they (in some cases conjecturally) generalize to other $\kappa$.

\begin{figure}
\includegraphics[width=0.5\textwidth]{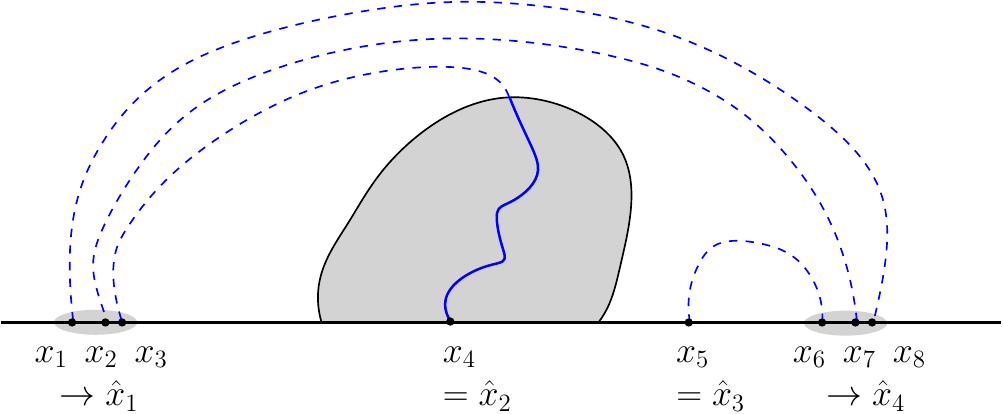}
\caption{
\label{fig:NSLE fusion}
Fusing endpoints of multiple SLEs: the initial segment from an unfused boundary point $x_4$ 
up to the exit time of its localization neighborhood (shaded), 
and an underlying link pattern compatible with the fusion.
}
\end{figure}

Now, fix valences $\multii \in \bZpos^d$ and a $\multii$-valenced link pattern $\hat{\alpha} \in \LP_\multii$, 
with $\imath(\hat{\alpha}) = \alpha$, and an index $1 \leq \hat{j} \leq d$ with $s_{\hat{j}} = 1$ so that the $\hat{j}$:th boundary point of $\hat{\alpha}$ corresponds after the unfusing map to the $j$:th point of $\alpha$ 
(i.e., $\summ_{\hat{j}} = j$).
Let us consider the limit of the multiple SLE process~(\ref{eq:Loewner ODE},~\ref{eq:multiple SLE}) 
as $x_1, \ldots, x_{2N}$ are fused to $\hat{x}_1, \ldots, \hat{x}_d$ according to the valences $\multii$ 
(Figure~\ref{fig:NSLE fusion}). 
Applying Proposition~\ref{prop:fused SLE} to both\footnote{With $\kappa=2$, the derivative $\partial_{j} \PartF_\alpha$ too 
is an inverse Fomin type sum, obtained from the pure partition function $\PartF_\alpha$ 
by differentiating the kernel $\mathfrak{K}(j, \cdot)$ with $\partial_j$. 
$\partial_{\hat{j}} \PartF_{\hat{\alpha}} $ can be described analogously in terms of $\PartF_{\hat{\alpha}}$. Alternatively, the uniform convergence $\PartF_{\alpha} \to \PartF_{\hat{\alpha}}$ given by Proposition~\ref{prop:fused SLE} can be extended to derivatives by the local complex analyticity of the functions.} 
$\PartF_\alpha$ and $\partial_{j} \PartF_\alpha$, 
we see that as the maximal distance $\smash{\underset{1\leq k \leq d}{\max}} \, \{ x_{\summ_k} - x_{\summ_{k-1}+1} \}$ 
of the boundary points to be fused tends to zero,
\begin{align*}
\frac{\partial_{j} \PartF_\alpha (x_1, \ldots, x_{2N})}{\PartF_\alpha (x_1, \ldots, x_{2N})} 
\quad \longrightarrow \quad 
\frac{\partial_{\hat{j}} \PartF_{\hat{\alpha}} (\hat{x}_1, \ldots, \hat{x}_d)}{\PartF_{\hat{\alpha}} (\hat{x}_1, \ldots, \hat{x}_d)},
\end{align*}
uniformly over compact sets $(\hat{x}_1, \ldots, \hat{x}_d) \in K \subset \chamber_d$. 
(Note that the strict \emph{positivity} of $\PartF_{\hat{\alpha}}$, i.e., 
property \textnormal{(POS)} in Theorem~\ref{thm:CFT properties}, is crucial here.) 
Consequently, the stopped driving functions converge in the same limit (in the uniform-over-compacts topology) 
almost surely to the stopped solution to the SDE
\begin{align}
\label{eq:fused multiple SLE}
\diff \hat{X}\super{\hat{j}}_t = \sqrt{\kappa} \, \diff B_t + \kappa \, \frac{\partial_{\hat{j}} \PartF_{\hat{\alpha}} (\hat{X}\super{1}_t, \ldots, \hat{X}\super{d}_t)}{\PartF_{\hat{\alpha}} (\hat{X}\super{1}_t, \ldots, \hat{X}\super{d}_t)} \diff t , 
\quad \hat{X}\super{i}_t = g_t (\hat{x}_i) , \textnormal{ for } i \neq \hat{j}.
\end{align}
In conclusion, the fused partition functions of the present work serve as partition functions of \quote{fused multiple SLE curves.}  
The convergence of driving functions could be fairly straightforwardly promoted to convergence of curves~\cite{Kemppainen-Smirnov:Random_curves_scaling_limits_and_Loewner_evolutions, Karrila:Multiple_SLE_local_to_global},
since the related random curves are simple (for $\kappa = 2 \leq 4$).
See also the recent~\cite{HPW:Multiradial_SLE_with_spiral} for a special case. 

Let us note that the above relied on Theorem~\ref{thm:CFT properties} \textnormal{(FUS)} (simultaneous generalization) and \textnormal{(POS)}, and tacitly also on \textnormal{(PDE)} and \textnormal{(COV)} via the definition of multiple SLEs. 
The \emph{linear independence} of $\{\PartF_{\hat{\alpha}} \;|\; \hat{\alpha} \in \LP_\multii \}$,
i.e., property \textnormal{(LIN)} in Theorem~\ref{thm:CFT properties}, 
also has a probabilistic interpretation: 
it implies that the laws of the curves/driving functions~\eqref{eq:fused multiple SLE} for different $\hat{\alpha} \in \LP_\multii$ are \emph{convex 
independent}. 
The proof of this fact would be similar to that of~\cite[Lemma~1]{Karrila:Computation_of_pairing_probabilities_in_multiple-curve_models}, replacing the maximum principle of elliptic PDEs 
with the (much simpler) fact that rational functions are either everywhere constant or nowhere locally constant.

\subsection{Scaling limits of half-watermelon branches}
\label{subsec:watermelon SLE}

We now briefly explain how one can characterize the scaling limit of the \quote{fused rainbow,} 
or \quote{half-watermelon,} type random curves, which were analyzed in Appendix~\ref{app:explicit determinants} and illustrated in Figure~\ref{fig:watermelon UST}.

Consider first the \textit{marginal law} of one curve. The analogous unfused case has been addressed in~\cite{Karrila:UST_branches_martingales_and_multiple_SLE2}, and the fused case can be treated identically, 
by investigating the convergence of observables~\cite[Proposition~5.2]{Karrila:UST_branches_martingales_and_multiple_SLE2} 
via our
Theorem~\ref{thm:scaling limit of pinched pertition functions}.
The conclusion is that, upon mapping conformally the limiting domain $(\domain; p_1, \ldots, p_{N+1})$ to $(\bH; x_1, \ldots, x_N, x_{N+1})$ with $x_{N+1}=y$, 
the image of the limit curve from any $x_j$, $1 \leq j \leq N$, is described by the Loewner growth process~\eqref{eq:Loewner ODE} where the driving function is determined by $X\super{j}_0 = x_j$ and
\begin{align}
\label{eq:SLE watermelon finite target}
\diff X\super{j}_t = \sqrt{\kappa} \,  \diff B_t + \kappa \,  \frac{\partial_{x_j} \PartF_\alpha (X\super{1}_t, \ldots, X\super{N}_t, Y_t)}{\PartF_\alpha (X\super{1}_t, \ldots, X\super{N}_t, Y_t)} \diff t , 
\qquad X\super{i}_t = g_t (x_i) , \textnormal{ for } i \neq j ,
\end{align}
and $Y_t = g_t(y)$ where in the present case, 
we have $\kappa=2$ and $\PartF_\alpha$ is given in Equation~\eqref{eq:watermelon partition function}. 

This convergence result is also a \emph{conformal invariance} property: the formula above holds (and describes the same random curves) for any $(x_1, \ldots, x_N, y)$ that can be obtained as a conformal image of the marked points. 
In particular, choosing different conformal maps so that $y \to \infty$ while 
the $x_i$:s converge, the SLE driving functions converge\footnote{The convergence holds in the uniform-on-compacts topology.} 
almost surely to a similar process with the SDEs~\eqref{eq:SLE watermelon finite target} altered to
the ($\SLE(\kappa; 2,2,\ldots,2)$ type) SDEs
\begin{align}
\label{eq:watermelon SLE marginal}
\diff X\super{j}_t = \sqrt{\kappa} \, \diff B_t + \sum_{i \neq j} \frac{2}{X\super{j}_t - X\super{i}_t} \diff t, 
\qquad X\super{i}_t = g_t (x_i) , \textnormal{ for } i \neq j .
\end{align}
This provides a more convenient description of the marginal law.
The processes~\eqref{eq:watermelon SLE marginal} (for any $\kappa < 8$) have for long been interpreted as the marginal law of one curve out of multiple SLEs growing from $x_1, \ldots, x_N$ to infinity~(see, e.g.,~\cite[Section~4.6]{BBK:Multiple_SLEs_and_statistical_mechanics_martingales}).

Lastly, we address the \emph{joint limit} of the multiple curves. 
It can be characterized in terms of iterated growth processes of this type\footnote{Promoting a spatial Markovian property of a discrete model to its scaling limit, see~\cite{Karrila:Multiple_SLE_local_to_global}.}. 
Note first that by an \emph{a posteriori} analysis of the $\SLE(\kappa; 2,2,\ldots,2)$ type process~\eqref{eq:watermelon SLE marginal}, it describes a chordal simple curve from $x_j $ to $\infty$. 
Now, let us grow the curve from $x_1$ until it reaches a spherical distance $\epsilon$ from $\infty$. 
The conditional law of the curve starting from $x_2$ is then a similar growth process in the domain slit by the first initial segment.  
The limiting curves being simple, knowing the laws of such initial segments with any $\epsilon$ characterizes 
the law of the entire curve collection.  

Alternatively, such iterated growth processes of the type~\eqref{eq:watermelon SLE marginal} produce % (at least for $\kappa \leq 4$) 
the same curves as the \emph{simultaneous growth} description using independent Brownian motions $(B\super{1}, \ldots, B\super{N})$:
\begin{align}
\label{eq:SLE simult growth 1}
\partial_t g_t (z) = \; & \frac{1}{N} \sum_{i=1}^N \frac{2}{g_t (z) - X\super{i}_t }, %\qquad \textnormal{where}
\\ 
\label{eq:SLE simult growth 2}
\diff X\super{i}_t = \; & \sqrt{ \frac{\kappa}{N}} \, \diff B\super{i}_t + \frac{1}{N}  \sum_{\substack{k= 1 \\ k \neq i }}^N \frac{4 \, \diff t}{X\super{i}_t - X\super{k}_t}, \qquad i \in \{1,2,\ldots,N\} ,
\end{align}
providing another characterization of the limiting joint law.
A rigorous proof of the above statement appears not to be completely present in the literature, 
as also discussed in~\cite[Section~1.3]{Katori-Koshida:Three_phases_of_multiple_SLE_driven_by_non-colliding_Dyson_Brownian_motions}.
A possible proof would combine an infinitesimal growth argument similar to~\cite[Section~4]{Healey-Lawler:N_sided_radial_SLE} with the commutation property of~\cite{Dubedat:Commutation_relations_for_SLE} (also,~\cite[Appendix~A]{Kytola-Peltola:Pure_partition_functions_of_multiple_SLEs}).
The details in the case of the half-watermelon SLE process will appear in~\cite{HPW:Multiradial_SLE_with_spiral}, where it is applied to address the resampling property of multiradial $\SLE_\kappa$ with $\kappa \in (0,4]$. 

We conclude by noting that we are not aware of the processes~\eqref{eq:watermelon SLE marginal} 
or~(\ref{eq:SLE simult growth 1},~\ref{eq:SLE simult growth 2}) 
having previously been identified as a \emph{scaling limit} of any lattice model or any $\kappa$
--- in spite of them being fairly well-studied as continuous objects~(e.g.,~\cite{Monaco-Schleissinger:Multiple_SLE_and_the_complex_Burgers_equation, Hotta-Katori:Hydrodynamic_limit_of_multiple_SLE, Katori-Koshida:Three_phases_of_multiple_SLE_driven_by_non-colliding_Dyson_Brownian_motions, Hotta-Schleissinger:Limits_of_radial_multiple_SLE_and_Burgers-Loewner_differential_equation, Chen-Margarint}).
Outlining this was one of our motivations to include Appendices~\ref{app:explicit determinants}~\&~\ref{app:fused SLE} to this article.

%% file: conformal_blocks-c_equals_minus2-arXiv-v3.bbl
\begin{thebibliography}{}

\end{thebibliography}


\begin{thebibliography}{FLPW24+}

\bibitem[Ang25]{Ang:Liouville_CFT_and_quantum_zipper}
Morris Ang.
\newblock Liouville conformal field theory and the quantum zipper.
\newblock {\em Ann. Probab.}, to appear, 2025. %Preprint in arXiv:2301.13200.

\bibitem[BB03]{Bauer-Bernard:Conformal_field_theories_of_SLEs}
Michel Bauer and Denis Bernard.
\newblock Conformal field theories of stochastic {L}oewner evolutions.
\newblock {\em Comm. Math. Phys.}, 239(3):493--521, 2003.


\bibitem[BBK05]{BBK:Multiple_SLEs_and_statistical_mechanics_martingales}
Michel Bauer, Denis Bernard, and Kalle Kyt{\"o}l{\"a}.
\newblock Multiple {S}chramm-{L}oewner evolutions and statistical mechanics
  martingales.
\newblock {\em J. Stat. Phys.}, 120(5-6):1125--1163, 2005.


\bibitem[BJ24]{Baverez-Jego:The_CFT_of_SLE_loop_measures_and_the_Kontsevich-Suhov_conjecture}
Guillaume Baverez and Antoine Jego.
\newblock The {CFT} of {SLE} loop measures and the {K}ontsevich-{S}uhov conjecture.
\newblock Preprint in arXiv:2407.09080, 2024.

\bibitem[BPZ84a]{BPZ:Infinite_conformal_symmetry_in_2D_QFT}
Alexander~A. Belavin, Alexander~M. Polyakov, and Alexander~B. Zamolodchikov.
\newblock Infinite conformal symmetry in two-dimensional quantum field theory.
\newblock {\em Nucl. Phys. B}, 241(2):333--380, 1984.


\bibitem[BPZ84b]{BPZ:Infinite_conformal_symmetry_of_critical_fluctuations_in_2D}
Alexander~A. Belavin, Alexander~M. Polyakov, and Alexander~B. Zamolodchikov.
\newblock Infinite conformal symmetry of critical fluctuations in two
  dimensions.
\newblock {\em J. Stat. Phys.}, 34(5-6):763--774, 1984.


\bibitem[BS89]{Bauer-Saleur:On_some_relations_between_local_height_probabilities_and_conformal_invariance}
Michel Bauer and Hubert Saleur.
\newblock On some relations between local height probabilities and conformal
  invariance.
\newblock {\em Nucl. Phys. B}, 320(3):591--624, 1989.


\bibitem[BSA88]{BSA:Degenerate_CFTs_and_explicit_expressions_for_some_null_vectors}
Louis Beno\^it and Yvan Saint-Aubin.
\newblock Degenerate conformal field theories and explicit expressions for some
  null vectors.
\newblock {\em Phys. Lett. B}, 215(3):517--522, 1988.

\bibitem[BW23a]{Baverez-Wu:Irreducibility_of_Virasoro_representations_in_Liouville_CFT}
Guillaume Baverez and Baojun Wu.
\newblock Irreducibility of {V}irasoro representations in {L}iouville {CFT}.
\newblock Preprint in arXiv:2312.07344, 2023.

\bibitem[BW23b]{Baverez-Wu:Higher_equations_of_motion_at_level_2_in_LCFT}
Guillaume Baverez and Baojun Wu.
\newblock Higher equations of motion at level {$2$} in {L}iouville {CFT}.
\newblock Preprint in arXiv:2312.13900, 2023.

\bibitem[Car84]{Cardy:Conformal_invariance_and_surface_critical_behavior}
John~L. Cardy.
\newblock Conformal invariance and surface critical behavior.
\newblock {\em Nucl. Phys. B}, 240(4):514--532, 1984.

\bibitem[Car92]{Cardy:Critical_percolation_in_finite_geometries}
John~L. Cardy.
\newblock Critical percolation in finite geometries.
\newblock {\em J. Phys. A}, 25(4):L201--206, 1992.

\bibitem[CP23]{Crampe-Poulain-d-Andecy:Fused_braids_and_centralisers_of_tensor_representations_of_Uq_gln}
Nicolas Cramp\'e and Lo\"ic~Poulain d'Andecy.
\newblock Fused braids and centralisers of tensor representations of
  {$U_q(\mathfrak{gl}_n)$}.
\newblock {\em Algebr. Represent. Theory}, 26(3):901--955, 2023.

\bibitem[Cer24]{Cercle:Higher_equations_of_motion_for_boundary_LCFT}
Baptiste Cercl{\'e}.
\newblock Higher equations of motion for boundary {L}iouville {C}onformal {F}ield {T}heory from the {W}ard identities.
\newblock Preprint in arXiv:2401.13271, 2024.

\bibitem[CFL28]{CFL:Uber_PDE_der_mathphys}
R.~Courant, K.~Friedrichs, and H.~Lewy.
\newblock {\"U}ber die partiellen {D}ifferenzengleichungen der mathematischen
  {P}hysik.
\newblock {\em Math. Ann.}, 100(1):32--74, 1928.


\bibitem[CHI21]{CHI:Correlations_of_primary_fields_in_the_critical_planar_Ising_model}
Dmitry Chelkak, Cl{\'e}ment Hongler, and Konstantin Izyurov.
\newblock Correlations of primary fields in the critical planar {I}sing model.
\newblock Preprint in arXiv:2103.10263, 2021.

\bibitem[CS11]{Chelkak-Smirnov:Discrete_complex_analysis_on_isoradial_graphs}
Dmitry Chelkak and Stanislav Smirnov.
\newblock Discrete complex analysis on isoradial graphs.
\newblock {\em Adv. Math.}, 228(3):1590--1630, 2011.

\bibitem[CM22]{Chen-Margarint}
Jiaming Chen and Vlad Margarint.
\newblock Perturbations of multiple {S}chramm--{L}oewner evolution with two non-colliding {D}yson {B}rownian motions.
\newblock {\em Stoch. Proc. Applic.}, 151:553--570, 2022.

\bibitem[DCS12]{DCS:Conformal_invariance_of_lattice_models}
Hugo Duminil-Copin and Stanislav Smirnov.
\newblock Conformal invariance of lattice models.
\newblock In {\em Probability and {S}tatistical {P}hysics in two and more
  dimensions}, volume~15, pages 213--276. Clay Mathematics Proceedings,
  American Mathematical Society, Providence, RI, 2012.

\bibitem[DFMS97]{DMS:CFT}
Philippe Di~Francesco, Pierre Mathieu, and David S{\'e}n{\'e}chal.
\newblock {\em Conformal field theory}.
\newblock Graduate Texts in Contemporary Physics. Springer-Verlag, New York,
  1997.


\bibitem[DJ89]{Dipper-James:Q_Schur_algebra}
Richard Dipper and Gordon James.
\newblock The {$q$}-{S}chur algebra.
\newblock {\em Proc. London Math. Soc.}, 59(1):23--50, 1989.


\bibitem[dMS16]{Monaco-Schleissinger:Multiple_SLE_and_the_complex_Burgers_equation}
Andrea del Monaco and Sebastian Schlei{\ss}inger.
\newblock Multiple {SLE} and the complex {B}urgers equation.
\newblock {\em Math. Nachr.}, 289(16):2007--2018, 2016.


\bibitem[DS86]{Duplantier-Saleur:Exact_surface_and_wedge_exponents_for_polymers_in_two_dimensions}
Bertrand Duplantier and Hubert Saleur.
\newblock Exact surface and wedge exponents for polymers in two dimensions.
\newblock {\em Phys. Rev. Lett.}, 57(25):3179--3182, 1986.


\bibitem[Dub06]{Dubedat:Excursion_decomposition_for_SLE_and_Watts_crossing_formula}
Julien Dub{\'e}dat.
\newblock Excursion decompositions for $\mathrm{SLE}$ and {W}atts' crossing
  formula.
\newblock {\em Probab. Theory Related Fields}, 134(3):453--488, 2006.


\bibitem[Dub07]{Dubedat:Commutation_relations_for_SLE}
Julien Dub{\'e}dat.
\newblock Commutation relations for {$\mathrm{SLE}$}.
\newblock {\em Comm. Pure Appl. Math.}, 60(12):1792--1847, 2007.


\bibitem[Dub15a]{Dubedat:SLE_and_Virasoro_representations_localization}
Julien Dub{\'e}dat.
\newblock $\mathrm{SLE}$ and {V}irasoro representations: localization.
\newblock {\em Comm. Math. Phys.}, 336(2):695--760, 2015.


\bibitem[Dub15b]{Dubedat:SLE_and_Virasoro_representations_fusion}
Julien Dub{\'e}dat.
\newblock $\mathrm{SLE}$ and {V}irasoro representations: fusion.
\newblock {\em Comm. Math. Phys.}, 336(2):761--809, 2015.


\bibitem[Dup06]{Duplantier:Conformal_random_geometry}
Bertrand Duplantier.
\newblock Conformal random geometry.
\newblock In {\em Mathematical Statistical Physics: Lecture Notes of the Les
  Houches Summer School (2005)}. Elsevier, 2006.

\bibitem[FF82]{Feigin-Fuchs:Invariant_skew-symmetric_differential_operators_on_the_line_and_Verma_modules_over_Virasoro}
Boris~L. Fe{\u\i}gin and Dmitry~B. Fuchs.
\newblock Invariant skew-symmetric differential operators on the line and
  {V}erma modules over the {V}irasoro algebra.
\newblock {\em Funct. Anal. Appl.}, 16(2):114--126, 1982.


\bibitem[FF83]{Feigin-Fuchs:Verma_modules_over_Virasoro}
Boris~L. Fe{\u\i}gin and Dmitry~B. Fuchs.
\newblock Verma modules over the {V}irasoro algebra.
\newblock {\em Funktsional Anal. i Prilozhen}, 17(3):91--92, 1983.


\bibitem[FF84]{Feigin-Fuchs:Verma_modules_over_Virasoro_book}
Boris~L. Fe{\u\i}gin and Dmitry~B. Fuchs.
\newblock Verma modules over the {V}irasoro algebra.
\newblock In {\em Topology (Leningrad 1982)}, volume 1060 of {\em Lecture Notes
  in Mathematics}, pages 230--245, Berlin Heidelberg, 1984. Springer-Verlag.

\bibitem[FF90]{Feigin-Fuchs:Representations_of_Virasoro}
Boris~L. Fe{\u\i}gin and Dmitry~B. Fuchs.
\newblock Representations of the {V}irasoro algebra.
\newblock In {\em Representation of Lie Groups and Related Topics}, volume~7 of
  {\em Advanced Studies in Contemporary Mathematics}, pages 465--554. Gordon
  and Breach, New York, 1990.

\bibitem[FK04]{Friedrich-Kalkkinen:On_CFT_and_SLE}
Roland Friedrich and Jussi Kalkkinen.
\newblock On conformal field theory and stochastic {L}oewner evolution.
\newblock {\em Nucl. Phys. B}, 687(3):279--302, 2004.

\bibitem[FLPW24]{FLPW:Multiple_SLEs_Coulomb_gas_integrals_and_pure_partition_functions}
Yu~Feng, Mingchang Liu, Eveliina Peltola, and Hao Wu.
\newblock Multiple {SLE}s for {$\kappa\in (0,8)$:} {C}oulomb gas integrals and
  pure partition functions.
\newblock Preprint in arXiv:2406.06522, 2024.

\bibitem[Fom01]{Fomin:LERW_and_total_positivity}
Sergey Fomin.
\newblock Loop-erased walks and total positivity.
\newblock {\em Trans. Amer. Math. Soc.}, 353(9):3363--3583, 2001.


\bibitem[FP18a]{Flores-Peltola:Generators_projectors_and_the_JW_algebra}
Steven~M. Flores and Eveliina Peltola.
\newblock Generators, projectors, and the {J}ones-{W}enzl algebra.
\newblock Preprint in arXiv:1811.12364, 2018.

\bibitem[FP18b]{Flores-Peltola:Standard_modules_radicals_and_the_valenced_TL_algebra}
Steven~M. Flores and Eveliina Peltola.
\newblock Standard modules, radicals, and the valenced {T}emperley-{L}ieb
  algebra.
\newblock Preprint in arXiv:1801.10003, 2018.

\bibitem[FP20]{Flores-Peltola:Higher_spin_QSW}
Steven~M. Flores and Eveliina Peltola.
\newblock Higher spin quantum and classical {S}chur-{W}eyl duality for
  {$\mathfrak{sl}(2)$}.
\newblock Preprint in arXiv:2008.06038, 2020.

\bibitem[FPW24]{FPW:Connection_probabilities_of_multiple_FK_Ising_interfaces}
Yu~Feng, Eveliina Peltola, and Hao Wu.
\newblock Connection probabilities of multiple {FK}-{I}sing interfaces.
\newblock {\em Probab. Theory Related Fields}, 189(1-2):281--367, 2024.

\bibitem[Fri04]{Friedrich:On_connections_of_CFT_and_SLE}
Roland Friedrich.
\newblock On connections of conformal field theory and stochastic {L}oewner
  evolution.
\newblock Preprint in arXiv:math-ph/0410029, 2004.

\bibitem[FPW25]{FPW:Crossing_probabilities_of_critical_percolation_interfaces}
Yu~Feng, Eveliina Peltola, and Hao Wu.
\newblock Crossing probabilities of multiple percolation interfaces:
  generalizations of {C}ardy's formula and {W}att's formula.
\newblock In preparation, 2025.

\bibitem[FSKZ17]{FSKZ:A_formula_for_crossing_probabilities_of_critical_systems_inside_polygons}
Steven~M. Flores, Jacob J.~H. Simmons, Peter Kleban, and Robert~M. Ziff.
\newblock A formula for crossing probabilities of critical systems inside
  polygons.
\newblock {\em J. Phys. A}, 50(6):064005, 2017.


\bibitem[FW03]{Friedrich-Werner:Conformal_restriction_highest_weight_representations_and_SLE}
Roland Friedrich and Wendelin Werner.
\newblock Conformal restriction, highest weight representations and
  {$\mathrm{SLE}$}.
\newblock {\em Comm. Math. Phys.}, 243(1):105--122, 2003.


\bibitem[GKR23]{GKR:Compactified_imaginary_Liouville}
Colin Guillarmou, Antti Kupiainen, and R{\'e}mi Rhodes.
\newblock Compactified imaginary {L}iouville theory.
\newblock Preprint in arXiv:2310.18226, 2023.

\bibitem[GL96]{Graham-Lehrer:Cellular_algebras}
John J. Graham and Gustav I. Lehrer. 
\newblock Cellular algebras.
\newblock {\em Invent. Math.}, 123(1):1--34, 1996.

\bibitem[HK18]{Hotta-Katori:Hydrodynamic_limit_of_multiple_SLE}
Ikkei Hotta and Makoto Katori.
\newblock Hydrodynamic limit of multiple {SLE}.
\newblock {\em J. Stat. Phys.}, 171:166--188, 2018.

\bibitem[HL21]{Healey-Lawler:N_sided_radial_SLE}
Vivian~O. Healey and Gregory~F. Lawler.
\newblock N-sided radial {S}chramm–{L}oewner evolution.
\newblock {\em Probab. Theory Related Fields}, 181(1-3):451--488, 2021.

\bibitem[HPW25+]{HPW:Multiradial_SLE_with_spiral}
Chongzhi Huang, Eveliina Peltola, and Hao Wu
\newblock Multiradial SLE with spiral: resampling property and boundary perturbation. 
\newblock In preparation, 2025.

\bibitem[HS21]{Hotta-Schleissinger:Limits_of_radial_multiple_SLE_and_Burgers-Loewner_differential_equation}
Ikkei Hotta and Sebastian Schlei{\ss}inger.
\newblock Limits of radial multiple {SLE} and a {B}urgers-{L}oewner
  differential equation.
\newblock {\em J. Theoret. Probab.}, 34(2):755--783, 2021.


\bibitem[IK11]{Iohara-Koga:Representation_theory_of_Virasoro}
Kenji Iohara and Yoshiyuki Koga.
\newblock {\em Representation theory of the {V}irasoro algebra}.
\newblock Springer Monographs in Mathematics. Springer-Verlag, London, 2011.


\bibitem[Izy15]{Izyurov:Smirnovs_observable_for_free_boundary_conditions_interfaces_and_crossing_probabilities}
Konstantin Izyurov.
\newblock Smirnov's observable for free boundary conditions, interfaces and
  crossing probabilities.
\newblock {\em Comm. Math. Phys.}, 337(1):225--252, 2015.


\bibitem[Jim86]{Jimbo:q_analog_of_UqglN_Hecke_algebra_and_YBE}
Michio Jimbo.
\newblock A $q$ analog of {$\mathcal{U}(\mathfrak{gl}(N+1))$}, {H}ecke algebra,
  and the {Y}ang-{B}axter equation.
\newblock {\em Lett. Math. Phys.}, 11(3):247--252, 1986.


\bibitem[JJK16]{JJK:SLE_boundary_visits}
Niko Jokela, Matti J{\"a}rvinen, and Kalle Kyt{\"o}l{\"a}.
\newblock {$\mathrm{SLE}$} boundary visits.
\newblock {\em Ann. Henri Poincar{\'e}}, 17(6):1263--1330, 2016.


\bibitem[Kac79]{Kac:Contravariant_form_for_infinite-dimensional_Lie_algebras_and_superalgebras}
Victor~G. Kac.
\newblock Contravariant form for the infinite-dimensional {L}ie algebras and
  superalgebras.
\newblock In {\em Lecture Notes in Physics}, volume~94, pages 441--445.
  Springer-Verlag, Berlin, 1979.

\bibitem[Kac80]{Kac:Highest_weight_representations_of_infinite_dimensional_Lie_algebras}
Victor~G. Kac.
\newblock Highest weight representations of infinite dimensional {L}ie
  algebras.
\newblock In {\em Proceedings of the ICM 1978, Helsinki, Finland}, volume~1,
  pages 299--304. Acad. Sci. Fenn., 1980.

\bibitem[Kar19]{Karrila:Multiple_SLE_local_to_global}
Alex Karrila.
\newblock Multiple {$\mathrm{SLE}$} type scaling limits: from local to global.
\newblock Preprint in arXiv:1903.10354, 2019.

\bibitem[Kar20]{Karrila:UST_branches_martingales_and_multiple_SLE2}
Alex Karrila.
\newblock U{ST} branches, martingales, and multiple {$\mathrm{SLE}(2)$}.
\newblock {\em Electron. J. Probab.}, 25:1--37, 2020.


\bibitem[Kar22]{Karrila:Computation_of_pairing_probabilities_in_multiple-curve_models}
Alex Karrila.
\newblock A new computation of pairing probabilities in several multiple-curve
  models.
\newblock Preprint in arXiv:2208.06008, 2022.

\bibitem[Kau87]{Kauffman:State_models_and_the_Jones_polynomial}
Louis H. Kauffman.
\newblock State models and the {J}ones polynomial.
\newblock {\em Topology}, 26(3):395--407, 1987.

\bibitem[KL94]{Kauffman-Lins:TL_recoupling_theory_and_invariants_of_3_manifolds}
Louis H. Kauffman and Sostenes L. Lins. 
\newblock {\em Temperley-{L}ieb recoupling theory and invariants of 3-manifolds}.
\newblock Annals of Mathematics Studies, Princeton University Press, 1994.

\bibitem[KK21]{Katori-Koshida:Three_phases_of_multiple_SLE_driven_by_non-colliding_Dyson_Brownian_motions}
Makoto Katori and Shinji Koshida.
\newblock Three phases of multiple {SLE} driven by non-colliding {D}yson's
  {B}rownian motions.
\newblock {\em J. Phys. A}, 54(32, Paper No. 325002):1--19, 2021.

\bibitem[KKP19]{KKP:Conformal_blocks_q_combinatorics_and_quantum_group_symmetry}
Alex Karrila, Kalle Kyt{\"o}l{\"a}, and Eveliina Peltola.
\newblock Conformal blocks, $q$-combinatorics, and quantum group symmetry.
\newblock {\em Ann. Inst. Henri Poincar{\'e} D}, 6(3):449--487, 2019.

\bibitem[KKP20]{KKP:Boundary_correlations_in_planar_LERW_and_UST}
Alex Karrila, Kalle Kyt{\"o}l{\"a}, and Eveliina Peltola.
\newblock Boundary correlations in planar {LERW} and {UST}.
\newblock {\em Comm. Math. Phys.}, 376(3):2065--2145, 2020.


\bibitem[KL07]{Kozdron-Lawler:Configurational_measure_on_mutually_avoiding_SLEs}
Michael~J. Kozdron and Gregory~F. Lawler.
\newblock The configurational measure on mutually avoiding {$\mathrm{SLE}$}
  paths.
\newblock {\em Fields Inst. Commun.}, 50:199--224, 2007.

\bibitem[Kon87]{Kontsevich:Virasoro_and_Teichmuller_spaces}
Maxim Kontsevich.
\newblock The {V}irasoro algebra and {T}eichm{\"u}ller spaces.
\newblock {\em Funct. Anal. Appl.}, 21(2):156--157, 1987.

\bibitem[Kon03]{Kontsevich:CFT_SLE_and_phase_boundaries}
Maxim Kontsevich.
\newblock {CFT}, $\mathrm{SLE}$, and phase boundaries.
\newblock In {\em Oberwolfach Arbeitstagung}, 2003.

\bibitem[KP16]{Kytola-Peltola:Pure_partition_functions_of_multiple_SLEs}
Kalle Kyt{\"o}l{\"a} and Eveliina Peltola.
\newblock Pure partition functions of multiple $\mathrm{SLE}$s.
\newblock {\em Comm. Math. Phys.}, 346(1):237--292, 2016.


\bibitem[KP20]{Kytola-Peltola:Conformally_covariant_boundary_correlation_functions_with_quantum_group}
Kalle Kyt{\"o}l{\"a} and Eveliina Peltola.
\newblock Conformally covariant boundary correlation functions with a quantum
  group.
\newblock {\em J. Eur. Math. Soc.}, 22(1):55--118, 2020.


\bibitem[KR09]{Kytola-Ridout:On_staggered_indecomposable_Virasoro_modules}
Kalle Kyt{\"o}l{\"a} and David Ridout.
\newblock On staggered indecomposable {V}irasoro modules.
\newblock {\em J. Math. Phys}, 50(12):123503, 2009.


\bibitem[KRV19]{KRV:Local_conformal_structure_of_LQG}
Antti Kupiainen, R{\'e}mi Rhodes, and Vincent Vargas.
\newblock Local conformal structure of {L}iouville {Q}uantum {G}ravity.
\newblock {\em Comm. Math. Phys.}, 371(3):1005--1069, 2019.

\bibitem[KRV20]{KRV:Integrability_of_Liouville_theory:proof_of_the_DOZZ_formula}
Antti Kupiainen, R{\'e}mi Rhodes, and Vincent Vargas.
\newblock Integrability of Liouville theory: proof of the DOZZ formula.
\newblock {\em Ann. of Math.},  191(1):81–166, 2020.


\bibitem[KS17]{Kemppainen-Smirnov:Random_curves_scaling_limits_and_Loewner_evolutions}
Antti Kemppainen and Stanislav Smirnov.
\newblock Random curves, scaling limits and {L}oewner evolutions.
\newblock {\em Ann. Probab.}, 45(2):698--779, 2017.


\bibitem[KS18]{Kemppainen-Smirnov:Configurations_of_FK_Ising_interfaces}
Antti Kemppainen and Stanislav Smirnov.
\newblock Configurations of {FK} {I}sing interfaces and hypergeometric
  {$\mathrm{SLE}$}.
\newblock {\em Math. Res. Lett.}, 25(3):875--88, 2018.


\bibitem[KW11a]{Kenyon-Wilson:Boundary_partitions_in_trees_and_dimers}
Richard~W. Kenyon and David~B. Wilson.
\newblock Boundary partitions in trees and dimers.
\newblock {\em Trans. Amer. Math. Soc.}, 363(3):1325--1364, 2011.


\bibitem[KW11b]{Kenyon-Wilson:Double_dimer_pairings_and_skew_Young_diagrams}
Richard~W. Kenyon and David~B. Wilson.
\newblock Double-dimer pairings and skew {Y}oung diagrams.
\newblock {\em Electron. J. Combin.}, 18(1):130--142, 2011.


\bibitem[Las09]{Lascoux:Pfaffians_and_representations_of_the_symmetric_group}
Alain Lascoux.
\newblock Pfaffians and representations of the symmetric group.
\newblock {\em Acta Math. Sin. (Engl. Ser.)}, 25:1929--1950, 2009.


\bibitem[LPR25]{LPR:Fused_Specht_polynomials_and_c_equals_1_degenerate_conformal_blocks}
Augustin Lafay, Eveliina Peltola, and Julien Roussillon.
\newblock Fused {S}pecht polynomials and $c=1$ degenerate conformal blocks.
\newblock {\em Trans. Amer. Math. Soc.}, to appear, 2025. %Preprint in arXiv:2410.09798.

\bibitem[LPW25]{LPW:UST_in_topological_polygons_partition_functions_for_SLE8_and_correlations_in_logCFT}
Mingchang Liu, Eveliina Peltola, and Hao Wu.
\newblock Uniform spanning tree in topological polygons, partition functions
  for {$\mathrm{SLE}(8)$}, and correlations in {$c=-2$} logarithmic {CFT}.
\newblock {\em Ann. Probab.,}, 53(1):23--78, 2025.

\bibitem[LW21]{Liu-Wu:Scaling_limits_of_crossing_probabilities_in_metric_graph_GFF}
Mingchang Liu and Hao Wu.
\newblock Scaling limits of crossing probabilities in metric graph {GFF}.
\newblock {\em Electron. J. Probab.}, 26(37):1--46, 2021.


\bibitem[MW18]{Miller-Werner:Connection_probabilities_for_conformal_loop_ensembles}
Jason Miller and Wendelin Werner.
\newblock Connection probabilities for conformal loop ensembles.
\newblock {\em Comm. Math. Phys.}, 362(2):415--453, 2018.


\bibitem[Pel19]{Peltola:Towards_CFT_for_SLEs}
Eveliina Peltola.
\newblock Towards a conformal field theory for {S}chramm-{L}oewner evolutions.
\newblock {\em J. Math. Phys.}, 60(10):103305, 2019.
\newblock Special issue (Proc. ICMP, Montreal, July 2018).


\bibitem[Pel20]{Peltola:Basis_for_solutions_of_BSA_PDEs_with_particular_asymptotic_properties}
Eveliina Peltola.
\newblock Basis for solutions of the {B}enoit \& {S}aint-{A}ubin {PDE}s with
  particular asymptotic properties.
\newblock {\em Ann. Inst. Henri Poincar{\'e} D}, 7(1):1--73, 2020.


\bibitem[Pol70]{Polyakov:Conformal_symmetry_of_critical_fluctuations}
Alexander~M. Polyakov.
\newblock Conformal symmetry of critical fluctuations.
\newblock {\em JETP Lett.}, 12(12):381--383, 1970.

\bibitem[PP77]{Patashinskii-Pokrovskii:The_RG_method_in_the_theory_of_phase_transitions}
A. Z. Patashinskii and V. L. Pokrovskii.
\newblock The renormalization-group method in the theory of phase transitions. 
\newblock {\em Sov. Phys. Usp.}, 20(1):31--54, 1977.

\bibitem[PW19]{Peltola-Wu:Global_and_local_multiple_SLEs_and_connection_probabilities_for_level_lines_of_GFF}
Eveliina Peltola and Hao Wu.
\newblock Global and local multiple $\mathrm{SLE}$s for $\kappa \leq 4$ and
  connection probabilities for level lines of {GFF}.
\newblock {\em Comm. Math. Phys.}, 366(2):469--536, 2019.


\bibitem[PW23]{Peltola-Wu:Crossing_probabilities_of_multiple_Ising_interfaces}
Eveliina Peltola and Hao Wu.
\newblock Crossing probabilities of multiple {I}sing interfaces.
\newblock {\em Ann. Appl. Probab.}, 33(4):3169--3206, 2023.


\bibitem[Rib14]{Ribault:CFT_in_the_plane}
Sylvain Ribault.
\newblock Conformal field theory on the plane.
\newblock Preprint in arXiv:1406.4290, 2014.

\bibitem[RS15]{Ribault-Santachiara:Liouville_theory_with_a_central_charge_less_than_one}
Sylvain Ribault and Raoul Santachiara.
\newblock Liouville theory with a central charge less than one.
\newblock {\em J. High Energy Phys.}, 8(109):1--25, 2015.

\bibitem[Sch06]{Schramm:ICM}
Oded Schramm.
\newblock Conformally invariant scaling limits, an overview and a collection of
  problems.
\newblock In {\em Proceedings of the ICM 2006, Madrid, Spain}, volume~1, pages
  513--543. European Mathematical Society, 2006.

\bibitem[Sch08]{Schottenloher:Mathematical_introduction_to_CFT}
Martin Schottenloher.
\newblock {\em A mathematical introduction to conformal field theory}, volume
  759 of {\em Lecture Notes in Physics}.
\newblock Springer-Verlag, Berlin Heidelberg, 2nd edition, 2008.


\bibitem[Smi06]{Smirnov:Towards_conformal_invariance_of_2D_lattice_models}
Stanislav Smirnov.
\newblock Towards conformal invariance of {$2D$} lattice models.
\newblock In {\em Proceedings of the ICM 2006, Madrid, Spain}, volume~2, pages
  1421--1451. European Mathematical Society, 2006.

\bibitem[Spe35]{Specht:Die_irreduziblen_Darstellungen_der_symmetrischen_Gruppe}
Wilhelm Specht.
\newblock Die irreduziblen {D}arstellungen der symmetrischen {G}ruppe.
\newblock {\em Math. Z.}, 39:696--711, 1935.


\bibitem[SS11]{Schramm-Smirnov:Scaling_limits_of_planar_percolation}
Oded Schramm and Stanislav Smirnov.
\newblock On the scaling limits of planar percolation.
\newblock {\em Ann. Probab.}, 39(5):1768--1814, 2011.
\newblock With an appendix by {C}hristophe {G}arban.


\bibitem[SW11]{Sheffield-Wilson:Schramms_proof_of_Watts_formula}
Scott Sheffield and David~B. Wilson.
\newblock Schramm's proof of {W}atts' formula.
\newblock {\em Ann. Probab.}, 39(5):1844--1863, 2011.


\bibitem[Wat96]{Watts:A_crossing_probability_for_critical_percolation_in_two_dimensions}
Gerard Watts.
\newblock A crossing probability for critical percolation in two dimensions.
\newblock {\em J. Phys. A}, 29:L363, 1996.


\bibitem[Wen87]{Wenzl:On_sequences_of_projections}
Hans Wenzl.
\newblock On sequences of projections.
\newblock {\em C. R. Math. Rep. Acad. Sci. Canada}, 9(1):5--9, 1987.


\bibitem[Wil96]{Wilson:Generating_random_spanning_trees_more_quickly_than_cover_time}
David~B. Wilson.
\newblock Generating random spanning trees more quickly than the cover time.
\newblock In {\em Proceedings of the {T}wenty-eighth {A}nnual {ACM} {S}ymposium
  on the {T}heory of {C}omputing ({P}hiladelphia, {PA}, 1996)}, pages 296--303,
  New York, 1996. ACM.

\bibitem[Zhu20]{Zhu:Higher_order_BPZ_equations_for_Liouville_CFT}
Tunan Zhu.
\newblock Higher order {BPZ} equations for {L}iouville conformal field theory.
\newblock Preprint in arXiv:2001.08476, 2020.

\end{thebibliography}
